\documentclass[11pt]{iopart}

\usepackage[style=numeric,natbib=true,backend=bibtex]{biblatex}

\usepackage[latin1]{inputenc}
\usepackage{tikz,tkz-tab}
\usepackage{variations}

\usepackage{mystyle}
\usepackage{notations}
\graphicspath{{./images/}}

\bibliography{bibliography.bib}
\AtEveryBibitem{\clearlist{url}}
\AtEveryBibitem{\clearlist{doi}}

\begin{document}

% \title[Discrete methods for deconvolution: an analysis on thin grids]{Discrete methods for deconvolution: \\ an analysis on thin grids}
\title{Sparse Spikes Deconvolution on Thin Grids}

\author{Vincent Duval$^1$, Gabriel Peyr{\'e}$^2$}
\address{$^1$ INRIA, MOKAPLAN}
\ead{vincent.duval@inria.fr}
\address{$^2$ CNRS and CEREMADE, Universit\'e Paris Dauphine}
\ead{gabriel.peyre@ceremade.dauphine.fr}

%\author{
%	\begin{tabular}{c}
%		Vincent Duval \\
%		INRIA \\
%		MOKAPLAN \\
%		{\small  \url{vincent.duval@inria.fr} }
%	\end{tabular}
%	\begin{tabular}{c}
%		Gabriel Peyr\'e \\
%		CNRS and Ceremade \\
%		Univ. Paris-Dauphine \\
%		{\small  \url{gabriel.peyre@ceremade.dauphine.fr} }
%	\end{tabular}
%}

% \date{\today}
% \maketitle

%%%%%

\begin{abstract}
% !TEX root = ../continuous-bp.tex
This article analyzes the recovery performance of two popular finite dimensional approximations of the sparse spikes deconvolution problem over Radon measures. 
We examine in a unified framework both the $\ell^1$ regularization (often referred to as \lasso or Basis-Pursuit) and the Continuous Basis-Pursuit (C-BP) methods. The \lasso is the de-facto standard for the sparse regularization of inverse problems in imaging. It performs a nearest neighbor interpolation of the spikes locations on the sampling grid. The C-BP method,  introduced by Ekanadham, Tranchina and Simoncelli, uses a linear interpolation of the locations to perform a better approximation of the infinite-dimensional optimization problem, for positive measures. 
We show that, in the small noise regime, both methods estimate twice the number of spikes as the number of original spikes. Indeed, we show that they both detect two neighboring spikes around the locations of an original spikes.
These results for deconvolution problems are based on an abstract analysis of the so-called extended support of the solutions of $\ell^1$-type problems (including as special cases the \lasso and C-BP for deconvolution), which are of an independent interest. They precisely characterize the support of the solutions when the noise is small and the regularization parameter is selected accordingly. We illustrate these findings to analyze for the first time the support instability of compressed sensing recovery when the number of measurements is below the critical limit (well documented in the literature) where the support is provably stable.
% {\color{red}in-between grid points. However, the precision of the two methods are differents. For a grid size $\De$, D-BP estimate locations and amplitude with an error of order $\De$, while C-BP reach an error of order $\De^3$. For this results to hold, the amplitude of the spikes needs to be large enough and the location of the spikes needs to be not too far from the grid points.} 

\end{abstract}

% \tableofcontents

% !TEX root = ../continuous-bp.tex
%%%%%%%%%%%%%%%%%%%%%%%%%%%%%%%%%%%%%%%%%%%%%%%%%%

%%%%%%%%%%%%%%%%%%%%%%%%%%%%%%%%%%%%%%%%%%%%%%%%%%
\section{Introduction}

We consider the problem of estimating an unknown Radon measure $m_0 \in \Mm(\TT)$ from low-resolution noisy observations
\eql{\label{eq-observ-measures}
	y=\Phi(m_0)+\noise \in \Ldeux(\TT)
} 
where $\noise \in \Ldeux(\TT)$ is some measurement noise, and $\Phi : \Mm(\TT) \rightarrow \Ldeux(\TT)$ is an integral transform with smooth kernel $\phi \in \Cont^2(\TT \times \TT)$, i.e.
\eql{\label{eq-def-phi}
	\foralls x \in \TT, \quad
	(\Phi m)(x) = \int_{\TT} \phi(x,y) \d m(y).
}
A typical example of such an operation is the convolution, where $\varphi(x,y)=\tilde{\varphi}(x-y)$ for some smooth function $\tilde{\varphi}$ defined on the torus $\TT=\RR/\ZZ$ (i.e. an interval with periodic boundary conditions).
We focus our attention here for simplicity on the compact 1-D domain $\TT$, but the algorithms considered (\lasso and C-BP) as well as our theoretical analysis can be extended to higher dimensional settings (see Section~\ref{sec-extensions}).

%%%%%%%%
\subsection{Sparse Regularization}

The problem of inverting~\eqref{eq-observ-measures} is severely ill-posed. A particular example is when $\Phi$ is a low pass filter, which is a typical setting for many problems in imaging. In several applications, it makes sense to impose some sparsity assumption on the data to recover. This idea has been introduced first in the geoseismic literature, to model the layered structure of the underground using sparse sums of Dirac masses~\cite{Claerbout-geophysics}. Sparse regularization has later been studied by David Donoho and co-workers, see for instance~\cite{Donoho-superresol-sparse}.  

In order to recover sparse measures (i.e. sums of Diracs), it makes sense to consider the following regularization
\eql{\label{eq-blasso}
	\umin{m \in \Mm(\TT)} \frac{1}{2}\norm{y-\Phi(m)}^2 + \la |m|(\TT)
}
where $|m|(\TT)$ is the total variation of the measure $m$, defined as
\eql{\label{eq-totalvariation}
	|m|(\TT) \eqdef \sup \enscond{ \int_\TT \psi(x) \d m(x) }{ \psi \in \Cont(\TT), \: \normi{\psi} \leq 1 }.
} 
This formulation of the recovery of sparse Radon measures has recently received lots of attention in the literature, see for instance the works of~\cite{Bredies-space-measures,deCastro-beurling,Candes-toward}.

%%%%%%%%
\subsection{\lasso}

The optimization problem~\eqref{eq-blasso} is convex but infinite dimensional, and while there exists solvers when $\Phi$ is measuring a finite number of Fourier frequency (see~\cite{Candes-toward}), they do not scale well with the number of frequencies. Furthermore, the case of an arbitrary linear operator $\Phi$ is still difficult to handle, see~\cite{Bredies-space-measures} for an iterative scheme. The vast majority of practitioners thus approximate~\eqref{eq-blasso} by a finite dimensional problem computed over a finite grid $\Gg \eqdef \enscond{z_i}{i\in \seg{0}{P-1}} \subset \TT$, by restricting their attention to measures of the form
\eq{
  m_{a,\Gg} \eqdef \sum_{i=0}^{P-1} a_i \de_{z_i} \in \Mm(\TT).
}
For such a discrete measure, one has $|m|(\TT) =\sum_{i=0}^{|\Gg|-1} |a_i|  =\norm{a}_1$, which can be interpreted as the fact that $|\cdot|(\TT)$ is the natural extension of the $\ell^1$ norm from finite dimensional vectors to the infinite dimensional space of measures. 
Inserting this parametrization in~\eqref{eq-blasso} leads to the celebrated Basis-Pursuit problem~\cite{chen1999atomi}, which is also known as the \lasso method in statistics~\cite{tibshirani1996regre}, 
\eql{\label{eq-bp}
	\umin{a \in \RR^N} \frac{1}{2}\norm{y-\Phi_\Gg a}^2 + \la \norm{a}_1
}
where in the following we make use of the notations
\begin{align}\label{eq-phix}
  \Phi_\Gg a \eqdef \Phi(m_{a,\Gg}) = \sum_{i=0}^{\taillegrid-1} a_i \phi(\cdot,z_i), \\
  \Phi'_\Gg b \eqdef \Phi'(m_{b,\Gg}) = \sum_{i=0}^{\taillegrid-1} b_i \partial_2 \phi(\cdot,z_i), 
\end{align}
and $\partial_2 \phi$ is the derivative with respect to the second variable. 
One can understand~\eqref{eq-bp} as performing a nearest neighbor interpolation of the Dirac's locations. 

This approximation is however quite crude, and we recently showed in~\cite{2013-duval-sparsespikes} that it leads to imperfect estimation of both the number of spikes and their locations.  Indeed, this problem typically recovers up to twice as many spikes as the input measures, because spikes of $m_0$ gets duplicated as the two nearest neighbors on the grid $\Gg$.%  These locations are furthermore restricted to lie on a grid, thus achieving a precision of order $O(\stepsize)$ where $\stepsize$ is the grid spacing.

Note that while we focus in this paper on convex recovery method, and in particular $\ell^1$-type regularization, there is a vast literature on the subject, which makes use of alternative algorithms, see for instance~\cite{Odendaal-music,Blu-fri} and the references therein.

%%%%
\subsection{Continuous Basis-Pursuit (C-BP)}

To obtain a better approximation of the infinite dimensional problem, \cite{Ekanadham-CBP} proposes to perform a first order approximation of the kernel. This method assumes that the unknown measure is positive. 
To ease the exposition, we consider a uniform grid $\Gg \eqdef \enscond{i/\taillegrid}{i\in\seg{0}{\taillegrid-1}}$ of $\taillegrid$ points, so that the grid size is $\stepsize \eqdef 1/\taillegrid$. 
The C-BP method of~\cite{Ekanadham-CBP} solves
\eql{\label{eq-c-bp-intro}
	\umin{(a,b) \in \RR^\taillegrid \times \RR^\taillegrid}
		\frac{1}{2} \norm{ y - \Phi_\Gg a - \Phi'_\Gg b }^2 + \la \norm{a}_1
		\quad\text{subject to}\quad
		|b| \leq \frac{\stepsize}{2} a, 
}
where the inequality should be understood component-wise. Note also that the obtained $a$ is always nonnegative, hence the C-BP method is tailored for the recovery of positive measures.
This is a convex optimization problem, which can be solved using traditional conic optimization methods.
As detailed in Section~\ref{sec-cbp-another-param}, this problem can also be re-cast as a \lasso in dimension $2P$ with positivity constraints (see Section~\ref{sec-cbp-another-param}). Hence it can be solved using a large variety of first order proximal method, the most simple one being the Forward-Backward, see~\cite{BauschkeCombettes11} and the references therein.

 % We detail in Section~\ref{sec-prox-algo} a simple proximal splitting scheme which can be used. 

If $(a^\star,b^\star)$ are solutions of~\eqref{eq-c-bp-intro}, one recovers an output discrete measure defined by
\eql{\label{eq-defn-recov-measure}
	m^\star = \sum_{a^\star_i \neq 0} a^\star_i \de_{x_i^\star}
	\qwhereq
		x_i^\star \eqdef i\stepsize + \frac{b_i^\star}{a_i^\star}, 
%		\tau_i^\star = \frac{ b_i^\star}{a_i^\star} .
}
where we set $\frac{b_i^\star}{a_i^\star}=0$ whenever $a_i^\star=0$.
The rationale behind~\eqref{eq-c-bp-intro} is to perform a first order Taylor approximation of the operator $\Phi$, where the variable $\tau_i \eqdef  b_i / a_i \in [-h/2,h/2]$ encodes the horizontal shift of the Dirac location with respect to the grid sample $i\stepsize$. The landmark idea introduced in~\cite{Ekanadham-CBP} is that, while the optimization is non-convex with respect to the pair $(a,\tau)$, it is convex with respect to the pair $(a,b)$. 

% {\color{red}In this article, we show that under a signal-to-noise ratio on both the location and amplitude of the spikes, C-BP correctly estimates the number of spikes, and that the precision on the their location is of order $O(\stepsize^3)$, thus outperforming traditional BP.}

%%%%
\subsection{Extensions}
\label{sec-extensions}

While we restrict here the exposition to 1-D problems, the C-BP formulation~\eqref{eq-c-bp-intro} can be extended to cope with measures in arbitrary dimension $d \geq 1$, i.e. to consider $m_0 \in \Mm(\TT^d)$. This requires to define at each sampling grid point indexed by $i$ a vector $b_i = (b_{i,k})_{k=1}^d \in \RR^d$ together with the constraint $\normi{b_i} \leq \frac{h}{2} a_i$, and also to use a matrix $\Phi_\Gg'$ defined as 
\eq{
  \Phi_\Gg' b \eqdef \sum_{i\in\Gg} \sum_{k=1}^d b_{i,k} \partial_{k}  \phi(\cdot,x_i) \in \Ldeux(\TT^d)
} 
where $\partial^{k}$ denote the differential operator with respect to the $k^{\text{th}}$ direction in $\RR^d$. Our analysis carries over to this setting without major difficulties. 

The paper~\cite{Ekanadham-CBP} also proposes other interpolation schemes than a first order Taylor expansion at the grid points. In particular, they develop a ``polar'' interpolation which makes use of two adjacent grid points. This method seems to outperform the linear interpolation in practice, and has been employed to perform spikes sorting in neuronal recordings~\cite{Ekanadham-Neuro}. 

Extending the results we propose in the present paper to these higher dimensional settings and alternative interpolation schemes is an interesting avenue for future work.  

Let us also mention that an important problem is to extend the C-BP method~\eqref{eq-defn-recov-measure} to measures with arbitrary signs and that can even be complex-valued. Unfortunately, the corresponding constraint $|b| \leq |a|$ is then non-convex, which makes the mathematical analysis apparently much more involved. A non-convex and non-smooth optimization solver is proposed for this problem in~\cite{FLORESCU-CBP-NonCvx}, and shows promising practical performance for spectrum estimation.

%%%%
\subsection{Previous Works}

Most of the early work to assess the performance of convex sparse regularization has focussed its attention on the finite dimensional case, thus considering only the \lasso problem~\eqref{eq-bp}. While the literature on this subject is enormous, only very few works actually deal with deterministic and highly correlated linear operators such as low-pass convolution kernels. The initial works of Donoho~\cite{Donoho-superresol-sparse} study the Lipschitz behavior of the inverse map $y \mapsto a^\star$, where $a^\star$ is a solution of~\eqref{eq-bp}, as a function of the bandwidth of the bandpass filter. The first work to address the question of spikes identification (i.e. recovery of the exact location of the spikes over a discrete grid) is~\cite{DossalMallat}. This work uses the analysis of $\ell^1$ regularization introduced by Fuchs in~\cite{fuchs2004on-sp}. This type of analysis ensures that the support of the input measure is stable under small noise perturbation of the measurements. Our finding is that this is however never the case (support is always unstable) when the grid is thin enough, and we thus introduce the notion of ``extended support'', which is in some sense the smallest extension of the support which is stable. The idea of extending the support to study the recovery performance of $\ell^1$ methods can be found in the work of Dossal~\cite{dossal2011necessary} who focusses on noiseless recovery and stability in term of $\ell^2$ error. 

Recently, a few works have studied the theoretical properties of the recovery over measures~\eqref{eq-blasso}. Cand\`es and Fernandez-Granda show in~\cite{Candes-toward} that this convex program does recover exactly the initial sparse measure when $w=0$ and $\la \rightarrow 0$, if the spikes are well-separated. The robustness to noisy measurements is analyzed by the same authors in~\cite{Candes-superresol-noisy} using an Hilbertian norm, and in~\cite{Fernandez-Granda-support,Azais-inaccurate} in terms of spikes localization. The work of~\cite{Tang-linea-spectral} analyzes the reconstruction error. Lastly, \cite{2013-duval-sparsespikes} provides a condition ensuring that~\eqref{eq-blasso} recovers the same number of spikes as the input measure and that the error in terms of spikes localization and elevation has the same order as the noise level. 

Very few works have tried to bridge the gap between these grid-free methods over the space of measures, and finite dimensional discrete approximations that are used by practitioners. The convergence (in the sense of measures) of the solutions of the discrete problem toward to ones of the grid-free problem is shown in~\cite{TangConvergence}, where a speed of convergence is shown using tools from semi-infinite programming~\cite{StillDiscretizationSDP}. The same authors show in~\cite{Bhaskar-line-spectral} that the discretized problem achieves a similar prediction $L^2$ error as the grid-free method. In~\cite{2013-duval-sparsespikes}, we have shown that solutions of the discrete \lasso problem estimate in general as much as twice the number of spikes as the input measure. We detail in the following section how the present work gives a much more precise and general analysis of this phenomenon.

%%%%
\subsection{Contributions}

Our first contribution is an improvement over the known analysis of the \lasso in an abstract setting (that is~\eqref{eq-bp} when $\Phi_\Gg$ is replaced with any linear operator $\RR^\taillegrid\rightarrow L^2(\RR)$). Whereas Fuchs' result~\cite{fuchs2004on-sp} characterizes the support recovery of the \lasso at low noise, our previous work~\cite{2013-duval-sparsespikes} has pointed out that when Fuchs'criterion is not satisfied, the nonzero components of the solutions of the Basis-Pursuit at low noise are contained in the \textit{extended support}, that is the saturation set of some \textit{minimal norm dual certificate}. In this work, we provide a characterization of this minimal norm certificate, and we give a sufficient condition which holds generically and which ensures that all the components of the extended support are actually nonzero (with a prediction on the signs). Our main result in this direction is Theorem~\ref{thm-stability-dbp}.

Our second contribution applies this result to Problem~\eqref{eq-bp} on thin grids. After recalling the convergence properties of Problem~\eqref{eq-bp} towards~\eqref{eq-blasso}, we show that under some assumption, if the input measure $m_0 = m_{\al_0,x_0}= \sum_{\nu=1}^N \alpha_{0,\nu}\delta_{x_{0,\nu}}$ has support on the grid (\ie $x_{0,\nu}\in\Gg$ for all $\nu$), the model  at low noise actually reconstructs pairs of Dirac masses, i.e. solutions of the form
\begin{align}
  m_\la=\sum_{\nu=1}^N \left(\alpha_{\la,\nu}\delta_{x_{0,\nu}}+ \beta_{\la,\nu}\delta_{x_{0,\nu}+\varepsilon_\nu\stepsize} \right),\qwhereq & \varepsilon_\nu\in\{-1,+1\},\\
\qandq  \sign(\alpha_{\la,\nu})=\sign(\beta_{\la,\nu})&=\sign(\alpha_{0,\nu}).
\end{align}
The precise statement of this result can be found in Theorem~\ref{thm-lasso-extended}.
Compared to~\cite{2013-duval-sparsespikes} where it is predicted that spikes could appear at most in pairs, this result states that all the pairs do appear, and it provides a closed-form expression for the shift $\varepsilon$. That closed-form expression does not vary as the grid is refined, so that the side on which each neighboring spike appears is in fact intrinsic to the measure, we call it the \textit{natural shift}. Moreover, we characterize the low noise regime as $\frac{\norm{w}_2}{\la}=O(1)$ and $\la = O(\stepsize)$.

Then, we turn to the Continuous Basis-Pursuit~\eqref{eq-c-bp-intro}. We first study this problem in an abstract setting, where it is reformulated as a \lasso with positivity constraints. We derive similar noise robustness properties as for the \lasso, and we characterize the extended support, see Theorem~\ref{thm-abstract-cbp}.  Working on a thin grid, we show the $\Gamma$-convergence of Problem~\eqref{eq-c-bp-intro} towards the Beurling \lasso~\eqref{eq-blasso} with positivity constraints and we give a fine analysis of the support of the solutions as the grid stepsize tends to zero. We also study the low noise behavior when the measure has support on the grid: under a suitable assumption, the recovered spikes appear again in pairs,
\eq{
  	m_\la = \sum_{\nu=1}^N \left(
		\alpha_{\la,\nu}\delta_{x_{0,\nu}+t_\nu}+ \beta_{\la,\nu}\delta_{x_{0,\nu}+\varepsilon_\nu\stepsize/2} 
		\right)
		\qwhereq 
		\choice{ 
			\varepsilon_\nu\in\{-1,+1\},\\
			-\stepsize/2<t_\nu<\stepsize/2,
		}
}
see Theorem~\ref{thm-cbpasso-extended}.
A closed form expression for $\varepsilon$ is given, which depends on some corresponding \textit{natural shift} intrinsic to the measure (which differs from the one of the \lasso). The corresponding low noise regime is characterized by $\frac{\norm{w}_2}{\la}=O(1)$ and $\la = O(\stepsize^3)$.

It is important to realize that, in this setting of convolution on thin grids, our contributions give important information about the structure of the recovered spikes when the noise $\noise$ is small. This is especially important since, on contrary to common belief, the spikes locations for \lasso and C-BP are not stable: even for an arbitrary small noise $\noise$, neither methods retrieve the correct input spikes locations.

Eventually, we illustrate in Section~\ref{sec-numerics} these theoretical results with numerical experiments. We first display the evolution of the solution path $\la \mapsto a_\la$ (a solution of~\eqref{eq-bp}) and $\la \mapsto (a_\la,b_\la)$ (a solution of~\eqref{eq-c-bp-intro}). These paths are piecewise-affine, and our contributions (Theorems~\ref{thm-lasso-extended} and \ref{thm-cbpasso-extended}) precisely characterize the first affine segment of these paths, which perfectly matches the numerical observations. We then illustrate our abstract analysis of the \lasso problem~\eqref{eq-bp} (as provided by Theorem~\ref{thm-stability-dbp}) to characterize numerically the behavior of the \lasso for compressed sensing (CS) recovery (i.e. when one replaces the filtering $\Phi_\Gg$ appearing in~\eqref{eq-bp} with a random matrix). The literature on CS only describes the regime where enough measurements are available so that the support is stable, or does not study support stability but rather $\ell^2$ stability. Theorem~\ref{thm-stability-dbp} allows us to characterize numerically how much the support becomes unstable (in the sense that the extended support's size increases) as the number of measurements decreases (or equivalently the sparsity increases).

%%%%%%%%%%%%%%%%%%%%%%%%%%%%%%%%%%%%%%%%%%%%%%%%%%%%%
\subsection{Notations and preliminaries}

The set of Radon measures (resp. positive Radon measures) is denoted by $\Mm(\TT)$ (resp. $\Mm^+(\TT)$). Endowed with the total variation norm~\eqref{eq-totalvariation}, $\Mm(\TT)$ is a Banach space. Another useful topology on $\Mm(\TT)$ is the weak* topology: a sequence of measures $(m_n)_{n\in\NN}$ weak* converges towards $m\in \Mm(\TT)$ if and only if for all $\psi\in C(\TT)$, $\lim_{n\to+\infty}\int_\TT \psi \d m_n= \int_\TT \psi \d m$.
Any bounded subset of $\Mm(\TT)$ (for the total variation) is relatively sequentially compact for the weak* topology. Moreover the topology induced by the total variation is stronger than the weak* topology, and the total variation is sequentially lower semi-continuous for the weak* topology. 
Throughout the paper, given $\alpha\in\RR^N$ and $x_0\in\TT^N$, the notation $m_{\alpha,x_0} \eqdef \sum_{\nu=1}^N \alpha_\nu \delta_{x_{0,\nu}}$ hints that $\alpha_\nu\neq 0$ for all $\nu$  (contrary to the notation $m_{a,\Gg}$), and that the $x_{0,\nu}$'s are pairwise distinct.

The properties of $\Phi:\Mm(\TT)\rightarrow L^2(\TT)$ and its adjoint are recalled in Proposition~\ref{lem-phi-compact} in Appendix. The $\infty,2$-operator norm of $\Phi^*:L^2(\TT)\rightarrow  \Cont(\TT)$ is defined as $\norm{\Phi^*}_{\infty,2}\eqdef \sup\enscond{\norm{\Phi^*w}_\infty}{w\in L^2(\TT),  \norm{w}_{L^2}\leq 1}$ (and the $\infty,2$ operator norm of a matrix is defined similarly).
Given a vector $x_0\in\TT^N$, $\Phi_{x_0}$ refers to the linear operator $\RR^N\rightarrow L^2(\TT)$, with 
\begin{align*}
\forall \alpha\in\RR^N,\quad   \Phi_{x_0} \alpha \eqdef \Phi(m_{\alpha,x_0}) = \sum_{\nu=1}^{N} \alpha_\nu \phi(\cdot,x_{0,\nu}).
\end{align*}
It may also be seen as the restriction of $\Phi$ to measures supported on the set $\enscond{x_{0,\nu}}{\nu\in\seg{1}{N}}$. A similar notation is adopted for $\Phi'_{x_0}$ (replacing $\phi(\cdot,x_{0,\nu})$ with $\partial_2\varphi(\cdot,,x_{0,\nu})$. The concatenation of $\Phi_{x_0}$ and $\Phi'_{x_0}$ is denoted by $\Gamma_{x_0}\eqdef \begin{pmatrix}
   \Phi_{x_0} & \Phi'_{x_0}
\end{pmatrix}$.

We shall rely on the notion of set convergence. Given a sequence $(C_n)_{n\in\NN}$ of subsets of $\TT$, we define 
\begin{align}
  \limsup_{n\to+\infty} C_n= \enscond{x\in\TT}{\liminf_{n\to+\infty} d(x,C_n)=0}\\
  \liminf_{n\to+\infty} C_n= \enscond{x\in\TT}{\limsup_{n\to+\infty} d(x,C_n)=0}
\end{align}
where $d$ is defined by $d(x,C)=\inf_{x'\in C}|x'-x|$ and $|x-x'|$ refers to the distance between $x$ and $x'$ on the torus. 
If both sets are equal, let $C$ be the corresponding set (then $C$ is necessarily closed), we write
\begin{align}
  \lim_{n\to+\infty} C_n= C.
\end{align}
If the sequence $(C_n)_{n\in\NN}$ is nondecreasing ($C_n\subset C_{n+1}$), then $\lim_{n\to\infty}C_n= \overline{\bigcup_{n\in\NN} C_n}$, and if it is nonincreasing  ($C_n\supset C_{n+1}$) then $\lim_{n\to\infty}C_n= \bigcap_{n\in\NN} \overline{C_n}$ (where $\overline{C}$ denotes the closure of $C$).
We refer the reader to~\cite{rockafellarwets} for more detail about set convergence. We shall also use this notion in Hilbert spaces, with obvious adaptations.

% !TEX root = ../continuous-bp.tex
%%%
\section{Abstract analysis of the \lasso}
\label{sec-discrete-lasso}

The aim of this section is to study the low noise regime of the \lasso problem in an abstract finite dimensional setting, regardless of the grid stepsize. In this framework, the columns of the (finite dimensional) degradation operator need not be the samples of a continuous (\textit{e.g.} convolution) operator, and the provided analysis holds  for any general \lasso problem. We extend the initial study of Fuchs of the basis pursuit method (see~\cite{fuchs2004on-sp}) which gives the analytical expression of the solution when the noise is low and the support is stable. Here, provided we have access to a particular dual vector $\certdiscO$, we give an explicit parametrization the solutions of the basis pursuit at low noise \textit{even when the support is not stable}. This is especially relevant for the deconvolution problem since the support is not stable when the grid is thin enough.

%%%
\subsection{Notations and optimality conditions}

We consider in this section observations in an arbitrary Hilbert space $\Hh$, which might be for instance $L^2(\TT)$ (as in the previous section) or a finite dimensional vector space. 
The linear degradation operator is then denoted as $\Op:\RR^\taillegrid \rightarrow \Hh$.
Let us emphasize that in this section, for $\vecdisc\in\RR^\taillegrid$,  $\|\vecdisc\|_\infty \eqdef \max_{0\leq k\leq \taillegrid-1} |\vecdisc_k|$. 

Given an observation $y_0=\Op \vecdiscO\in \Hh$ (or $y=y_0+w$, where $w\in \Hh$), we aim at reconstructing the vector $\vecdiscO\in \RR^\taillegrid$ by solving the \lasso problem for $\la>0$,
\begin{align}
  \umin{\vecdisc \in \RR^\taillegrid} \frac{1}{2}\norm{y-\Op \vecdisc}^2 + \la \norm{\vecdisc}_1 \tag{$\Pp_\la(y)$}\label{eq-abstract-lasso}
\end{align}
 and for $\la =0$ we consider the (Basis-Pursuit) problem 
\begin{align}
  \umin{\vecdisc\in \RR^\taillegrid} \norm{\vecdisc}_1 \mbox{ such that } \Op \vecdisc=y_0.  \tag{$\Pp_0(y_0)$}\label{eq-abstract-bp}
\end{align}

If $\vecdisc\in \RR^\taillegrid$, we denote by $I(\vecdisc)$, or $I$ when the context is clear, the support of $\vecdisc$, \textit{i.e.} $I(\vecdisc) \eqdef \enscond{i\in \seg{0}{\taillegrid-1}}{\vecdisc_i\neq 0}$. Also, we let $s_I\eqdef\sign(\vecdisc_I)$, and $\suppm(\vecdisc)\eqdef\enscond{(i,s_i)}{i\in I}$ the \textit{signed support} of $\vecdisc$.

The optimality conditions for Problems~\eqref{eq-abstract-lasso} and~\eqref{eq-abstract-bp} are quite standard, as detailed in the following proposition.

\begin{prop}
  Let $y\in \Hh$, and $\vecdisc_\la \in \RR^\taillegrid$. Then $a_\la$ is a solution to \eqref{eq-abstract-lasso} if and only if there exists $p_\la\in \Hh$ such that
  \begin{align}
    \|\Op^*p_\la \|_{\infty} \leq 1,\qandq (\Op^*p_\la)_I=\sign(\vecdisc_{\la,I}),\label{eq-optimal-subdiff}\\
    \la  \Op^*p_\la + \Op^*(\Op \vecdisc_\la-y) =0.\label{eq-optimal-lagrange}
  \end{align}

  Similarly, if $\vecdiscO\in \RR^\taillegrid$, then $\vecdiscO$ is a solution to~\eqref{eq-abstract-bp} if and only if $\Op \vecdiscO=y_0$ and there exists $p\in \Hh$ such that.
\begin{align}
  \|\Op^*p\|_\infty \leq 1 \qandq (\Op^* p)_{I}=\sign(\vecdisc_{0,I}).\label{eq-optimal-source}
\end{align}
\end{prop}

Conditions~\eqref{eq-optimal-subdiff} and~\eqref{eq-optimal-source} merely express the fact that $\certdisc_\la \eqdef \Op^*p_\la$ (resp. $\eta \eqdef \Op^*p$) is in the subdifferential of the $\ell^1$-norm at $\vecdisc_\la$ (resp. $\vecdiscO$). In that case we say that $\certdisc_\la$ (resp. $\certdisc$) is a \textit{dual certificate} for $\vecdisc_\la$ (resp. $\vecdiscO$). Condition~\eqref{eq-optimal-source} is also called the \textit{source condition} in the literature~\cite{Burger_Osher04}.

The term \textit{dual certificate} stems from the fact that $p_\la$ (resp. $p$) is a solution to the dual problem to~\eqref{eq-abstract-lasso} (resp.~\eqref{eq-abstract-bp}), 
\begin{align}
  \inf_{p\in C} &\left\|\frac{y}{\la}-p  \right\|_2^2 \tag{$\Dd_\la(y)$},\label{eq-abstract-dual-lasso}\\
  \mbox{resp. }\quad  \sup_{p\in C} &\langle y_0, p\rangle, \tag{$\Dd_0(y_0)$}\label{eq-abstract-dual-bp}\\
  \qwhereq C& \eqdef \enscond{p\in \Hh}{\max_{k\in\seg{0}{\taillegrid-1}}|(\Op^*p)_k| \leq 1}.
\end{align}
If $\vecdisc$ is a solution to~\eqref{eq-abstract-lasso} and $p_\la$ is a solution to~\eqref{eq-abstract-dual-lasso}, then~\eqref{eq-optimal-subdiff} and~\eqref{eq-optimal-lagrange} hold. Conversely, for any $\vecdisc\in\RR^P$ and any $p_\la\in \Hh$, if~\eqref{eq-optimal-subdiff} and~\eqref{eq-optimal-lagrange} hold, then  $\vecdisc$ is a solution to~\eqref{eq-abstract-lasso} and $p_\la$ is a solution to~\eqref{eq-abstract-dual-lasso}. A similar equivalence hold for~\eqref{eq-abstract-bp} and~\eqref{eq-abstract-dual-bp}.

\begin{rem}
In general, the solutions to~\eqref{eq-abstract-lasso} and~\eqref{eq-abstract-bp} need not be unique. However, the dual certificate $\certdisc_\la=\Op^*p_\la$ which appears in~\eqref{eq-optimal-subdiff} and~\eqref{eq-optimal-lagrange} is unique. On the contrary, the dual certificate $\eta=\Op^*p$ which appears in~\eqref{eq-optimal-source} is not unique in general.
\end{rem}

We say that a vector $\vecdiscO$ is \textit{identifiable} if it is the unique solution to~\eqref{eq-abstract-bp} for the input $y=\Op \vecdiscO$.
The following classical result gives a sufficient condition for $\vecdiscO$ to be identifiable.

\begin{prop}
  Let $\vecdiscO\in \RR^\taillegrid$ such that $\Op_I$ is injective and that there exists $p\in \Hh$ such that 
\begin{align} 
  \norm{ (\Op^*p)_{I^c} }_\infty < 1 
  \qandq 
  (\Op^* p)_{I}=\sign(\vecdiscOI),\label{eq-optimal-strictsource}
\end{align}
where $I^c=\seg{1}{\taillegrid}\setminus I$. Then $\vecdiscO$ is identifiable.
\end{prop}

Conversely, if $\vecdiscO$ is identifiable, there exists $p\in L^{2}(\TT)$ such that~\eqref{eq-optimal-strictsource} holds and $\Op_I$ is injective (see~\cite[Lemma~4.5]{Grasmair-cpam}).

\subsection{Extended support of the \lasso}

From now on, we assume that the vector $\vecdiscO\in \RR^\taillegrid$ is identifiable (\ie $\vecdiscO$ is the unique solution to~\eqref{eq-abstract-bp} where $y_0=\Op \vecdiscO$).
We denote by $I=\supp(\vecdiscO)$ and $s_I=\sign(\vecdiscOI)$ the support and the sign of $\vecdiscO$. 

It is well known that~\eqref{eq-abstract-bp} is the limit of \eqref{eq-abstract-lasso} for $\la \to 0$ (see~\cite{chen1999atomi} for the noiseless case and~\cite{Grasmair-cpam} when the observation is $y=y_0+w$ and the noise $w$ tends to zero as a multiple of $\la$) at least in terms of the $\ell^2$ convergence.
In terms of the support of the solutions, the study in~\cite{2013-duval-sparsespikes}, which extends the one by Fuchs~\cite{fuchs2004on-sp}, emphasizes the role of a specific \textit{minimal-norm certificate} $\certdiscO$ which governs the behavior of the model at low noise regimes. 

\begin{defn}[Minimal-norm certificate and extended support]
  Let $\vecdiscO\in \RR^\taillegrid$, and let $p_0$ be the solution to~\eqref{eq-abstract-dual-bp} with minimal $L^2$-norm. The minimal-norm certificate of $a_0$ is defined as $\certdiscO \eqdef \Op^* p_0$.
The set of indices $\ext(\vecdiscO) \eqdef \enscond{1\leq j\leq \taillegrid   }{|(\certdiscO)_j|=1}$ is called the extended support of $\vecdiscO$, and 
the set $\extpm(\vecdiscO) \eqdef \enscond{(j,(\certdiscO)_j)}{j\in \ext(\vecdiscO)}\subset \seg{0}{\taillegrid-1}\times\{-1,1\}$ is called the extended signed support of $\vecdiscO$.
\end{defn}

\begin{rem}
  In the case where $a_0$ is a solution to~\eqref{eq-abstract-bp} (which is the case here since we assume that $a_0$ is an identifiable vector for~\eqref{eq-abstract-bp}), we have $(I,\sign(\vecdiscOI))\subset \extpm(\vecdiscO)$. The minimal norm certificate thus turns out to be
\eql{\label{eq-eta0}
	\certdiscO = \Op^* p_0
	\qwhereq
  p_0 = \uargmin{p\in \Hh} \enscond{\norm{p}_2}{ \normi{\Op^* p} \leq 1 \qandq \Op_I^* p = s_I }.
}
\end{rem}

It is shown in~\cite{2013-duval-sparsespikes} that there exists a low noise regime where the (signed) support of any solution $\tilde{\vecdisc}_{\la}$ of~$\Pp_\la(y_0+w)$ is included in $\extpm(\vecdiscO)$, $\suppm \tilde{\vecdisc}_\la\subset \extpm(\vecdiscO)$. It is therefore crucial to understand precisely the behavior of $\certdiscO$ and the structure of the extended (signed) support  $\extpm(\vecdiscO)$. The following (new) result gives a characterization of $\certdiscO$.

\begin{lem}\label{lem-eta0}
  Let $(J,s_J)\subset \seg{0}{\taillegrid-1} \times \{-1,1\}$ such that $(I,\sign((\vecdiscO)_I))\subset (J,s_J)$ and $\Op_J$ has full rank. Define $v_J=(\Op_J^*\Op_J)^{-1}s_J$.

  Then $(J,s_J)$ is the extended signed support of $\vecdiscO$, \textit{i.e.}  $(J,s_J)=\extpm(\vecdiscO)$, if and only if  the following two conditions hold:
\begin{itemize}
  \item for all $j\in J\setminus I$, $v_j=0$ or $s_j = -\sign(v_j)$,
  \item $\|\Op_{J^c}^* \Op_J v_J\|_{\infty}<1$.
\end{itemize}
In that case, the minimal norm certificate is given by $\certdiscO=\Op^*\Op_J^{+,*}s_J$.
\end{lem} 

\begin{proof}
  Writing the optimality conditions for~\eqref{eq-eta0}, we see that $p\in L^{2}(\TT)$ is equal to $p_0$ if and only if ${\|\Op^*p\|_{\infty}}\leq 1$, $\Op_I^*p=\sign (\vecdiscOI)$, and there exists $u_+\in (\RR^+)^\taillegrid$ and $u_-\in (\RR^+)^\taillegrid$ such that:
\begin{align}
  2p+ \Op u_+ - \Op u_- =0,
  \label{eq-optim-contrainte}
\end{align}
where for $i\in I^c$, $u_{+,i}$ (resp. $u_{-,i}$) is a Lagrange multiplier for the constraint $(\Op^*p)_i\leq 1$ (resp. $(\Op^*p)_i\geq -1$) which satisfies  the complementary slackness condition: $u_{+,i}((\Op^*p)_i- 1)=0$ (resp $u_{-,i}((\Op^*p)_i+1)=0$), and for $i\in I$, $(u_{+,i}-u_{-,i})$ is the Lagrange multiplier for the constraint $(\Op^*p)_i=\sign(\vecdiscO)_i$.

Now, let $(J,s_J)=\extpm(\vecdiscO)$ (so that $J$ determines the set of active constraints) and $p=p_0$. Using the complementary slackness condition we may reformulate~\eqref{eq-optim-contrainte} as 
\begin{align*}
  p_0- \Op_J v_J=0,
\end{align*}
for some $v\in \RR^{\taillegrid}$, where $v_j=0$ or $\sign v_j= -(\Op^*p_0)_j$ for $j\in J\setminus I$, and $v_j=0$ for $j\in\seg{0}{\taillegrid-1}\setminus J$. Inverting this relation, we obtain $v_J=(\Op_J^*\Op_J)^{-1}(\certdiscO)_J$, and the stated conditions hold.

Conversely, let $(J,s_J)\subset \seg{0}{\taillegrid-1}\times\{-1,1\}$ (not necessarily equal to $\extpm(\vecdiscO)$) such that $(I,\sign((\vecdiscO)_I))\subset (J,s_J)$ and that the conditions of the lemma hold, with $v_J=(\Op_J^*\Op_J)^{-1}s_J$. Then, setting $p=-\Op_J v_J$, we see that $\|\Op^* p\|_\infty\leq 1$, $\Op_I^*p=\sign (\vecdiscO)_I$, and~\eqref{eq-optim-contrainte} holds with the complementary slackness when setting $u_{+,j}=\frac{1}{2} \max(v_j,0)$, $u_{-,j}=\frac{1}{2} \max(-v_j,0)$ for $j\in J$ and $u_{\pm,j}=0$ for $j\notin J$. Then $p=p_0$ and the equivalence is proved.
\end{proof}

As mentioned above, the minimal norm certificate governs the (signed) support of the solution at low noise regimes insofar as the latter is contained in the extended signed support. The following theorem shows that, in the generic case, both signed supports are equal.

\begin{thm}\label{thm-stability-dbp}
  Let $\vecdiscO\in \RR^\taillegrid\setminus\{0\}$ be an identifiable signal , $J \eqdef \ext(\vecdiscO)$ such that $\Op_J$ has full rank, and $v_J \eqdef (\Op_J^* \Op_J)^{-1} \sign(\certdiscOJ)$. Assume that for all $j\in J\setminus I$, $v_j\neq 0$. Then, there exists constants $C^{(1)}>0$, $C^{(2)}>0$ (which depend only on $\Op$, $I$ and $\sign(\vecdisc_{0,I})$) such that for $\la \leq C^{(1)}\left( \umin{i \in I} |\vecdiscOI| 
  \right)$  and all $w\in \Hh$ with $\|w\|\leq C^{(2)}\la$ 
the solution $\soldisc$ of~\eqref{eq-abstract-lasso} is unique, $\supp(\soldisc) = J$ and it reads
	\eq{
    \soldiscJ = \vecdiscOJ + \Op_J^+ w - \la (\Op_J^* \Op_J)^{-1}\sign(\certdiscOJ),
	}
  where $\Op_J^+=(\Op_J^*\Op_J)^{-1}\Op_J^*$.
\end{thm}

\begin{proof}
	We define a candidate solution $\hat \vecdisc$ by 
	\eq{
		\hatvecdiscJ = \vecdiscOJ + \Op_J^+ w - \la v_{J}, \quad \hat \vecdisc_{J^c}=0
	} 
  and we prove that $\hat \vecdisc$ is the unique solution to~$(\Pp_\la(y_0+w))$ using the optimality conditions~\eqref{eq-optimal-subdiff} and~\eqref{eq-optimal-lagrange}. 
	
We first exhibit a condition for $\sign(\hatvecdiscJ)=\sign(\certdiscOJ)$. To shorten the notation, we write $s_J=\sign(\certdiscOJ)$.
Since for $i\in I$, $\vecdisc_{0,i}\neq 0$, the constraint $\sign(\hatvecdiscI)=s_I$ is implied by 
	\eq{
		\norm{R_I \Op_J^+}_{\infty,2} \norm{w} + \norm{v_I}_{\infty} \la < T, \qwhereq T=\umin{i \in I} |\vecdiscOI| > 0,
	}
and  $R_I: u\mapsto u_I$  is the restriction operator.
As for $ K=J\setminus I$, for all $k\in K$ $\vecdisc_{0,k}=0$ but we know from Lemma~\ref{lem-eta0} that $\sign(v_{k}) = -s_k$.
	The constraint $\sign(\hatvecdisc_K)=s_K$ is thus implied by 
	\eq{
    \norm{R_K \Op_J^+}_{\infty,2} \norm{w} \leq \la \underbrace{\left(\umin{k \in K} |v_{k}|\right)}_{>0}.
	}
  Hence, we have $\sign\hatvecdiscJ= \sign \certdiscOJ=s_J$, and by construction $\supp(\hatvecdisc) = J$ with
	\begin{align}
		\Op_J^*( y-\Op \hat a ) = \la s_J.
	\end{align}

  To ensure that $\hatvecdisc$ is the unique solution to~$(\Pp_\la(y))$ with $y=y_0+w$, it remains to check that
  \begin{align}
    \normi{\Op_{J^c}^*(y-\Op \hatvecdisc)}<\la.
    \label{eq-abstract-strict}
  \end{align}

	Since we have
	\eq{
    y - \Op \hatvecdisc = w-\Op_J(\Op_J^*\Op_J)^{-1}\Op_J^*w + \la \Op_J (\Op_J^*\Op_J)^{-1} s_J  = P_{\ker(\Op_J^*)} w + \la \Op_J^{+,*} s_J,
	}
  we see that~\eqref{eq-abstract-strict} is implied by
  \eq{\label{eq-rough-snr}
		\norm{ \Op_{J^c}^* P_{\ker(\Op_J^*)} }_{2,\infty} \norm{w}
		- \la ( 1 - \normi{\eta_{0,J^c}} ) < 0
	}
	where by construction $\normi{\eta_{0,J^c}}<1$.
	
	Putting everything together, one sees that $\hat a$ is the unique solution of~$(\Pp_\la(y))$ if the following affine inequalities hold simultaneously 
    \begin{align}
    	c_1 \norm{w} + c_2 \la &< T \qwhereq
		    \choice{
			    c_1 \eqdef \norm{R_I \Op_J^+}_{\infty,2}, \\
			    c_2 \eqdef \norm{v_I}_{\infty}, 	
		    }\\
        \norm{w} &\leq c_3\la  
      \qwhereq c_3\eqdef (\norm{R_K \Op_J^+}_{\infty,2})^{-1}\left(\umin{k \in K} |v_{k}|\right)>0,\\
      c_4 \norm{w} - c_5 \la &< 0
	    	\qwhereq
		    \choice{
			      c_4 \eqdef  \norm{ \Op_{J^c}^* P_{\ker(\Op_J^*)} }_{2,\infty}, \\
			      c_5 \eqdef  1 - \normi{\eta_{0,J^c}} >0.
        }
     %\left(\umin{k \in K} |v_{k}|\right) \norm{R_K \Op_J^+}_{\infty,2}^{-1} )>0.
    \end{align}
Hence, for $\|w\| <\min(c_3,\frac{c_5}{c_4}) \la$ and $\left(\frac{c_1c_5}{c_4}+c_2 \right)\la<T$, the first order optimality conditions hold.
\end{proof}

\begin{rem}[Comparison with the analysis of Fuchs]\label{rem-fuchs}
	When $J=I$, Theorem~\ref{thm-stability-dbp} recovers exactly the result of Fuchs~\cite{fuchs2004on-sp}. Note that this result has been extended beyond the $\ell^1$ setting,see in particular~\cite{2015-vaiter-piecewiseregular,2014-vaiter-ps-consistency} for a unified treatment of arbitrary partly smooth convex regularizers. For this result to hold, i.e. to obtain $I=J$, one needs to impose that the following pre-certificate
	\eql{\label{eq-fuchs-precertif}
		\fuchsdisc \eqdef \Op^* \Op_I^{+,*} s_I
	} 
  is a valid certificate, i.e. one needs that $\normi{\fuchsdiscIC} < 1$. This condition is often called the \textit{irrepresentability condition} in the statistics literature (see for instance~\cite{Zhao-irrepresentability}). It implies that the support $I$ is stable for small noise. Unfortunately, it is easy to verify that for the deconvolution problem, in general, this condition does not hold when the grid stepsize is small enough (see~\cite[Section 5.3]{2013-duval-sparsespikes}), so that one cannot use the initial result. This motivates our additional study of the extended support $\ext(\vecdiscO) \supset I$, which is always stable to small noise. While this new result is certainly very intuitive, to the best of our knowledge, it is the first time it is stated and proved, with explicit values of the stability constant involved. 
\end{rem}

\begin{rem}
  Theorem~\ref{thm-stability-dbp} guarantees that the support of the reconstructed signal $\soldisc$ at low noise is equal to the extended support. The required condition $v_j\neq 0$ in Theorem~\ref{thm-stability-dbp} is tight in the sense that if $v_j=0$ for some $j\in J\setminus I$, then the saturation point of $\certdisc_\la$ may be strictly included in $J$. Indeed, it is possible, using similar calculations as above, to construct $w$ such that $\supp \soldisc\subsetneq J$ with $\la$ and $\|w\|_2/\la$ arbitrarily small.
\end{rem}

% !TEX root = ../continuous-bp.tex
\section{\lasso on thin Grids}
\label{sec-discbp-thin}

In this section, we focus on inverse problems with smooth kernels, such as for instance the deconvolution problem. 
Our aim is to recover a measure $m_0\in\Mm(\TT)$ from the observation $y_0=\Phi m_0$ or $y=\Phi m_0+w$, where $\varphi\in \Cont^k(\TT\times \TT)$ ($k\geq 2$), $w\in L^2(\TT)$ and 
\begin{align}
\forall x\in \TT,\  (\Phi m)(x) \eqdef \int_\TT \varphi(x,y) \d m(y),
\end{align}
so that $\Phi:\Mm(\TT)\rightarrow L^2(\TT)$ is a bounded linear operator. Observe that $\Phi$ is in fact weak* to weak continuous and its adjoint is compact (see Lemma~\ref{lem-phi-compact} in Appendix).

Typically, we assume that the unknown measure $m_0$ is sparse, in the sense that it is of the form $m_0=\sum_{\nu=1}^N \alpha_{0,\nu}x_{0,\nu}$ for some $N\in\NN^*$, here $\alpha_{0,\nu}\in \RR^*$ and the $x_{0,\nu}\in \TT$ are pairwise distinct. 

The first approach we study is the one of the (discrete) Basis Pursuit. We look for measures that have support on a certain discrete grid $\Gg\subset \TT$, and we want to recover the original signal by solving an instance of~\eqref{eq-abstract-bp} or~\eqref{eq-abstract-lasso} on that grid. Specifically, we aim at analyzing the behavior of the solutions at low noise regimes (\textit{i.e.} when the noise $w$ is small and $\la$ well chosen) as the grid gets thinner and thinner.
To this end, we take advantage of the characterizations given in Section~\ref{sec-discrete-lasso} with $\Hh \eqdef L^2(\TT)$, regardless of the grid, and we use the Beurling \lasso~\eqref{eq-blasso} as a limit of the discrete models.

%%%%%%%%%%%%%%%%%%%%%%%%%%%%%%%%%%%%%%%%%%%%%%%%%%%%
\subsection{Notations and preliminaries}

For the sake of simplicity we only study uniform grids, i.e. $\Gg \eqdef \enscond{ih}{i\in \seg{0}{\taillegrid-1}}$ where $\stepsize \eqdef \frac{1}{\taillegrid}$ is the stepsize. Moreover, we shall consider sequences of grids $(\Gg_n)_{n\in\NN}$ such that the stepsize vanishes ($\stepsizen=\frac{1}{\taillegridn}\to 0$ as $n\to+\infty$) and to ensure monotonicity, we assume that $\Gg_{n+1}\subset \Gg_n$. For instance, the reader may think of a dyadic grid (\ie $\stepsizen=\frac{\stepsize_0}{2^n}$).
We shall identify in an obvious way measures with support in $\Gg_n$ (\ie of the form $\sum_{k=0}^{\taillegridn-1} \vecthin_k \delta_{k\stepsizen}$) and vectors $a\in\RR^\taillegridn$.

The problem we consider is a particular instance of~\eqref{eq-abstract-lasso} (or~\eqref{eq-abstract-bp}) when choosing $\Op$ as the restriction of $\Phi$ to measures with support in the grid~$\Gg_n$,
\begin{align}
  \Op\eqdef\Phi_{\Gg_n} = \begin{pmatrix}
    \phi(\cdot,0),\ldots,\phi(\cdot,(\taillegrid-1)\stepsizen)
  \end{pmatrix}.
\end{align}
More explicitely, on the grid $\Gg_n$, we solve
\begin{align}
  \umin{\vecthin \in \RR^\taillegridn} \frac{1}{2}\norm{y-\Phi_{\Gg_n} \vecthin}^2 + \la \norm{\vecthin}_1, \tag{$\Pp_\la^n(y)$}\label{eq-thin-lasso}\\
  \mbox{and }  \umin{\vecthin\in \RR^\taillegridn} \norm{\vecthin}_1 \mbox{ such that } \Phi_{\Gg_n} \vecthin=y_0.  \tag{$\Pp_0^n(y_0)$}\label{eq-thin-bp}
\end{align}

We say that a measure $m_0=\sum_{\nu=1}^N \alpha_{0,\nu} \delta_{x_{0,\nu}}$ (with $\alpha_{0,\nu}\neq 0$ and the $x_{0,\nu}$'s pairwise distinct) is identifiable through~\eqref{eq-thin-bp} if it can be written as $m_0=\sum_{k=0}^{\taillegridn-1} \vecthin_i \delta_{i\stepsizen}$ and that the vector $\vecthin$ is identifiable using~\eqref{eq-thin-bp}. 

As before, given $\vecthin\in \RR^{\taillegridn}$, we shall write $I(\vecthin) \eqdef \enscond{i\in \seg{0}{\taillegridn-1}}{\vecthin_i\neq 0}$ or simply $I$ when the context is clear.

The optimality conditions~\eqref{eq-optimal-lagrange} amount to the existence of some $p_\la\in L^2(\TT)$ such that
  \begin{align}
    \max_{0\leq k\leq \taillegridn-1} |(\Phi^*p_\la)(k\stepsizen)|\leq 1,\qandq (\Phi^*p_\la)(\gIn)&=\sign(\solthinI), \label{eq-optimal-thin-subdiff}\\
    \la  \Phi^*p_\la + \Phi^*(\Phi \vecthin_\la-y) &=0.\label{eq-optimal-thin-lagrange}
    \end{align}
Similarly the optimality condition~\eqref{eq-optimal-source} is equivalent to the existence of $p\in L^2(\TT)$ such that
    \begin{align}
    \max_{0\leq k\leq \taillegridn-1}|(\Phi^*p)(k\stepsizen)| \leq 1 \qandq (\Phi^* p)(\gIn)&=\sign(\vecthinOI).\label{eq-optimal-thin-source}
\end{align}
Notice that the dual certificates are naturally given by the sampling of continuous functions $\certthin=\Phi^*p: \TT\rightarrow \RR$, and that the notation $\certthin(\gIn)$ or $(\Phi^* p)(\gIn)$ stands for $(\certthin(i\stepsizen))_{i\in I}$ where $I=I(\vecthinO)$ (and similarly for $\certthin_\la=\Phi^*p_\la$ and $I(\solthin)$).

%\subsection{Minimal norm certificates on thin grids}
If $m_0$ is identifiable through~\eqref{eq-thin-bp}, the minimal norm certificate for the problem~\eqref{eq-thin-bp} (see Section~\ref{sec-discrete-lasso}) is denoted by $\certthinO$, whereas the extended support on $\Gg_n$ is defined as 
\begin{align}
  \ext_n m_0 \eqdef \enscond{t\in \Gg_n}{\certthinO(t)=\pm 1}.
\end{align}
From Section~\ref{sec-discrete-lasso}, we know that the extended support is the support of the solutions at low noise.

\subsection{The limit problem: the Beurling lasso}
\label{subsec-beurl-lasso}

It turns out that Problems~\eqref{eq-thin-lasso} and~\eqref{eq-thin-bp} have natural limits when the grid gets thin.
Embedding those problems into the space $\Mm(\TT)$ of Radon measures, the present authors have studied in~\cite{2013-duval-sparsespikes} their convergence towards the Beurling-\lasso used in~\cite{deCastro-beurling,Candes-toward,Bredies-space-measures,Tang-linea-spectral}.

The idea is to recover the measure $m_0$ using the following variants of~\eqref{eq-abstract-lasso} and~\eqref{eq-abstract-bp}:
\begin{align}
  \umin{m\in \Mm(\TT)} \frac{1}{2}\norm{y-\Phi m}^2 + \la |m|(\TT), \tag{$\Pp_\la^\infty(y)$}\label{eq-beurl-lasso}\\
 \qandq  \umin{m\in \Mm(\TT)} |m|(\TT) \quad\mbox{such that}\quad \Phi m=y_0,  \tag{$\Pp_0^\infty(y_0)$}\label{eq-beurl-bp}
\end{align}
where $|m|(\TT)$ refers to the total variation of the measure $m$
\begin{align}
  |m|(\TT) \eqdef \sup\enscond{\int_\TT \psi(x) \d m(x)}{\psi\in C(\TT)\mbox{ and } \|\psi\|_\infty\leq 1}.
\end{align}
Observe that in this framework, the notation $\|\psi\|_\infty$ stands for $\sup_{t\in\TT}|\psi(t)|$.
When $m$ is of the form $m=\sum_{\nu=1}^N \alpha_{\nu}x_{\nu}$ where $\alpha_{\nu}\in\RR^*$ and $x_{\nu}\in\TT$ (with the $x_{\nu}$'s pairwise distinct), $|m|(\TT)=\sum_{\nu=1}^N |\alpha_\nu|$, so that those problems are natural extensions of~\eqref{eq-abstract-lasso} and~\eqref{eq-abstract-bp}.
This connection is emphasized in~\cite{2013-duval-sparsespikes} by embedding~\eqref{eq-thin-lasso} and~\eqref{eq-thin-bp} in the space of Radon measures $\Mm(\TT)$, using the fact that
\begin{align*}
  &\sup\enscond{\int_\TT \psi(x) \d m(x)}{\psi\in C(\TT), \forall k\in\seg{0}{\taillegridn-1}\  |\psi|(k\stepsizen)\leq 1 }\\
  &\qquad = \left\{\begin{array}{ll}
    \|\vecthin\|_1 \mbox{ if } m=\sum_{k=0}^{\taillegridn-1} \vecthin_k \delta_{k\stepsizen},\\
    +\infty \mbox{ otherwise.}
  \end{array}\right.
\end{align*}

We say that $m_0$ is identifiable through~\eqref{eq-beurl-bp} if it is the unique solution of~\eqref{eq-beurl-bp}.
A striking result of~\cite{Candes-toward} is that when $\Phi$ is the ideal low-pass filter and that the spikes $m_0=\sum_{\nu=1}^N \alpha_{0,\nu}x_{0,\nu}$ are sufficiently far from one another, the measure $m_0$ is identifiable through~\ref{eq-beurl-bp}. 

The optimality conditions for~\eqref{eq-beurl-lasso} and~\eqref{eq-beurl-bp} are similar to those of the abstract \lasso (respectively \eqref{eq-optimal-subdiff}, \eqref{eq-optimal-lagrange} and \eqref{eq-optimal-source}). The corresponding dual problems are 
\begin{align}
  \inf_{p\in C^\infty} &\left\|\frac{y}{\la}-p  \right\|_2^2 \tag{$\Dd^\infty_\la(y)$},\label{eq-beurling-dual-lasso}\\
  \mbox{resp. }\quad  \sup_{p\in C^\infty} &\langle y_0, p\rangle, \tag{$\Dd^\infty_0(y_0)$}\label{eq-beurling-dual-bp}\\
  \qwhereq C^\infty& \eqdef \enscond{p\in L^2(\TT)}{\|\Phi^*p\|_{\infty} \leq 1}.
\end{align}
The source condition associated with~\eqref{eq-beurl-bp} is of particular interest. It amounts to the existence of some $p\in L^2(\TT)$ such that
 \begin{align}
   \|\Phi^*p\|_\infty \leq 1\qandq (\Phi^*p)(x_{0,\nu}) =\sign(\vecbeurlOnu) \mbox{ for all } \nu \in\{1,\ldots,N\}.
\label{eq-optimal-beurl-source}
 \end{align}
 Here, $\|\Phi^*p\|_\infty= \sup_{t\in\TT}|(\Phi^*p)(t)|$.  Moreover, if such $p$ exists and satisifies $|(\Phi^*p)(t)|<1$ for all $t\in\TT\setminus \{x_{0,1},\ldots ,x_{0,N} \}$, and $\Phi_{x_0}$ has full rank, then $m_0$ is the \textit{unique} solution to~\eqref{eq-beurl-bp} (\ie $m_0$ is identifiable). 
 
 Observe that in this infinite dimensional setting, the source condition~\eqref{eq-optimal-beurl-source} implies the optimality of $m_0$ for~\eqref{eq-beurl-bp} but the converse is not true (see~\cite{2013-duval-sparsespikes}). 

\begin{rem}\label{rem-beurl-vs-thin}
  A simple but crucial remark made in~\cite{Candes-toward} is that if $m_0$ is identifiable through~\eqref{eq-beurl-bp} and that $\supp m_0\subset \Gg_n$, then $m_0$ is identifiable for~\eqref{eq-thin-bp}. 
Similarly, observe that the source condition for~\eqref{eq-beurl-bp} implies the source condition for the~\eqref{eq-thin-bp}.
\end{rem}

If we are interested in noise robustness, a stronger assumption is the \textit{Non Degenerate Source Condition} which relies on the notion of minimal norm certificate for~\eqref{eq-beurl-bp}. When there is a solution to~\eqref{eq-beurling-dual-bp}, the one with minimal $L^2$ norm, $p_0^\infty$, determines the \textit{minimal norm certificate} $\eta_0^\infty\eqdef \Phi^*p_0^\infty$. When $m_0$ is a solution to~\eqref{eq-beurl-bp}, the \textit{minimal norm certificate} can be characterized as 
\begin{align}
  \certbeurlO&=\Phi^*\pbeurlO \qwhereq \label{eq-certbeurl0}\\
  \pbeurlO &= \uargmin{p\in L^2(\TT)} \enscond{\norm{p}_2}{ \normi{\Phi^* p} \leq 1, \: (\Phi^* p)(x_{0,\nu}) = \sign(\alpha_{0,\nu}), 1\leq \nu\leq N}.
\end{align}
 As with the discrete \lasso problem, a notion of extended (signed) support $\extpm_\infty$ may be defined and the minimal norm certificate governs the behavior of the solutions at low noise (see~\cite{2013-duval-sparsespikes} for more details).

 \begin{defn}\label{defn-ndsc-bp}
   Let $m_0=\sum_{\nu=1}^N\alpha_{0,\nu} \delta_{x_{0,\nu}}$ an identifiable measure for~\eqref{eq-beurl-bp}, and $\certbeurlO\in C(\TT)$ its minimal norm certificate. We say that $m_0$ satisfies the \textit{Non Degenerate Source Condition} if
\begin{itemize}
  \item $|\certbeurlO(t)|<1$ for all $t\in \TT\setminus \{x_{0,1},\ldots x_{0,N} \}$,
  \item $\certbeurlO''(x_{0,\nu})\neq 0$ for all $\nu\in\{1,\ldots,N\}$.
\end{itemize}
 \end{defn}

The Non Degenerate Source Condition might seem difficult to check in practice. The following proposition shows that it is in fact easy to check numerically on the vanishing derivatives precertificate.

\begin{defn}\label{defn-vanishing-bp}
  Let $m_0=\sum_{\nu=1}^N\alpha_{0,\nu} \delta_{x_{0,\nu}}$ an identifiable measure for~\eqref{eq-beurl-bp} such that $\Gamma_{x_0}\eqdef \begin{pmatrix} \Phi_{x_0} & \Phi'_{x_0}\end{pmatrix}$ has full rank. We define the vanishing derivatives precertificate as $\certvanishing \eqdef \Phi^*\pvanishing$ where 
\begin{align}
  \pvanishing \eqdef \uargmin{p\in L^2(\TT)} \enscond{\norm{p}_2}{(\Phi^* p)(x_{0,\nu}) = \sign(\alpha_{0,\nu}), (\Phi^* p)'(x_{0,\nu})=0, \ 1\leq \nu\leq N}.
  \label{eq-vanishing-p}
\end{align}
\end{defn}

The following proposition shows that this precertificate is easily computed by solving a linear system in the least square sense. 

\begin{prop}[\cite{2013-duval-sparsespikes}]
  Let $m_0=\sum_{\nu=1}^N\alpha_{0,\nu} \delta_{x_{0,\nu}}$ an identifiable measure for the problem~\eqref{eq-beurl-bp} such that $\Gamma_{x_0}$ has full rank.
  
  Then, the vanishing derivatives precertificate can be computed by 
  \eql{\label{eq-vanishing-expression}
		\eta_V^\infty \eqdef \Phi^*\pvanishing \qwhereq \pvanishing \eqdef \Gamma_{x_0}^{+,*}\begin{pmatrix} \sign(\vecbeurl_{0,\cdot})\\0
      \end{pmatrix}, 
     }
     and $\Gamma_{x_0}^{+,*}=\Gamma_{x_0}(\Gamma_{x_0}^*\Gamma_{x_0})^{-1}$.
Moreover, the following conditions are equivalent:
  \begin{enumerate}
    \item $m_0$ satisfies the Non Degenerate Source Condition.
    \item The vanishing derivatives precertificate satisfies:
      \begin{itemize}
        \item $|\certvanishing(t)|<1$ for all $t\in \TT\setminus \{x_{0,1},\ldots x_{0,N} \}$,
        \item $\certvanishing''(x_{0,\nu})\neq 0$ for all $\nu\in\{1,\ldots,N\}$.
      \end{itemize}
   \end{enumerate}
   And in that case, $\certvanishing$ is equal to the minimal norm certificate $\certbeurlO$.
   \label{prop-etav-nonvanish}
\end{prop}

\begin{rem}\label{rem-vanishing-formula}
  Using the block inversion formula in~\eqref{eq-vanishing-expression}, it is possible to check that 
  \begin{align}
    \pvanishing&=\Phi_{x_0}^{+,*}\sign(\vecbeurl_{0,\cdot}) - \Pi \Phi_{x_0}' ({\Phi_{x_0}'}^*\Pi \Phi_{x_0}')^{-1}{\Phi_{x_0}'}^*\Phi_{x_0}^{+,*}\sign(\vecbeurl_{0,\cdot}),\label{eq-vanishing-formula}
  \end{align} where $\Pi$ is the orthogonal projector onto $(\Im \Phi_{x_0})^\perp$. If we denote by $\pfuchsbeurl$ the vector introduced by Fuchs (see~\eqref{eq-fuchs-precertif}), which turns out to be
  \begin{align*}
    \pfuchsbeurl= \uargmin{p\in L^2(\TT)} \enscond{\norm{p}_2}{(\Phi^* p)(x_{0,\nu}) = \sign(\alpha_{0,\nu}), \ 1\leq \nu\leq N},
  \end{align*}
  we observe that $\pvanishing = \pfuchsbeurl - \Pi \Phi_{x_0}' ({\Phi_{x_0}'}^*\Pi \Phi_{x_0}')^{-1}{\Phi_{x_0}'}^*\pfuchsbeurl$.
\end{rem}

 \begin{rem}
   At this stage, we see that two different minimal norm certificates appear: the one for the discrete problem~\eqref{eq-thin-bp} which should satisfy~\eqref{eq-optimal-thin-source} on a discrete grid $\Gg_n$, and the one for gridless problem~\eqref{eq-beurl-bp} which should satisfy~\eqref{eq-optimal-beurl-source}. One should not mingle them.
 \end{rem}

\subsection{The \lasso on thin grids for fixed $\la>0$}

As hinted by the notation, Problem~\eqref{eq-beurl-lasso} is the limit of Problem~\eqref{eq-thin-lasso} as the stepsize of the grid vanishes (\ie $n\to+\infty$).
Indeed, we may identify each vector $a\in\RR^{\taillegridn}$ with the measure $m_a= \sum_{k=0}^{\taillegridn-1} a_k\delta_{k\stepsizen}$ (so that $\|a\|_1=|m_a|(\TT)$) and embed~\eqref{eq-thin-lasso} into the space of Radon measures. With this identification, the Problem~\eqref{eq-thin-lasso} $\Gamma$-converges towards Problem~\eqref{eq-beurl-lasso} (see the definition below), and as a result, any accumulation point of the minimizers of~\eqref{eq-thin-lasso} is a minimizer of~\eqref{eq-beurl-lasso}.

\begin{rem}
  The space $\Mm(\TT)$ endowed with the weak* topology is a topological vector space which does not satisfy the first axiom of countability (\ie the existence of a countable base of neighborhoods at each point). However, each solution $m^n_\la$ of~\eqref{eq-thin-lasso} (resp. $m^\infty_\la$ of~\eqref{eq-beurl-lasso}) satisfies 
  \begin{align}
\la |m^n_\la|(\TT)\leq \la |m^n_\la|(\TT)+\frac{1}{2}\norm{\Phi m^n_\la-y}^2 \leq \frac{1}{2}\norm{y}^2.
\label{eq-X-bounded}
  \end{align}
  Hence we may restrict those problems to the set
  	\eq{
  		X \eqdef \enscond{m\in \Mm(\TT)}{\la |m|(\TT)\leq \frac{1}{2}\norm{y}^2}
	} 
	which is a metrizable space for the weak* topology. As a result, we shall work with the definition of $\Gamma$-convergence in metric spaces, which is more convenient than working with the general definition~\cite[Definition~4.1]{Dalmaso1993}). For more details about $\Gamma$-convergence, we refer the reader to the monograph~\cite{Dalmaso1993}.
\end{rem}

\begin{defn}
  We say that the Problem~\eqref{eq-thin-lasso} $\Gamma$-converges towards Problem~\eqref{eq-beurl-lasso} if, for all $m\in X$, the following conditions hold
  \begin{itemize}
    \item{(Liminf inequality)} for any sequence of measures $(m^n)_{n\in\NN}\in X^\NN$ such that $\supp(m^n)\subset\Gg_n$ and that $m^n$ weakly* converges towards $m$, 
      \begin{align*}
        \liminf_{n\to +\infty}\left(\la |m^n|(\TT)+\frac{1}{2}\norm{\Phi m^n-y}^2\right) \geq \la |m|(\TT) + \frac{1}{2}\norm{\Phi m-y}^2.
      \end{align*}
    \item{(Limsup inequality)} there exists a sequence of measures $(m^n)_{n\in \NN}\in X^\NN$ such that $\supp(m^n)\subset \Gg_n$, $m^n$ weakly* converges towards $m$  and
      \begin{align*}
        \limsup_{n\to +\infty} \left(\la |m^n|(\TT)+\frac{1}{2}\norm{\Phi m^n-y}^2\right) \leq \la |m|(\TT) + \frac{1}{2}\norm{\Phi m-y}^2.
      \end{align*}
  \end{itemize}
  \label{def-gammacv}
\end{defn}

The following proposition shows the $\Gamma$-convergence of the discretized problems toward the Beurling \lasso problem.  This  ensures in particular the convergence of the minimizers, which was already proved in~\cite{TangConvergence}. 

\begin{prop}
  The Problem~\eqref{eq-thin-lasso} $\Gamma$-converges towards \eqref{eq-beurl-lasso}, and 
\begin{align}
  \lim_{n\to +\infty} \inf \eqref{eq-thin-lasso} = \inf \eqref{eq-beurl-lasso}.\label{eq-cvinf-lasso}
\end{align}
Each sequence $(m^n_\la)_{n\in\NN}$ such that $m^n_\la$ is a minimizer of~\eqref{eq-thin-lasso} has accumulation points (for the weak*) topology, and each of these accumulation point is a minimizer of~\eqref{eq-beurl-lasso}.
\label{prop-gammacv-lasso}
\end{prop}
In particular, if the solution $m_\la$ to~\eqref{eq-beurl-lasso} is unique, the minimizers of \eqref{eq-thin-lasso} converge towards $m_\la$.
\begin{proof}
  The liminf inequality of Definition~\eqref{def-gammacv} is a consequence of the lower semi-continuity of the total variation and the $L^2$ norm (since $\Phi$ is weak* to weak continuous, $\Phi_{\Gg_n}m^n -y$ weakly converges towards $\Phi m-y$): 
  \begin{align*}
    & \liminf_{n\to +\infty} \left(\la |m^n|(\TT)+\frac{1}{2}\norm{\Phi m^n-y}^2 \right) \\
    & \qquad \geq \la \liminf_{n\to +\infty} (|m^n|(\TT))+\frac{1}{2}\liminf_{n\to +\infty}\left(\norm{\Phi m^n-y}^2 \right)\\
    & \qquad \geq \la |m|(\TT) + \frac{1}{2}\norm{\Phi m-y}^2.
  \end{align*}

  As for the limsup inequality, we approximate $m$ with the measure $m^n= \sum_{k=0}^{\taillegridn-1} b_k \delta_{k\stepsizen}$, where $b_k=m([k\stepsizen,(k+1)\stepsizen))$. Then, for any $\psi\in C(\TT)$,
  \begin{align*}
    \left|\int_{\TT} \psi \d m - \int_{\TT} \psi\d m^n \right| &= \left|\sum_{k=0}^{\taillegridn-1} \int_{[k\stepsizen,(k+1)\stepsizen)} (\psi(x)-\psi(k\stepsizen))\d m \right|\\
    &\leq \omega_\psi(\stepsizen)|m|(\TT),
  \end{align*}
  where $\omega_\psi: t\mapsto \sup_{|x'-x|\leq t}|\psi(x)-\psi(x')|$ is the modulus of continuity of $\psi$. Therefore, $\lim_{n\to +\infty} \langle m^n, \psi\rangle = \langle m,\psi\rangle$, and $m^n$ weakly* converges towards $m$. Incidentally, observe that $|m^n|(\TT)\leq |m|(\TT)$, so that using the liminf inequality we get $\lim_{n\to+\infty} |m^n|(\TT)= |m|(\TT)$. Moreover, by similar majorizations, one may prove that $\Phi m^n$ converges strongly in $L^2(\TT)$ towards $\Phi m$. As a result $\lim_{n\to +\infty} \norm{\Phi m^n-y}^2 =\norm{\Phi m-y}^2$, and the limsup inequality is proved.

  Eventually, from~\eqref{eq-X-bounded} we deduce the compactness of $X$, hence the existence of accumulation points,  and~\cite[Theorem~7.8]{Dalmaso1993} implies that accumulation points of $(m^n_\la)_{n\in\NN}$ are minimizers of \eqref{eq-beurl-lasso}, as well as~\eqref{eq-cvinf-lasso}.
\end{proof}

The weak* convergence of the minimizers of~\eqref{eq-thin-lasso} can be described more accurately by studying the dual certificates $p_\la$ and looking at the support of the solutions $m_\la^n$ to~\eqref{eq-thin-lasso} (see~\cite[Section~5.4]{2013-duval-sparsespikes}). One may prove that $m_\la^n$ is generally composed of at most one pair of Dirac masses in the neighborhood of each Dirac mass of the solution $m_{\la}^\infty=\sum_{\nu=1}^{N_\la} \alpha_{\la,\nu}\delta_{x_{\la,\nu}}$ to~\eqref{eq-beurl-lasso}. More precisely, 
\begin{prop}
  Let $\la >0$, and assume that there exists a solution to~\eqref{eq-beurl-lasso} which is a sum of a finite number of Dirac masses: $m_{\la}^\infty=\sum_{\nu=1}^{N_\la} \alpha_{\la,\nu}\delta_{x_{\la,\nu}}$ (where $\alpha_\nu\neq 0$). Assume that the corresponding dual certificate $\eta_\la^{\infty}=\Phi^*p_{\la}^\infty$ satisfies $|\eta_\la^\infty(t)|<1$ for all $t\in \TT\setminus \{x_1,\ldots ,x_N \}$.
    
  Then any sequence of solution $m_{\la}^n=\sum_{i=0}^{\taillegridn-1} \veccont_{\la,i} \delta_{i\stepsizen}$ to~\eqref{eq-thin-lasso} satisfies 
  \begin{align*}
    \limsup_{n\to+\infty}\left(\supp(m_{\la}^n)\right)\subset \{x_1, \ldots x_N\}.
  \end{align*}
  If, moreover, $m_{\la}^\infty$ is the unique solution to~\eqref{eq-beurl-lasso}, 
  \begin{align}
    \lim_{n\to+\infty}\left(\supp(m_{\la}^n)\right)= \{x_1, \ldots x_N\}.
  \end{align}

  If, additionally, $(\eta_\la^{\infty})''(x_\nu)\neq 0$ for some $\nu \in\{1,\ldots,N\}$, then for all $n$ large enough, the restriction of $m_{\la}^{n}$ to $(x_\nu-r,x_\nu+r)$ (with $0<r<\frac{1}{2} \min_{\nu-\nu'}|x_{\la,\nu}-x_{\la,\nu'}|$) is a sum of Dirac masses of the form $ \vecthin_{\la,i}\delta_{i\stepsizen}+ \vecthin_{\la,i+\varepsilon_{i,n}}\delta_{(i+\varepsilon_{i,n})\stepsizen}$ with $\varepsilon_{i,n}\in\{-1,1\}$, $\vecthin_{\la,i}\neq 0$ and $\sign(\vecthin_{\la,i})=\sign(\alpha_{\la,\nu})$. Moreover, if $\vecthin_{\la,i+\varepsilon_{i,n}}\neq 0$, $\sign(\vecthin_{\la,i+\varepsilon_{i,n}})=\sign(\alpha_{\la,\nu})$.
\end{prop}
We skip the proof as it is very close to the arguments of~\cite[Section~5.4]{2013-duval-sparsespikes}. Moreover the proof of Proposition~\ref{prop-cbp-thin-supplambda} below for the C-BP is quite similar.

%$m_\la^n$ is composed of at most one pair of Dirac masses in the neighborhood of each Dirac mass of the solution $m_{\la}^\infty=\sum_{\nu=1}^{N_\la} \alpha_{\la,\nu}\delta_{x_{\la,\nu}}$ to~\eqref{eq-beurl-lasso}, provided that $(\Phi^{*}p_{\la^\infty})''(x_{\la,\nu})\neq 0$ and $|\Phi^{*}p_{\la^\infty}(t)|<1$ for all $t\in\TT\setminus\{(x_{\la,1},\ldots , (x_{\la,N_\la} \}$, where $p_\la^\infty$ is the solution to~\eqref{eq-beurling-dual-lasso}. 

%\todo{Proposition convergence des supports}

\subsection{Convergence of the extended support}
Now, we focus on the study of low noise regimes. The convergence of the extended support for~\eqref{eq-thin-bp} towards the extended support of~\eqref{eq-beurl-bp} is analyzed by the following proposition.

From now on, we assume that the source condition for~\eqref{eq-beurl-bp} holds, and that $\supp m_0\subset \Gg_n$ for $n$ large enough (in other words, $y_0=\Phi_{\Gg_n} \vecthinO$ for some $\vecthinO\in \RR^{\taillegridn}$), so that $m_0=\sum_{\nu=1}^N\alpha_{0,i}\delta_{x_{0,\nu}}$ is a solution of~\eqref{eq-thin-bp}. Moreover we assume that $n$ is large enough so that $|x_{0,\nu}-x_{0,\nu'}|> {2}{\stepsizen}$ for $\nu'\neq \nu$.

\begin{prop}[\cite{2013-duval-sparsespikes}]
  The following result holds:
  \begin{align}
    \lim_{n\to +\infty} \certthinO = \certbeurlO,
  \end{align}
  in the sense of the uniform convergence (which also holds for the first and second derivatives). Moreover, if $m_0$ satisfies the Non Degenerate Source Condition, for $n$ large enough, there exists $\varepsilon^n\in\{-1,0,+1\}^N$ such that
\begin{align}
\extpm^{n}(m_0) =\suppm(m_0) \cup \left(\suppm(m_0) + \varepsilon^n \stepsizen \right), 
\end{align}
where $\suppm(m_0) + \varepsilon^n \stepsizen \eqdef \enscond{ (x_{0,\nu}+\varepsilon^n_\nu h_n,\certbeurlO(x_{0,\nu})) }{ 1\leq \nu\leq N }$.
\label{prop-extsupp}
\end{prop}

That result ensures that on thin grids, there is a low noise regime for which the solutions are made of the same spikes as the original measure, plus possibly one immediate neighbor of each spike with the same sign. However, it does not predict which neighbors may appear and where (is it at the left or at the right of the original spike?).

The following Theorem refines that result by giving a sufficient condition for the spikes to appear in pairs (\ie $\epsilon_\nu=\pm 1$ for $1\leq \nu\leq N$). Moreover, it shows that the value of $\varepsilon^n$ does not depend on $n$, and it gives the explicit positions of the added spikes $\varepsilon_\nu$, for $1\leq \nu\leq N$.

\begin{thm}\label{thm-lasso-extended}
  Assume that the  operator $\Gamma_{x_0}=\begin{pmatrix}
  \Phi_{x_0} & \Phi_{x_0}'
\end{pmatrix}$ has full rank, and that $m_0$ satisfies the Non-Degenerate Source Condition. Moreover, assume that all the components of the natural shift
\begin{align}\label{eq-dfn-rho}
\rho\eqdef  (\Phi_{x_0}'^*\Pi\Phi_{x_0}')^{-1}\Phi_{x_0}'^*\Phi_{x_0}^{+,*}\sign(m_0(x_0))
\end{align}
 are nonzero, where $\Pi$ is the orthogonal projector onto $(\Im \Phi_{x_0})^\perp$. 

Then, for $n$ large enough, the extended signed support of $m_0$ on $\Gg_n$ has the form 
\begin{align}
  \extpm^n(m_0)&= \{(x_\nu, \sign(\alpha_{0,\nu}))\}_{1\leq \nu\leq N}\cup \{(x_\nu+\varepsilon_\nu \stepsizen, \sign(\alpha_{0,\nu})\}_{1\leq \nu \leq N}   \\
  \mbox{where } \varepsilon &= \sign\left(\diag(\sign(\alpha_0))\rho\right).
\end{align}
\end{thm}

In the above theorem, observe that $\Phi_{x_0}'^*\Pi\Phi_{x_0}'$ is indeed invertible since $\Gamma_{x_0}$ has full rank.

\begin{cor}\label{cor-lasso-extended}
  Under the hypotheses of Theorem~\ref{thm-lasso-extended}, for $n$ large enough, there exists constants $C^{(1)}_n>0$, $C^{(2)}_n>0$ such that 
  for $\la\leq C^{(1)}_n \min_{1\leq \nu\leq N} |\alpha_{0,\nu}|$, and for all $w\in L^2(\TT)$ such that $\|w\|_2\leq C^{(2)}_n \la$, the solution to~\eqref{eq-thin-lasso} is unique, and reads $m_\la=\sum_{\nu=1}^N (\alpha_{\la,\nu}\delta_{x_{0,\nu}}  + \beta_{\la,\nu}\delta_{x_{0,\nu}+\varepsilon\stepsizen})$, where 
\begin{align*}
  \begin{pmatrix}
    \alpha_\la\\\beta_\la
  \end{pmatrix}&= \begin{pmatrix}
    \alpha_0\\0
  \end{pmatrix}+ \Phi_{\ext_n}^+ w - \la (\Phi_{\ext_n}^* \Phi_{\ext_n})^{-1} \sign\begin{pmatrix}
    \alpha_0\\ \alpha_0
  \end{pmatrix},\\
\qwhereq   \ext_n(m_0)&= \{x_\nu\}_{1\leq \nu\leq N}\cup \{x_\nu+\varepsilon_\nu \stepsizen\}_{1\leq \nu \leq N},\\
  \varepsilon&= \sign\left(\diag(\sign(\alpha_0))\rho\right),\\
  \sign(\alpha_{\la,\nu})&= \sign(\beta_{\la,\nu})=\sign(\alpha_{0,\nu}).
\end{align*}
\end{cor}

\begin{proof}[Proof of Theorem~\ref{thm-lasso-extended}]
  We define a good candidate for $\certthinO$ and using Lemma~\ref{lem-eta0} we prove that it is indeed equal to $\certthinO$ when the grid is thin enough.
  
  	To comply with the notations of Section~\ref{sec-discrete-lasso}, we write 
  	\eq{
  		\sum_{\nu=1}^N\alpha_{0,i}\delta_{x_{0,\nu}}= \sum_{k=0}^{\taillegridn-1} a_{0,k}\delta_{k\stepsizen}, 
	}
	and we let $I\eqdef\enscond{i\in\seg{0}{\taillegridn-1}}{a_{0,i}\neq 0}$. Moreover, for any choice of sign $(\epsilon_i)_{i\in I}\in \{-1,+1\}^{N}$, we set $J\eqdef\bigcup_{i\in I} \{i, i+\varepsilon_i\}$ and $s_J=(s_j)_{j\in J}$ where $s_i\eqdef s_{i+\epsilon_i}\eqdef \sign(a_{0,i})$ for $i\in I$.
  Since $|x_{0,\nu}-x_{0,\nu'}|> {2}{\stepsizen}$ for $\nu'\neq \nu$, we have $\Card J=2\times \Card I=2N$. 
  
Recalling that $\Op=\begin{pmatrix}
\phi(\cdot,0),\ldots \phi(\cdot, (\taillegridn-1)\stepsizen)
  \end{pmatrix}$, we consider the submatrices 
  \begin{align*}
  	\Op_I &\eqdef \begin{pmatrix} \phi(\cdot, i\stepsizen) \end{pmatrix}_{i\in I}
		= \begin{pmatrix} \phi(\cdot, x_{0,1}), & \ldots \phi(\cdot, x_{0,N}) \end{pmatrix} \qandq 
	\Op_{J\setminus I} &\eqdef \begin{pmatrix} \phi(\cdot, (i+\epsilon_i)\stepsizen) \end{pmatrix}_{i\in I}
  \end{align*}
  	so that up to a reordering of the columns $\Op_J=\begin{pmatrix} \Op_I & \Op_{J\setminus I} \end{pmatrix}$. 
	In order to apply Lemma~\ref{lem-eta0}, we shall exhibit a choice of $(\epsilon_i)_{i\in I}$ such that $\Op_J$ has full rank, that $v\eqdef(\Op_J^*\Op_J)^{-1}s_J$ satisfies $\sign(v_j)=-s_j$ for $j\in J\setminus I$ and $\|\Op_{J^c}^* \Op_J v\|_{\infty} <1$ . 

The following Taylor expansion holds for $\Op_{J\setminus I}$ as $n \to \infty$:
  \begin{align*}
    \Op_{J\setminus I} &= A_0 + \stepsizen (B_0 +O(\stepsizen)), \quad    \mbox{ with }   A_0=\Op_I=\Phi_{x_0}\\
      \mbox{and }   
      B_{0} &= \begin{pmatrix}
      	(\partial_2)\phi(\cdot, x_{0,1}) & \ldots & (\partial_2) \phi(\cdot, x_{0,N}) 
    	\end{pmatrix} \diag\left((\epsilon_{i_1}),\ldots (\epsilon_{i_N})\right) \\
 		&=\Phi_{x_0}'\diag\left((\epsilon_{i_1}),\ldots (\epsilon_{i_N})\right).
  \end{align*}

By Lemma~\ref{lem-apx-asymptolasso} in Appendix, the Gram matrix $\Op_J^*\Op_J$ is invertible for $n$ large enough, and
\begin{align*}
  (\Op_J^*\Op_J)^{-1} \begin{pmatrix}
    s_I\\s_I
  \end{pmatrix}=   \frac{1}{\stepsizen} \begin{pmatrix}
(\diag(\varepsilon_{i_1},\ldots, \varepsilon_{i_N}))^{-1} \rho \\
-(\diag(\varepsilon_{i_1},\ldots, \varepsilon_{i_N}))^{-1} \rho
  \end{pmatrix}+O(1),
\end{align*}
% (\Phi_{x_0}'^* \Pi \Phi'_{x_0})^{-1}\Phi_{x_0}'^* \Phi_{x_0}^{+,*}s_I
where $\rho$ is defined in~\eqref{eq-dfn-rho}, 
where $\Pi$ is the orthogonal projector onto $(\Im \Phi_{x_0})^\perp$, and for $\nu\in \seg{1}{N}$, $i_\nu$ refers to the index $i\in I$ such that $i\stepsizen=x_{0,\nu}$.
Therefore, $v_{J\setminus I}$ has the sign of $-\diag(\varepsilon_{i_1},\ldots \varepsilon_{i_N})\rho $, and it is sufficient to choose 
$\varepsilon_{i_\nu} = s_{i_\nu}\times \sign(\rho_\nu)$ to ensure that  $\sign v_{J\setminus I}=-s_{J\setminus I}$ for $n$ large enough. 

With that choice of $\varepsilon$, it remains to prove that  $\|\Op_{J^c}^* \Op_J v\|_{\infty} <1$. Let us write $\tilde{p}_n\eqdef \Op_J v=\Op_J^{+,*}\begin{pmatrix}
    s_I\\s_I
  \end{pmatrix}$. It is equivalent to prove that for $k\in J^c$,  $|\Phi^* \tilde{p}_n (k\stepsizen)|<1$. 
  Using the above Taylor expansion and Lemma~\ref{lem-apx-asymptolasso} in Appendix, we obtain that 
  \begin{align*}
    \lim_{n\to+\infty} \tilde{p}_n &= A_0^{+,*}s_I -\Pi B_0(B_0^*\Pi B_0)^{-1}B_0^*A_0^{+,*}s_I\\
    &=  \Phi_{x_0}^{+,*}\sign(\vecbeurl_{0,\cdot}) - \Pi \Phi_{x_0}' ({\Phi_{x_0}'}^*\Pi \Phi_{x_0}')^{-1}{\Phi_{x_0}'}^*\Phi_{x_0}^{+,*}\sign(\vecbeurl_{0,\cdot})\\
    &= \pvanishing \mbox{ (by~\eqref{eq-vanishing-formula}).}
  \end{align*}

Hence, $\Phi^*\tilde{p}_{n}$ and its derivatives converge to those of $\certvanishing=\certbeurlO$, and there exists $r>0$ such that for all $n$ large enough, for all $1\leq \nu\leq N$, $\Phi^*\tilde{p}_{n}$ is strictly concave (or stricly convex, depending on the sign of $\certbeurlO''(x_{0,\nu})$) in $(x_{0,\nu}-r,x_{0,\nu}+r)$. Hence, for $t\in (x_{0,\nu}-r,x_{0,\nu}+r)\setminus [x_{0,\nu},x_{0,\nu}+\varepsilon_{i(\nu)}\stepsizen]$, we have $|\Phi^* \tilde{p}_n (t)|<1$.
Since by compactness 
\eq{
	\max \enscond{ |\certbeurlO(t)|}{t\in \TT\setminus \bigcup_{\nu=1}^N (x_{0,\nu}-r,x_{0,\nu}+r)}<1
} 
we also see that for $n$ large enough 
\eq{
	\max \enscond{ |\Phi^* \tilde{p}_n(t)|}{t\in \TT\setminus \bigcup_{\nu=1}^N (x_{0,\nu}-r,x_{0,\nu}+r)}<1.
}
As a consequence, for $k\in J^c$,  $|\Phi^* \tilde{p}_n (k\stepsizen)|<1$, and from Lemma~\ref{lem-eta0}, we obtain that $\Phi^* \tilde{p}_n = \certthinO$ and $\bigcup_{\nu=1}^N \{x_{0,\nu},x_{0,\nu}+\varepsilon_{i(\nu)}\stepsizen\}$ is the extended support on $\Gg_n$.
\end{proof}

\subsection{Asymptotics of the constants}

To conclude this section, we examine the decay of the constants $C^{(1)}_n$, $C^{(2)}_n$ in Corollary~\ref{cor-lasso-extended} as $n\to +\infty$. For this we look at the values of $c_1, \ldots, c_5$ given in the proof of Theorem~\ref{thm-stability-dbp}.

By Lemma~\ref{lem-apx-asymptolasso} applied to $\Phi_{\ext_n(m_0)}=\begin{pmatrix} \Phi_{x_0} & \Phi_{x_0}+\stepsizen (\Phi_{x_0}'+O(\stepsizen))
\end{pmatrix}$, we see that 
\begin{align}
	\label{eq-lipsch-lasso-cst-1}
c_{1,n}&\eqdef \norm{R_I \Phi_{\ext_n(m_0)}^+}_{\infty,2}\sim \frac{1}{\stepsizen}\norm{(\Phi_{x_0}'^*\Pi\Phi_{x_0}'^*)^{-1}\Phi_{x_0}'^*\Pi}_{\infty,2},\\
	\label{eq-lipsch-lasso-cst-2}
c_{2,n}&\eqdef \norm{v_I}_\infty=  \normb{R_I(\Phi_{\ext_n(m_0)}^*\Phi_{\ext_n(m_0)})^{-1} \begin{pmatrix}
    s_I\\s_I
  \end{pmatrix}}_\infty \sim \frac{1}{\stepsizen}\norm{\rho}_{\infty},\\ % (\Phi_{x_0}'^*\Pi\Phi_{x_0}'^*)^{-1}\Phi_{x_0}'^*\Phi_{x_0}^{+,*}s_I
  \label{eq-lipsch-lasso-cst-3}
c_{3,n}&\eqdef (\norm{R_K \Phi_{\ext_n(m_0)}^+}_{\infty,2})^{-1}\left(\umin{k \in K} |v_{k}|\right) 
	\sim \frac{\min_{k \in K} |\rho_k|}{\norm{(\Phi_{x_0}'^*\Pi\Phi_{x_0}'^*)^{-1}\Phi_{x_0}'^*\Pi}_{\infty,2}}.
\end{align} % (\Phi_{x_0}'^*\Pi\Phi_{x_0}'^*)^{-1}\Phi_{x_0}'^*\Phi_{x_0}^{+,*}s_I\right
  However, the expressions of $c_4$ and $c_5$ lead to an overly pessimistic bound on the signal-to-noise ratio. Indeed the majorization used in~\eqref{eq-rough-snr} is too rough in this framework: it does not distinguish between neighborhoods of $x_{0,\nu}$'s, where the certificate is close to $1$, and the rest of the domain. 

\begin{prop}\label{prop-asympto-constant-lasso}
  The constants $C^{(1)}_n, C^{(2)}_n$ in Corollary~\ref{cor-lasso-extended} can be chosen as $C^{(1)}_n=O(\stepsizen)$ and $C^{(2)}_n=O(1)$, and one has
  	\eql{\label{eq-lipsch-lasso}
		\normb{ 	
  		\begin{pmatrix} \alpha_\la\\\beta_\la \end{pmatrix}
		- 
		\begin{pmatrix} \alpha_0\\0 \end{pmatrix}
 		}_{\infty} = O\pa{ \frac{w}{\stepsizen}, \frac{\la}{\stepsizen} }.
	}
\end{prop}

\begin{proof}  
	The proof of~\eqref{eq-lipsch-lasso} follows from applying~\eqref{eq-lipsch-lasso-cst-1} and~\eqref{eq-lipsch-lasso-cst-2} in the expression for $\al_\la$ and $\be_\la$ provided by Corollary~\ref{cor-lasso-extended}. 
  Let $\omega\eqdef\Phi^*\Pi w$, where $\Pi$ is the orthogonal projector onto $(\Im \Phi_{x_0})^\perp= \ker \Phi_{x_0}^*$. In order to ensure~\eqref{eq-abstract-strict} we may ensure that :
  \begin{align}
    |\omega(j\stepsizen) + \la\certthinO(j\stepsizen)| -\la<0,
    \label{eq-tight-snr}
  \end{align}
  for all $j\in J^c$ (that is ($j\stepsizen\notin \ext_n(m_0)$). 
  
By the Non-Degenerate Source Condition, there exists $r>0$ such that for all $\nu\in\{1,\ldots, N\}$, 
\begin{align*}
	\forall t\in (x_{0,\nu}-r,x_{0,\nu}+r),\ |\certbeurlO(t)|>0.95 \qandq |(\certbeurlO)''(t)|>\frac{3}{4}|(\certbeurlO)''(x_{0,\nu})|,
\end{align*}
and by compactness $\sup_{\TT\setminus \bigcup_{\nu=1}^N (x_{0,\nu}-r,x_{0,\nu}+r)} |\certbeurlO|<1$.
Since $\certthinO \to \certbeurlO$ (with uniform convergence of all the derivatives), for $n$ large enough,
\eq{
    \forall \nu\in\{1,\ldots, N\}, \ \forall t\in (x_{0,\nu}-r,x_{0,\nu}+r),\ |\certthinO(t)|>0.9 \qandq |(\certthinO)''(t)|>\frac{1}{2}|(\certbeurlO)''(x_{0,\nu})|,
}
(with equality of the signs) and 
\eq{
	 \sup_{\TT\setminus \bigcup_{\nu=1}^N (x_{0,\nu}-r,x_{0,\nu}+r)} |\certthinO|\leq k \eqdef\frac{1}{2} \left(\sup_{\TT\setminus \bigcup_{\nu=1}^N (x_{0,\nu}-r,x_{0,\nu}+r)} |\certbeurlO|+1\right)<1.
}

First, for $j$ such that  $j\stepsizen\in \TT\setminus \bigcup_{\nu=1}^N (x_{0,\nu}-r,x_{0,\nu}+r)$, we see that it is sufficient to assume $\norm{\Phi^*}_{\infty,2}\norm{w}_2< (1-k)\la$ to obtain~\eqref{eq-tight-snr}.

Now, let $\nu\in \{1,\ldots, N\}$ and assume that $\certbeurlO(x_{0,\nu})=1$  (so that $(\certbeurlO)''(x_{0,\nu})<0$) and that $\epsilon_{\nu}=1$, the other cases being similar.
We make the following observation: if a function $f\colon(-r,+r)\rightarrow \RR$ satisfies $f''(t)\leq C$ for some $C<0$ and $f(0)=f(\stepsizen)=0$, then $f(t)\leq \frac{C}{2} t(t-\stepsizen)<0$ for $t\in (-r,0]\cup [\stepsizen,r)$.

Notice that $\omega=\Phi^*\Pi w$ is a $C^2$ function which vanishes on $\ext_n(m_0)$ (hence at $x_{0,\nu}$ and $x_{0,\nu}+\stepsizen$), and that its second derivative is bounded by $\norm{(\Phi'')^*}_{\infty,2}\norm{w}_2$.
  Moreover, $\certthinO(x_{0,\nu})=\certthinO(x_{0,\nu}+\stepsizen)=1$ and $\sup_{(x_{0,\nu}-r,x_{0,\nu}+r)}(\certthinO)''\leq\frac{1}{2}(\certbeurlO)''(x_{0,\nu})<0$. Thus, for $\frac{\norm{w}_2}{\la}< \frac{|(\certbeurlO)''(x_{0,\nu})|}{2\norm{(\Phi'')^*}_{\infty,2}}$, we may apply the observation to $\omega(\cdot -x_{0,\nu})+ \la (\certthinO(\cdot -x_{0,\nu})-1)$ so as to get 
\begin{align*}
  \omega(t)+ \la (\certthinO(t)-1)\leq \left(\norm{(\Phi'')^*}_{\infty,2}\norm{w}_2 +\la \frac{1}{2}(\certbeurlO)''(x_{0,\nu}) \right) (t-x_{0,\nu})(t-x_{0,\nu}-\stepsizen) <0
\end{align*}
for $t\in (x_{0,\nu}-r,x_{0,\nu}]\cup[x_{0,\nu}+\stepsizen,x_{0,\nu}+r)$. 
  
On the other hand, the inequality $-\omega(t) -\la (\certthinO(t)+1)<0$ holds for $\norm{\Phi^*}_{\infty,2}\norm{w}_2 <1.9\la$.  
As a result~\eqref{eq-tight-snr} holds for all $j$ such that $j\stepsizen\in (x_{0,\nu}-r,x_{0,\nu}+r)$, provided that the signal-to-noise ratio satisfies $\frac{\norm{w}_2}{\la}\leq c$, where $c>0$ is a constant which only depends on $\min_{\nu}|(\certbeurlO)''(x_{0,\nu})|$, $\norm{\Phi^*}_{\infty,2}$, $\norm{(\Phi'')^*}_{\infty,2}$ and $\sup_{\TT\setminus \bigcup_{\nu=1}^N (x_{0,\nu}-r,x_{0,\nu}+r)} |\certbeurlO|$.
In other words, including the condition involving $c_{3,n}$, we may choose $C^{(2)}_n=\min(c_{3,n},c) = O(1)$.
\end{proof}

% !TEX root = ../continuous-bp.tex
\section{Abstract analysis of the \lasso with cone constraint}
\label{sec-continuous-abstract}

This section studies a simple variant of the \lasso with cone constraint in an abstract setting. The results stated here shall be useful in Section~\ref{sec-contbp-thin}, since this variant turns out to be the Continuous Basis-Pursuit when the degradation operator is a convolution with an impulse response and its derivative.
Similarly to Section~\ref{sec-discrete-lasso}, we consider in this section observations in an arbitrary Hilbert space $\Hh$.

\subsection{Notations}

Given a parameter $\stepsize>0$, we consider the cone generated by the vectors $(1,\frac{\stepsize}{2})$ and $(1,-\frac{\stepsize}{2})$,
\eql{\label{eq-C-cone}
	\coneh \eqdef \enscond{(c,d) \in \RR \times \RR}{ c\geq 0 \qandq  -c\frac{\stepsize}{2}+|d|\leq 0 }.
 }
 %$e_\ell = (1,\ell) \in \RR^2$, for $\ell \in \{+\frac{\stepsize}{2},-\frac{\stepsize}{2}\}$.
We also define the cone $\coneh^\taillegrid$ as the set of vectors $(\veccont,\shiftcont)\in \RR^\taillegrid\times \RR^\taillegrid$ such that for all $k\in\seg{0}{\taillegrid-1}$ $(\veccont_k,\shiftcont_k)\in \coneh$.

Now, given a vector $(\veccontO,\shiftcontO)\in \coneh^\taillegrid$ (\ie  $\forall k\in \seg{0}{\taillegrid-1}$, $\veccont_{0,k}\geq \frac{2}{\stepsize}|\shiftcont_{0,k}|$), we observe $y_0=\OpU \veccontO + \OpD \shiftcontO$, where $\OpU: \RR^\taillegrid\rightarrow \Hh$ and $\OpD: \RR^\taillegrid\rightarrow \Hh$ are linear operators, or its noisy version $y=y_0+w$ where $w\in \Hh$.
To recover $(\veccontO,\shiftcontO)$ from $y$ or $y_0$, we consider the following reconstruction problems:
\begin{align}
  \umin{(\veccont,\shiftcont)\in \coneh^\taillegrid} \frac{1}{2}\normd{y-\OpU \vecdisc-\OpD\shiftcont}^2 + \la \norm{\veccont}_1, \tag{$\Qq_\la(y)$}\label{eq-abstract-cbpasso}
\end{align}
 and for $\la =0$,
\begin{align}
  \umin{(\veccont,\shiftcont)\in \coneh^\taillegrid} \norm{\veccont}_1 \mbox{ such that } \OpU \veccont+ \OpD \shiftcont=y_0.  \tag{$\Qq_0(y_0)$}\label{eq-abstract-cbp}
\end{align}

Our main focus is on the support recovery properties of~\eqref{eq-abstract-cbpasso}. Precisely, we split the ``support'' of $(\veccont,\shiftcont)\in\coneh^\taillegrid$ into several parts:
\begin{align}
  I~&\eqdef \supp(\veccont) \eqdef \enscond{i \in \seg{0}{\taillegrid-1}}{\veccont_i > 0}\\
&=\Iup\cup\Idown\\
\mbox{where } \Iup &\eqdef \enscond{i\in I}{\veccont_i +\frac{2}{\stepsize}\shiftcont_i >0 },\quad 
	\Idown \eqdef \enscond{i\in I}{\veccont_i -\frac{2}{\stepsize}\shiftcont_i >0 }.
\end{align}
Observe that in general $\Iup\cap \Idown\neq \emptyset$. If $(\veccont_\la, \shiftcont_\la)$ is a solution of~\eqref{eq-abstract-cbpasso}, we say that we have exact support recovery provided that $\Iup(\veccont_\la, \shiftcont_\la)= \Iup(\veccontO, \shiftcontO)$ and $\Idown(\veccont_\la, \shiftcont_\la)= \Idown(\veccontO, \shiftcontO)$. 

\begin{rem}
  The notation $\Iup$, $\Idown$, which might seem a bit obscure at this point, shall become clearer in the next section.
  It turns out that when considering the Continuous Basis-Pursuit on a grid with stepsize $\stepsize>0$, points $i$ in $\Iup$ correspond to Dirac masses which ``tend to be on the right'', that is they do not coincide with the left half-grid point $i\stepsize-\frac{\stepsize}{2}$. Similarly, points in $\Idown$ correspond to Dirac masses which ``tend to be on the left'', as they do not coincide with the right half-grid point $i\stepsize+\frac{\stepsize}{2}$. In fact, if $i\in \Iup\setminus\Idown$, it correponds to a Dirac mass at the right half-grid point: $\delta_{i\stepsize+\frac{\stepsize}{2}}$, and if $i\in \Idown\setminus\Iup$, it correponds to a Dirac mass  at the left half-grid point: $\delta_{i\stepsize-\frac{\stepsize}{2}}$. If $i\in \Iup\cap \Idown$, it correponds to a Dirac mass which may belong ``freely'' to the interval $(i\stepsize-\frac{\stepsize}{2},i\stepsize+\frac{\stepsize}{2})$.  
\end{rem}

\subsection{Parametrization as a positive \lasso }
\label{sec-cbp-another-param}

To characterize the solutions of~\eqref{eq-abstract-cbpasso} and~\eqref{eq-abstract-cbp}, it is convenient to reparametrize the problem as a \lasso with positivity constraint. Indeed, let us write for all $i\in\seg{0}{\taillegrid-1}$,
\begin{align}
	\begin{pmatrix}
\veccont_i\\\shiftcont_i
\end{pmatrix}  \eqdef  \begin{pmatrix}
    1&1\\ \frac{\stepsize}{2}&-\frac{\stepsize}{2}
  \end{pmatrix} \begin{pmatrix}
    \parU_i\\ \parD_i
  \end{pmatrix}
 	\quad\text{or}\quad
 	\begin{pmatrix}
    \parU_i\\\parD_i
  \end{pmatrix}=\frac{1}{2} \begin{pmatrix}
    \veccont_i+\frac{2}{\stepsize} \shiftcont_i\\
    \veccont_i-\frac{2}{\stepsize} \shiftcont_i\\
  \end{pmatrix}. 
\label{eq-cbp-cdv}
 \end{align}
In the following, we define the linear map
\eql{\label{eq-def-H_h}
	H_\stepsize: 
	\begin{pmatrix} \parU \\ \parD \end{pmatrix}
	\longmapsto
	\begin{pmatrix} \veccont \\ \shiftcont \end{pmatrix}
} 
 
It is clear that $(\veccont_i, \shiftcont_i)\in\coneh$ if and only if $\parU_i\geq 0$ and $\parD_i\geq 0$. Moreover, given $(\veccont, \shiftcont)\in \coneh^\taillegrid$,
  \begin{align*}
    I^c=\enscond{i\in\seg{0}{\taillegrid-1}}{(\parU_i,\parD_i)=(0,0)},\quad  \Iup &=\enscond{i\in\seg{0}{\taillegrid-1}}{\parU_i> 0},\\
\qandq    \Idown&=\enscond{i\in\seg{0}{\taillegrid-1}}{\parD_i>0}.
  \end{align*}

%More globally, we shall write this transformation $(\veccont,\shiftcont)=H\begin{pmatrix}
%\parU\\ \parD
%\end{pmatrix}$ where $H\in \RR^{2\taillegrid\times 2\taillegrid }$ and  $(\parU,\parD)\in\RR^{\taillegrid}\times \RR^{\taillegrid}$.

Therefore, Problems~\eqref{eq-abstract-cbpasso} and~\eqref{eq-abstract-cbp} are respectively equivalent to the \lasso and Basis Pursuit with positivity constraint:
\begin{align}
  	\umin{(\parU,\parD)\in (\RR_+)^{\taillegrid}\times(\RR_+)^{\taillegrid}}  
    \la \normb{ \begin{pmatrix}\parU\\ \parD \end{pmatrix} }_1 + \frac{1}{2}\left\|y-\Cop\begin{pmatrix}\parU\\ \parD \end{pmatrix}\right\|_2^2  \tag{$\tilde{\Qq}_\la(y)$}\label{eq-reparam-cbpasso}
\end{align}
\begin{align}
\qandq
  \umin{(\parU,\parD)\in (\RR_+)^{\taillegrid}\times (\RR_+)^{\taillegrid}} 
  	\normb{\begin{pmatrix}\parU\\ \parD \end{pmatrix}}_1 
	\quad\mbox{such that}\quad 
  \Cop\begin{pmatrix}\parU\\ \parD \end{pmatrix}=y_0,  \tag{$\tilde{\Qq}_0(y_0)$}\label{eq-reparam-cbp}
\end{align}
where $\Cop \eqdef \begin{pmatrix}\OpU+\frac{h}{2} \OpD & \OpU-\frac{h}{2}\OpD\end{pmatrix}:\RR^{2\taillegrid}\rightarrow \Hh$.

Observe that there is ``support recovery'' of $(\veccontO,\shiftcontO)$ through~\eqref{eq-abstract-cbpasso} if and only if there is support recovery of $(\parU_0,\parD_0)$ through~\eqref{eq-reparam-cbpasso}.
But precisely, as we shall explain below, the characterization of minimizers and the support recovery properties of the \lasso with positivity constraint~\eqref{eq-reparam-cbpasso} are quite similar to those exposed in Section~\ref{sec-discrete-lasso}. 
 
The regularization term may be written as $J:\RR^\taillegrid\rightarrow \RR\cup\{+\infty\}$, where for all $(\parU,\parD) \in \RR^{\taillegrid} \times \RR^\taillegrid$,
\begin{align*}
  \quad J(\parU,\parD) &\eqdef \left\{\begin{array}{cl}
  	\sum_{i=0}^{\taillegrid-1} (\parU_i+\parD_i) &\mbox{ if } \parU_i\geq 0 \mbox{ and } \parD_i\geq 0 \mbox{ for all } i\in\seg{0}{\taillegrid-1},\\
    +\infty     & \mbox{ otherwise.}\end{array} \right.\\
  &= \sum_{i=0}^{\taillegrid-1} j(\parU_i) + \sum_{i=0}^{\taillegrid-1}j(\parD_i),
\end{align*}
\eq{  
  \qwithq j(x) \eqdef \sup 
  	\enscond{ qx }{q\leq 1} = 
	\left\{
		\begin{array}{cc}x & \mbox{if } x\geq 0,\\+\infty &\mbox{otherwise.}\end{array}
	\right.
}
Hence, the subdifferential of $J$ is the product of the subdifferentials $\partial j(\parU_i)$ and $\partial j(\parD_i)$ for $1\leq i\leq \taillegrid-1$, where
\begin{align*}
\partial j(x)&=\left\{\begin{array}{cc}\{1\} & \mbox{if } x> 0,\\ (-\infty, 1] &\mbox{if } x=0.\end{array}\right.
\end{align*}
That is quite similar to the subdifferential of $|\cdot|$ at $x\in \RR$ which is $-1$, $[-1,1]$ or $1$ if $x<0$, $x=0$ or $x>0$ respectively, and one may adapt all the results of Section~\ref{sec-discrete-lasso} to the \lasso with positivity constraint. It essentially amounts to  replacing the conditions $\normi{\eta}\leq 1$ with $\max \eta \leq 1$ (and similarly for strict inequalities) wherever they appear. We leave the detail to the reader, and in the following, we use those results freely to derive the properties of the \lasso with cone constraint~\eqref{eq-abstract-cbpasso}.

\subsection{Optimality conditions}

Applying the results (or their straightforward adaptations) of Section~\ref{sec-discrete-lasso} to~\eqref{eq-reparam-cbpasso} and~\eqref{eq-reparam-cbp}, then composing by $H_h$, we immediately get the following results.

\begin{prop}
  Let $y\in \Hh$, $(\veccont_\la,\shiftcont_\la) \in \coneh^\taillegrid$, and $I=I(\veccont_\la,\shiftcont_\la)$. Then $(\veccont_\la,\shiftcont_\la)$ is a solution to~\eqref{eq-abstract-cbpasso} if and only if there exists $q_\la\in \Hh$ such that
\begin{align}
  \max\left( (\OpU^*+\frac{h}{2}\OpD^*)q_\la\right) \leq 1, \qandq \max\left( (\OpU^*-\frac{h}{2}\OpD^*)q_\la\right) \leq 1,\\
  (\OpU_{\Iup}^*+\frac{h}{2}\OpD_{\Iup}^*)q_\la =\bun_{\Iup}, \qandq (\OpU_{\Idown}^*-\frac{h}{2}\OpD_{\Idown}^*)q_\la = \bun_{\Idown},\label{eq-optimal-cbp-subdiff}\\
  \la \begin{pmatrix}\OpU^* \\ \OpD^* \end{pmatrix}q_\la + \begin{pmatrix}\OpU^* \\ \OpD^*\end{pmatrix}(\OpU \veccont_\la +\OpD \shiftcont_\la-y) =0. \label{eq-optimal-cbp-lagrange}
\end{align}
Similarly,  $(\veccontO,\shiftcontO)\in \coneh^\taillegrid$ is a solution to~\eqref{eq-abstract-cbp} if and only if $\OpU\veccontO + \OpD\shiftcontO=y_0$ and there exists $q\in \Hh$ such that
 \begin{align}
   \max\left( (\OpU^*+\frac{h}{2}\OpD^*)q\right) \leq 1, \qandq \max\left((\OpU^*-\frac{h}{2}\OpD^*)q\right) \leq 1,\\
  (\OpU_{\Iup}^*+\frac{h}{2}\OpD_{\Iup}^*)q =\bun_{\Iup}, \qandq (\OpU_{\Idown}^*-\frac{h}{2}\OpD_{\Idown}^*)q \leq \bun,
\end{align}
  where $I=I(\veccontO,\shiftcontO)$.
 \label{prop-optim-cbp}
\end{prop}
The corresponding dual problems are given by
\begin{align}
  \inf_{q\in D} &\left\|\frac{y}{\la}-q  \right\|^2 \tag{$\Ee_\la(y)$}\label{eq-abstract-dual-cbpasso}\\
  \sup_{q\in D} &\langle y, q\rangle \tag{$\Ee_0(y)$}\label{eq-abstract-dual-cbp}\\
  \text{where } D&\eqdef\enscond{q\in \Hh}{\!\!\!\!\max_{k\in\seg{0}{\taillegridn-1}}(\OpU^*q)_k+\frac{\stepsize}{2}|(\OpD^*q)_k| \leq 1},
\end{align}

Again, if the inequalities outside the support are strict, it is possible to ensure the uniqueness of the solution.

\begin{prop}
  Under the hypotheses of Proposition~\ref{prop-optim-cbp}, if $\begin{pmatrix}(\OpU+\frac{\stepsize}{2}\OpD)_{\Iup} &(\OpU-\frac{\stepsize}{2}\OpD)_{\Idown}\end{pmatrix}$ has full rank and if $q_\la$ (resp. $q$) satisfies
  \begin{align}
\forall k\in I^c, \quad  (\OpU^*q_\la)_k+ \frac{\stepsize}{2}|(\OpD^*q_\la)_k| <1,\\
\forall i\in \Idown\setminus\Iup,\quad  ((\OpU^* + \frac{\stepsize}{2}\OpD^*)q_\la)_i <1,\\
\forall i\in \Iup\setminus\Idown,\quad  ((\OpU^*- \frac{\stepsize}{2}\OpD^*)q_\la)_i <1,
   \end{align}
   then $(\veccont_\la,\shiftcont_\la)$ (resp. $(\veccontO,\shiftcontO)$) is the unique solution to~\eqref{eq-abstract-cbpasso} (resp.~\eqref{eq-abstract-cbp}).
\label{prop-cbp-strict}
\end{prop}

%Conversely, if $(\parUO,\parDO)$) is the unique solution to~\eqref{eq-reparam-cbp}, then $\begin{pmatrix}\CopU_{\Iup} &\CopD_{\Idown}\end{pmatrix}$ has full rank and there exists $p\in \Hh$ such that~\eqref{eq-cbp-strict} holds.
%\begin{proof}
%  It suffices to prove the result for~\eqref{eq-abstract-cbp} since every solution $u_\la$ of~\eqref{eq-abstract-cbpasso} yields the same value of $\OpT u_\la$ and $u_\la$ is a solution to $\Qq_0(\OpT u_\la)$ (with the certificate $\OpT^*p_\la$).

 % Let $\hat u$, $u^\ast$ be any solution of \eqref{eq-abstract-cbp}, and let $(\hat\certcont,\hat \certshift)$ be a certificate for $\hat u$ such that~\eqref{eq-cbp-strict} holds with $I^c(\hat u)$. From the extremality conditions for problems in duality (see~\cite{ekeland1976convex}), we see that $(\hat\certcont,\hat \certshift)$ is also a dual certificate for $u^\ast$.
%  Hence using $\langle \hat\certcont, \certcont^\ast\rangle + \langle\hat \certshift, \tau^\ast\rangle = \sum_{j=0}^{\taillegrid-1}\certcont^\ast_j$, we obtain that $\certcont^\ast_j=0$ for all $j\in I^c(\hat u)$. 
%We conclude using the injectivity assumption.
%\end{proof}

\subsection{Low noise behavior of C-BP}

The Theorem of Fuchs~\cite{fuchs2004on-sp} for the Lasso (see Remark~\ref{rem-fuchs}) extends to this setting as follows. 

\begin{prop}\label{prop-fuchs-cbp}
  Let $(\veccontO,\shiftcontO)\in \coneh^\taillegrid\setminus\{0\}$ such that 
  \eq{
  	\Coph\eqdef\begin{pmatrix}(\OpU+\frac{\stepsize}{2}\OpD)_{\Iup} &(\OpU-\frac{\stepsize}{2}\OpD)_{\Idown}\end{pmatrix}
	} 
	has full rank, and let 
\begin{align}
	T& \eqdef \min \enscond{\shiftcontOi+\frac{2}{\stepsize}\veccontOi}{i\in \Iup(\shiftcontO,\veccontO)}\cup \enscond{\shiftcontOi-\frac{2}{\stepsize}\veccontOi }{i\in\Idown(\shiftcontO,\veccontO)}.\end{align}
Then there exists constants $C^{(1)}>0,C^{(2)}>0$ such that for $\la \leq C^{(1)} T$, $\|w\|< C^{(2)}\la$, the solution to~$(\Qq_\la(y_0+w))$ is unique, satisfies $\Iup(\shiftcont_\la,\veccont_\la)=
\Iup(\shiftcontO,\veccontO)$, $\Idown(\shiftcont_\la,\veccont_\la)=\Idown(\shiftcontO,\veccontO)$, and it reads:
\begin{align*}
  \begin{pmatrix}
  \veccont_\la \\ \shiftcont_\la
\end{pmatrix} = \begin{pmatrix}
  \veccontO \\ \shiftcontO
\end{pmatrix} +H_\stepsize\Coph^+ w - \la H_\stepsize(\Coph^*\Coph)^{-1}s,
%\mbox{where } G=\begin{pmatrix}(\OpU+\frac{\stepsize}{2}\OpD)_{\Iup} &(\OpU-\frac{\stepsize}{2}\OpD)_{\Idown}\end{pmatrix}.
\end{align*}
where $H_\stepsize$ is defined in~\eqref{eq-def-H_h}.
\end{prop}

In general, the conditions of Proposition~\ref{prop-fuchs-cbp} do not hold, and the support at low noise is strictly larger than $(\Iup(\veccontO,\shiftcontO),\Idown(\veccontO,\shiftcontO))$. This support is governed by the minimal norm certificate.

\begin{defn}[Minimal norm certificate]
  Let $(\veccontO,\shiftcontO)\in \coneh^\taillegrid$. Its minimal norm certificate is $\veccertcontO \eqdef \begin{pmatrix}
    \OpU^* +\frac{\stepsize}{2} \OpD^*\\
    \OpU^* -\frac{\stepsize}{2} \OpD^*
  \end{pmatrix} q_0$ where $q_0$ is the solution to~\eqref{eq-abstract-dual-cbp} with minimal $L^2$ norm.
  The extended support is $\extud (\veccontO,\shiftcontO)=\left(\extu(\veccontO,\shiftcontO),\extd(\veccontO,\shiftcontO)\right)$, where 
\begin{align}
  \extu(\veccontO,\shiftcontO)&=\enscond{j\in\seg{0}{\taillegrid-1}}{((\OpU^* +\frac{\stepsize}{2} \OpD^*)q_0)_j=1  },\\
  \extd(\veccontO,\shiftcontO)&=\enscond{j\in\seg{0}{\taillegrid-1}}{ ((\OpU^* -\frac{\stepsize}{2} \OpD^*)q_0)_j=1}.
\end{align}
\end{defn}

From the optimality conditions, if $(\veccontO,\shiftcontO)$ is a solution of~\eqref{eq-abstract-cbp} then $\Iup \subset \extu(\veccontO,\shiftcontO)$ and $\Idown\subset \extd(\veccontO,\shiftcontO)$
(where $I=I(\veccontO,\shiftcontO)$), and $q_0$ can be characterized as
 \begin{align}
    q_0 &= \uargmin{q\in \Hh} \enscond{\|q\|_2}{\begin{pmatrix}
    \OpU^* +\frac{\stepsize}{2} \OpD^*\\
    \OpU^* -\frac{\stepsize}{2} \OpD^*
  \end{pmatrix} q\in\partial J(\parUO,\parDO)}.
 \end{align}

\begin{lem}\label{lem-mu0}
  Let $\Jup,\Jdown \subset \seg{0}{\taillegrid-1}$, and $(\veccontO,\shiftcontO)\in\coneh^\taillegrid$. Assume that $(\Iup,\Idown)\eqdef (\Iup(\veccontO,\shiftcontO),\Idown(\veccontO,\shiftcontO))$  is such that 
  	\eq{
	 	\Iup\subset \Jup,  \quad
		\Idown\subset \Jdown
		\qandq \Copt\eqdef\begin{pmatrix}(\OpU+\frac{\stepsize}{2}\OpD)_{\Jup} &(\OpU-\frac{\stepsize}{2}\OpD)_{\Jdown}\end{pmatrix}
	} 
	has full rank. Define $\begin{pmatrix} u_{\Jup}\\v_{\Jdown} \end{pmatrix} \eqdef -(\Copt^*\Copt)^{-1}s$ where $s \eqdef \begin{pmatrix}1\\\vdots\\ 1\end{pmatrix}\in \RR^{|\Jup|+|\Jdown|}$. Then $(\Jup,\Jdown)$ is the extended support of $(\veccontO,\shiftcontO)$ if and only if the following two conditions hold:
\begin{itemize}
  \item for all $j\in \Jup\setminus \Iup$, $u_j \geq 0$, and for all $j\in \Jdown\setminus \Idown$, $v_j\geq 0$.
  \item $\max \left[\begin{pmatrix}(\OpU+\frac{\stepsize}{2}\OpD)_{\Jup^c}^* \\(\OpU-\frac{\stepsize}{2}\OpD)_{\Jdown^c}^* \end{pmatrix} \Copt \begin{pmatrix} u_{\Jup}\\v_{\Jdown} \end{pmatrix}\right]<1$.
\end{itemize}
Moreover, in that case, the minimal norm certificate is given by 
\eq{
	\veccertcontO =-\Cop^*(\Copt^*\Copt)^{-1}s.
}
\end{lem}

The proof is identical to the one of Lemma~\ref{lem-eta0}, therefore we omit it. We are now in position to describe the behavior of~\eqref{eq-abstract-cbpasso} at low  noise in the generic case:

\begin{thm}\label{thm-abstract-cbp}
  Let $(\veccontO,\shiftcontO)\in (\RR_+)^2\setminus\{0\}$ be an identifiable signal, $(\Jup,\Jdown)=\extud(\veccontO,\shiftcontO)$ such that
   $\Copt\eqdef\begin{pmatrix}(\OpU+\frac{\stepsize}{2}\OpD)_{\Jup} &(\OpU-\frac{\stepsize}{2}\OpD)_{\Jdown}\end{pmatrix}$ has full rank. Let  
   $\begin{pmatrix} u_{\Jup}\\v_{\Jdown} \end{pmatrix}\eqdef-(\Copt^*\Copt)^{-1}s$ where $s \eqdef \begin{pmatrix}1\\\vdots\\ 1\end{pmatrix}\in \RR^{|\Jup|+|\Jdown|}$,
  and assume that for all $j\in \Jup\setminus \Iup$, $u_j>0$, and that for all $j\in \Jdown\setminus \Idown$, $v_j> 0$.

  Then, there exists constants $C^{(1)}>0$, $C^{(2)}>0$ such that for 
\begin{align}
  \la &\leq C^{(1)} \min \enscond{\shiftcontOi+\frac{2}{\stepsize}\veccontOi}{i\in \Iup(\veccontO,\shiftcontO)}\cup \enscond{\shiftcontOi-\frac{2}{\stepsize}\veccontOi }{i\in\Idown(\veccontO,\shiftcontO)}\end{align}
and $\norm{w}_{2}\leq C^{(2)}\la$, the solution $(\veccont_\la,\shiftcont_\la)$ to~\eqref{eq-abstract-cbpasso} is unique, $\Iup(\veccont_\la,\shiftcont_\la)=\Jup$, $\Idown(\veccont_\la,\shiftcont_\la)=\Jdown$, and it reads
\begin{align*}
  \begin{pmatrix}
  \veccont_\la \\ \shiftcont_\la
\end{pmatrix} = \begin{pmatrix}
  \veccontO \\ \shiftcontO
\end{pmatrix} +H_\stepsize\Copt^+ w - \la H_\stepsize(\Copt^*\Copt)^{-1}s,
\end{align*}
where $H_\stepsize$ is defined in~\eqref{eq-def-H_h}.
\end{thm}

% !TEX root = ../continuous-bp.tex
\section{Continuous-Basis Pursuit on thin grids}
\label{sec-contbp-thin}

Facing the same inverse problem as in Section~\ref{sec-discbp-thin}, but this time assuming that each $\alpha_\nu$ ($1\leq \nu \leq N$) is positive, we aim at recovering $m_0$ using the Continuous Basis-Pursuit proposed in~\cite{Ekanadham-CBP}. Given a grid $\Gg_n$ as in Section~\ref{sec-discbp-thin}, the goal is to reconstruct a measure $m=\sum_{i=0}^{\taillegridn-1} \veccont_{i} \delta_{i\stepsizen+t_i}$ where $t_i\in[-\frac{\stepsizen}{2},\frac{\stepsizen}{2}]$ which estimates $m_0$. Applying a Taylor expansion and setting $\shiftcont_i= t_i\veccont_i$, the authors of~\cite{Ekanadham-CBP} are led to solve
\begin{align}
  \umin{(\veccont,\shiftcont)\in \conehn^\taillegridn} \frac{1}{2}\norm{y-\Phi_{\Gg_n} \veccont-\Phi_{\Gg_n}'\shiftcont}^2 + \la \norm{\veccont}_1 \tag{$\Qq^n_\la(y)$}\label{eq-thin-cbpasso}\\
  \umin{(\veccont,\shiftcont)\in \conehn^\taillegridn} \norm{\veccont}_1 \mbox{ such that } \Phi_{\Gg_n} \veccont+ \Phi'_{\Gg_n} \shiftcont=y_0.  \tag{$\Qq^n_0(y_0)$}\label{eq-thin-cbp}
\end{align}
which are particular instances of~\eqref{eq-abstract-cbpasso} and~\eqref{eq-abstract-cbp}.
The dual problems are respectively:
\begin{align}
  \inf_{q\in D^n} &\left\|\frac{y}{\la}-q  \right\|^2 \tag{$\Ee_\la^n(y)$}\label{eq-thin-dual-cbpasso}\\
  \sup_{q\in D^n} &\langle y_0, q\rangle \tag{$\Ee_0^n(y_0)$}\label{eq-thin-dual-cbp}
\end{align}
\eql{\label{eq-defn-Dn}
	\qwhereq 
	D^n\eqdef\enscond{q\in L^2(\TT)}{\!\!\!\!\max_{k\in\seg{0}{\taillegridn-1}}(\Phi^*q)(k\stepsizen)+\frac{\stepsizen}{2}|(\Phi^*q)'(k\stepsizen)| \leq 1},
}
To study the behavior of the solutions to these problems as $n$ increases, we aim at applying the results of the previous section, and in particular Lemma~\ref{lem-mu0}, in the setting where $\Hh = L^2(\TT)$ and $(A,B) = (\Phi_{\Gg_n},\Phi_{\Gg_n}')$.

\subsection{The positive Beurling \lasso}

The situation with the continuous basis pursuit on thin grids is quite similar to the situation of the \lasso.
Still, the non-negativity constraint on the components $(a_k)_{0\leq k\leq \taillegridn}$ passes to the limit, and the appropriate limit model is the \textit{positive Beurling \lasso},
\begin{align}
  	\umin{m\in \Mm^+(\TT)} \frac{1}{2}\norm{y-\Phi m}^2 + \la m(\TT), 
	\tag{$\Qq_\la^\infty(y)$}\label{eq-beurl-cbpasso}\\
 	\qandq
	\umin{m\in \Mm^+(\TT)} m(\TT) \mbox{ such that } \Phi m=y_0,  \tag{$\Qq_0^\infty(y_0)$}\label{eq-beurl-cbp}
\end{align}
where $\Mm^+(\TT)$ refers to the space of positive Radon measures. The indicator function of positive measures plus the total mass may be encoded in the quantity:
\begin{align}
  m(\TT)+ \iota_{\Mm^+(\TT)}(m)=\sup\enscond{\int_\TT \psi(t)\d m(t)}{\psi\in C(\TT)\mbox{ and } \max_{t\in\TT} \psi(t)\leq 1}.
\end{align}
As a result, the characterization of optimality, the notions of minimal norm certificates and extended support may be adapted from Section~\ref{subsec-beurl-lasso} in a straightforward manner, replacing condition $\|\eta\|_{\infty}\leq 1$ by $\sup_{t\in\TT} \mu(t)\leq 1$ where $\mu=\Phi^*q$ for $q\in L^2(\TT)$.
For instance, up to the addition of a constant, the dual problems to~\eqref{eq-beurl-cbpasso} and~\eqref{eq-beurl-cbp} are respectively:
\begin{align}
  \inf_{q\in D^\infty} \left\|\frac{y}{\la}-p  \right\|^2 \tag{$\Ee_\la^\infty(y)$}\label{eq-beurl-dual-cbpasso}\\
  \sup_{q\in D^\infty} \langle y_0, q\rangle \tag{$\Ee_0^\infty(y_0)$}\label{eq-beurl-dual-cbp}
\end{align}
\eql{	\label{eq-defn-Dinf}
  \qwhereq D^\infty \eqdef \enscond{q\in L^2(\TT)}{\max_{t\in\TT} (\Phi^*q)(t)\leq 1}. 
}

\subsection{The limit problem for thin grids}

To consider the limit of~\eqref{eq-thin-cbpasso}, let us recall that we obtain a measure from the vector $(\veccont,\shiftcont)\in \conehn^\taillegridn$ by setting
\begin{align}
  m=\sum_{i=0}^{\taillegridn-1} \veccont_{i} \delta_{i\stepsizen+\shiftcont_i/\veccont_i } \label{eq-cbp-measure}
\end{align}
with the convention that $\shiftcont_i/\veccont_i=0$ if $\veccont_i=0$. It should be noticed that $\shiftcont_i/\veccont_i\in [-\frac{\stepsizen}{2},\frac{\stepsizen}{2}]$.

We rely again on the notion on $\Gamma$-convergence to express the convergence of~\eqref{eq-thin-cbpasso} towards~\eqref{eq-beurl-cbpasso}. As before, we may restrict the problems to $X_+ \eqdef \enscond{m\in \Mm^+(\TT)}{\la |m|(\TT)\leq \frac{1}{2}\norm{y}^2}$ which is metrizable for the weak* topology.

\begin{defn}
  We say that the Problem~\eqref{eq-thin-cbpasso} $\Gamma$-converges towards Problem~\eqref{eq-beurl-cbpasso} if, for all $m\in X_+$, the following conditions hold
  \begin{itemize}
    \item{(Liminf inequality)} for any sequence of measures $(m^n)_{n\in\NN}\in X_+^\NN$ of the form~\eqref{eq-cbp-measure} with $(\veccont^{(n)},\shiftcont^{(n)})\in \conehn^\taillegridn$ such that $m^n$ weakly* converges towards $m$, 
      \begin{align*}
        \liminf_{n\to +\infty}\left(\la \|\veccont^{(n)}\|_1+\frac{1}{2}\norm{\Phi_{\Gg_n}\veccont^{(n)}+\Phi'_{\Gg_n}\shiftcont^{(n)}-y}^2\right) \geq \la m(\TT) + \frac{1}{2}\norm{\Phi m-y}^2.
      \end{align*}
    \item{(Limsup inequality)} there exists a sequence of measures $(m^n)_{n\in \NN}\in X_+^\NN$ of the form~\eqref{eq-cbp-measure} with $(\veccont^{(n)},\shiftcont^{(n)})\in \conehn^\taillegridn$ such that $m^n$ weakly* converges towards $m$  and
      \begin{align*}
        \limsup_{n\to +\infty}\left(\la \|\veccont^{(n)}\|_1+\frac{1}{2}\norm{\Phi_{\Gg_n}\veccont^{(n)}+\Phi'_{\Gg_n}\shiftcont^{(n)}-y}^2\right)\leq \la m(\TT) + \frac{1}{2}\norm{\Phi m-y}^2.
      \end{align*}
  \end{itemize}
  \label{def-gammacv-cbp}
\end{defn}

\begin{prop}
  The Problem~\eqref{eq-thin-cbpasso} $\Gamma$-converges towards \eqref{eq-beurl-cbpasso}, and 
\begin{align}
  \lim_{n\to +\infty} \inf \eqref{eq-thin-cbpasso} = \inf \eqref{eq-beurl-cbpasso}.\label{eq-cvinf-cbpasso}
\end{align}
Each sequence $(m^n_\la)_{n\in\NN}$ such that $m^n_\la$ is a minimizer of~\eqref{eq-thin-cbpasso} has accumulation points (for the weak*) topology, and each of these accumulation points is a minimizer of~\eqref{eq-beurl-cbpasso}.
\label{prop-cbpthin-convergence}
\end{prop}
In particular, if the solution $m^\infty_\la$ to $\eqref{eq-beurl-cbpasso}$ is unique, the whole sequence $(m^n_\la)_{n\in\NN}$ converges towards $m^\infty_\la$.

\begin{proof}
  The proof is the same as for Proposition~\ref{prop-gammacv-lasso} with minor adaptations,  observing that $\|\veccont^{(n)}\|_1=m_n(\TT)$.
  For the liminf inequality, let $(m^n)_{n\in\NN}$ be of the form~\eqref{eq-cbp-measure} which weakly* converges towards $m$. We notice that $\Phi'_{\Gg_n}\shiftcont^{(n)}= \Phi' (\sum_{i=0}^{\taillegridn-1}\shiftcont^{(n)}_i\delta_{i\stepsizen})$, and
\begin{align*}
  \left|\sum_{i=0}^{\taillegridn-1}\shiftcont^{(n)}_i\delta_{i\stepsizen}\right|(\TT) &\leq \frac{\stepsizen}{2}\left(\sum_{i=0}^{\taillegridn-1}\veccont^{(n)}_i\right)  
    \leq \frac{\stepsizen}{2\la}\left(\frac{1}{2}\norm{y}^2+1\right)\to 0,
%    \stackrel{*}{\rightharpoonup} 0,
  \end{align*}
so that $\Phi'_{\Gg_n}\shiftcont^{(n)}$ (strongly) converges towards $0$ in $L^2(\TT)$.
Moreover, $\Phi_{\Gg_n}\veccont^{(n)}= \Phi (\sum_{i=0}^{\taillegridn-1}\veccont^{(n)}_i\delta_{i\stepsizen})$ and for all $\psi\in C(\TT)$, 
  \begin{align*}
    \left|\left\langle\sum_{i=0}^{\taillegridn-1} \veccont_{i}^{(n)} \delta_{i\stepsizen+\shiftcont_i/\veccont_i }- \sum_{i=0}^{\taillegridn-1} \veccont_{i} \delta_{i\stepsizen}, \psi \right\rangle\right| &= \left|\sum_{i=0}^{\taillegridn-1} \veccont_{i}^{(n)} (\psi(i\stepsizen+\shiftcont_i/\veccont_i)-\psi(i\stepsizen))\right|\\
    &\leq \sum_{i=0}^{\taillegridn-1} \veccont_{i}^{(n)} \omega_\psi\left(\frac{\stepsizen}{2} \right) \to 0
  \end{align*}
  where $\omega_\psi: t\mapsto \sup_{|x'-x|\leq t}|\psi(x)-\psi(x')|$ is the modulus of continuity of $\psi$. As a result, $\sum_{i=0}^{\taillegridn-1} \veccont_{i}^{(n)} \delta_{i\stepsizen}-m^n \stackrel{*}{\rightharpoonup} 0$ and $\sum_{i=0}^{\taillegridn-1} \veccont_{i}^{(n)} \delta_{i\stepsizen}$ weakly* converges to $m$. Hence, $\Phi_{\Gg_n}\veccont^{(n)}$ weakly converges towards $\Phi m$ in $L^2(\TT)$.
To sum up, $\Phi_{\Gg_n}\veccont^{(n)}+\Phi'_{\Gg_n}\shiftcont^{(n)}-y$ weakly converges towards $\Phi m-y$ and we conclude as before.
  
For the limsup inequality, the only difference is in the construction of $m^n$ for the limsup inequality:
it is sufficient to choose $\veccont^{(n)}_k=m([k\stepsizen,(k+1)\stepsizen))$ and $\shiftcont^{(n)}_k=0$ for all $k\in \seg{0}{\taillegridn-1}$.
\end{proof}

\subsection{Asymptotics of the support: generalities}

Though Proposition~\ref{prop-cbpthin-convergence} states the convergence of the solutions of~\eqref{eq-thin-cbpasso} towards those of~\eqref{eq-beurl-cbpasso}, it does not describe the supports of the solutions. We now study the convergence of those supports using dual certificates and the optimality conditions (Proposition~\ref{prop-optim-cbp}). In this context, a dual certificate is determined by a function $\mu=\Phi^*q\in\Cont(\TT)$ where $q\in L^2(\TT)$, and 
\begin{align*}
  \Iup=\enscond{i\in\seg{0}{\taillegridn-1}}{\shiftcont_i>-\frac{\stepsizen}{2}a_i} \subset \enscond{i\in\seg{0}{\taillegridn-1}}{\left(\mu+\frac{\stepsizen}{2}\mu'\right)(i\stepsizen)=1},\\
  \Idown=\enscond{i\in\seg{0}{\taillegridn-1}}{\shiftcont_i<\frac{\stepsizen}{2}a_i} \subset \enscond{i\in\seg{0}{\taillegridn-1}}{\left(\mu-\frac{\stepsizen}{2}\mu'\right)(i\stepsizen)=1}.
\end{align*}
To sum up, we shall exploit the following observations
\begin{itemize}
  \item if $\left(\mu+\frac{\stepsizen}{2}\mu'\right)(i\stepsizen)=1$ but $\left(\mu-\frac{\stepsizen}{2}\mu'\right)(i\stepsizen)<1$, a spike may appear at $i\stepsizen+\frac{\stepsizen}{2}$,
  \item if $\left(\mu-\frac{\stepsizen}{2}\mu'\right)(i\stepsizen)=1$ but $\left(\mu+\frac{\stepsizen}{2}\mu'\right)(i\stepsizen)<1$, a spike may appear at $i\stepsizen-\frac{\stepsizen}{2}$,
  \item if $\left(\mu+\frac{\stepsizen}{2}\mu'\right)(i\stepsizen)=1$ and $\left(\mu-\frac{\stepsizen}{2}\mu'\right)(i\stepsizen)=1$, a spike may appear anywhere in the interval $[i\stepsizen-\frac{\stepsizen}{2},i\stepsizen+\frac{\stepsizen}{2}]$.
\end{itemize}

The following lemma is central in our analysis. We consider a sequence of functions $(\mu^n)_{n\in\NN}$ and for $0<r<\frac{1}{2} \min_{\nu\neq \nu'} |x_\nu - x_{\nu'}|$, $\nu\in\{1,\ldots,N\}$, we study:
\begin{align*}
    \Sright(r)\eqdef\enscond{t\in \Gg_n\cap(x_\nu-r,x_\nu+r)}{\left(\mu^n+\frac{\stepsizen}{2}{\mu^n}'\right)(t)=1},\\
    \Sleft(r)\eqdef\enscond{t\in \Gg_n\cap(x_\nu-r,x_\nu+r)}{\left(\mu^n-\frac{\stepsizen}{2}{\mu^n}'\right)(t)=1}.
\end{align*}

\begin{lem}
  Let $(x_1,\ldots, x_N)\in \TT^N$ pairwise distinct, and let $\{\mu^n\}_{n\in\NN}\in (\Cont^3(\TT))^\NN$ be a sequence of functions which converges uniformly towards some $\mu^\infty$ (and similarly for the derivatives) such that for all $\nu\in\{1,\ldots, N\}$, $\mu^\infty(x_\nu)=1$ and for all $t\in \TT\setminus\{x_1,\ldots,x_N\}$, $\mu^\infty(t)<1$.

\begin{enumerate}
  \item Then 
\begin{align}
  \limsup_{n\to +\infty} \enscond{t\in \Gg_n}{\mu^n(t)+\frac{\stepsizen}{2}|{\mu^n}'(t)|=1}\subset \{x_1,\ldots, x_N\}.\label{eq-cbp-limsup}
\end{align} In particular 
there exists $n_0\in\NN$ such that for $n\geq n_0$
 \begin{align*}
   \enscond{t\in \Gg_n}{\mu^n(t)+\frac{\stepsizen}{2}|{\mu^n}'(t)|=1} = \bigcup_{\nu=1}^N  \left(\Sright(r)\cup  \Sleft(r)\right) \subset \bigcup_{\nu=1}^N(x_\nu-r,x_\nu+r).
 \end{align*}
\end{enumerate}
Assume moreover that for all $n\in\NN$ and all $t\in\Gg_n$, $\mu^n(t)+\frac{\stepsizen}{2}|{\mu^n}'(t)| \leq 1$. For each $\nu \in\{1,\ldots,N\}$:
\begin{enumerate}
\setcounter{enumi}{1}
\item\label{item-deriv2} If $(\mu^{\infty})''(x_\nu)\neq 0$, then there exists $n_0\in\NN$ such that for $n\geq n_0$, each set $\Sright(r)$ and $\Sleft(r)$ is of the form
$\emptyset$, $\{i\stepsizen\}$, or $\{i\stepsizen,(i+1)\stepsizen\}$, and if both sets are nonempty:
\begin{align*}
  \max \Sright(r) \leq \min \Sleft(r).
\end{align*}

\item If $(\mu^\infty)^{(3)}(x_\nu)\neq 0$, then there exists $n_0\in\NN$ such that for each $n\geq n_0$, $\Sright(r)=\emptyset$ or $\Sleft(r)=\emptyset$.
\item\label{item-deriv4} If $(\mu^\infty)^{(4)}(x_\nu)\neq 0$, the set of $n\in\NN$ such that $\Sright(r)=\{(i_n-1)\stepsizen,i_n\stepsizen\}$ and $\Sleft (r)= \{i_n\stepsizen,(i_n+1)\stepsizen\}$ (with the same $i_n\in\seg{0}{\taillegridn-1}$) is finite.
\end{enumerate}
\label{lem-cbp-dualspikes}
\end{lem}
\begin{proof}
   Observe that both $\mu^{n} + \frac{\stepsizen}{2}{\mu^{n}}'$ and $\mu^{n} - \frac{\stepsizen}{2}{\mu^{n}}'$ converge uniformly towards $\mu^{\infty}$ as $n\to +\infty$ (and similarly for the derivatives).
   \begin{enumerate}
     \item  For all $\tilde{r}\in(0,r)$, by compactness, $\sup \enscond{\mu^\infty(t) }{t\in\TT\setminus\bigcup (x_\nu-\tilde{r},x_\nu+\tilde{r})}<1$. Thus by uniform convergence there exists  $n_0\in \NN$ such that for all $n\geq n_0$, $(\mu^{n} \pm \frac{\stepsizen}{2}{\mu^{n}}')<1$ on $\TT\setminus\bigcup_{\nu=1}^N (x_\nu-\tilde{r},x_\nu+\tilde{r})$, and the first claim is proved.

     \item If moreover $(\mu^{\infty})''(x_\nu)\neq 0$, it is in fact negative. Choosing $\tilde{r}\in(0,r)$ small enough and then $n$ large enough, we may assume that ${\mu^{n}}'' <-k_0$ in $(x_\nu-\tilde{r},x_\nu+\tilde{r})$, for some $k_0>0$, and by~\eqref{eq-cbp-limsup} that $\Sright(r)\cup\Sleft(r)\subset (x_{\nu}-\tilde{r},x_{\nu}+\tilde{r})$. By uniform convergence, ${\mu^{n}}''+\frac{\stepsizen}{2}|\mu^{n (3)}|<-\frac{k_0}{2}$ for $n$ large enough, so that both functions $\mu^{n} + \frac{\stepsizen}{2}{\mu^{n}}'$ and $\mu^{n} - \frac{\stepsizen}{2}{\mu^{n}}'$ are strictly concave in $(x_\nu-\tilde{r},x_\nu+\tilde{r})$. This implies that $\Sright(r)$ (resp. $\Sleft(r)$) is of the form $\emptyset$, $\{i\stepsizen\}$, or $\{i\stepsizen,(i+1)\stepsizen\}$.

Observe also that $\mu^{n} + \frac{\stepsizen}{2}{\mu^{n}}' -(\mu^{n} - \frac{\stepsizen}{2}{\mu^{n}}')= \stepsizen{\mu^n}'$.
Since the function ${\mu^{n}}'$ is strictly decreasing in $(x_\nu-\tilde{r},x_\nu+\tilde{r})$, it vanishes at most once. If $\Sright(r)\neq\emptyset$ and $\Sleft(r)\neq\emptyset$, it must change sign in $(x_\nu-\tilde{r},x_\nu+\tilde{r})$ and thus it vanishes exactly once, at some $\xi\in (x_\nu-\tilde{r},x_\nu+\tilde{r})$. 
Then for $t\in (x_\nu-\tilde{r},\xi)$, 
\begin{align*}
 (\mu^{n} - \frac{\stepsizen}{2}{\mu^{n}}')(t)=(\mu^{n} + \frac{\stepsizen}{2}{\mu^{n}}')(t)- \stepsizen (\mu^{n})'(t)\leq 1 - \stepsizen (\mu^{n})'(t)<1
\end{align*}
so that $\min \Sleft(r)\geq\xi$. Similarly $\max \Sright(r)\leq \xi$.

\item By contradiction, assume that the set of $n'\in\NN$ such that $\Srightp(r)\neq \emptyset$ and $\Sleftp(r)\neq \emptyset$ is infinite. We may extract a subsequence $n=n'(m)$ such that there exists $i_n,j_n\in\seg{0}{\taillegridn-1}$ (denoted hereafter $i,j$) with $i\stepsizen\in \Sright(r_m)$, $j\stepsizen\in \Sleft$.
 Combining the Taylor expansions of $\mu^n$ and $({\mu^n})'$ around $i\stepsizen$ (resp. $j\stepsizen$), we get
\begin{align*}
  1&\geq \mu^n((i+1)\stepsizen)-\frac{\stepsizen}{2}{\mu^n}'((i+1)\stepsizen)\\
  & =  \underbrace{\mu^n(i\stepsizen) +\stepsizen{\mu^n}'(i\stepsizen)(1-\frac{1}{2} )}_{=1} +\stepsizen^2{\mu^n}''(i\stepsizen)\underbrace{\left(\frac{1}{2!} -\frac{1}{2}\right)}_{=0} + \stepsizen^3 {\mu^n}^{(3)}(i\stepsizen)\al_3\\
  & \qquad\qquad+ \stepsizen^4\int_0^1 {\mu^n}^{(4)}(i\stepsizen+t\stepsizen)\left(\frac{(1-t)^3}{3!} -\frac{(1-t)^2}{2!\times 2}\right)\d t, \mbox{ and }\\
  1&\geq \mu^n((j-1)\stepsizen)+\frac{\stepsizen}{2}{\mu^n}'((j-1)\stepsizen)\\
  & =  \underbrace{\mu^n(j\stepsizen) -\stepsizen{\mu^n}'(j\stepsizen)(1-\frac{1}{2} )}_{=1} +\stepsizen^2{\mu^n}''(j\stepsizen)\underbrace{\left(\frac{1}{2!} -\frac{1}{2}\right)}_{=0} -\stepsizen^3 {\mu^n}^{(3)}(j\stepsizen)\al_3\\
  & \qquad\qquad+ \stepsizen^4\int_0^1 {\mu^n}^{(4)}(j\stepsizen-t\stepsizen)\left(\frac{(1-t)^3}{3!} -\frac{(1-t)^2}{2!\times 2}\right)\d t
\end{align*}
where $\al_k$ is defined in~\eqref{eq?defn-alphak}.
Now, let $n\to +\infty$. By~\eqref{eq-cbp-limsup}, $i\stepsizen\to x_{\nu}$ and $j\stepsizen\to x_\nu$, and using the uniform convergence of $(\mu^n)^{(k)}$ towards $ (\mu^\infty)^{(k)}$, dividing by $\stepsizen^3$, we obtain respectively $0\geq -(\mu^\infty)^{(3)}(x_\nu)\times\frac{1}{12}$ and $0\geq (\mu^\infty)^{(3)}(x_\nu)\times \frac{1}{12}$, thus $(\mu^\infty)^{(3)}(x_\nu)=0$.

\item Assume by contradiction, that the mentioned set is infinite. For such $n$, a Taylor expansion at $i\stepsizen$ yields (we write $i$ for $i_n$):
\begin{align*}
1&=\mu^n((i+1)\stepsizen)-\frac{\stepsizen}{2}{\mu^n}'((i+1)\stepsizen)\\
& =  \underbrace{\mu^n(i\stepsizen) +\frac{\stepsizen}{2}{\mu^n}'(i\stepsizen)}_{=1} + \gamma_3\stepsizen^3 {\mu^n}^{(3)}(i\stepsizen)+ \gamma_4\stepsizen^4 {\mu^n}^{(4)}(i\stepsizen)\\
&\quad\quad+\stepsizen^5\int_0^1 {\mu^n}^{(5)}(i\stepsizen+t\stepsizen)\left(\frac{(1-t)^4}{4!} -\frac{(1-t)^3}{3!\times 2}\right)\d t, \mbox{ and }\\
1&= \mu^n((i-1)\stepsizen)+\frac{\stepsizen}{2}{\mu^n}'((i-1)\stepsizen)\\
& =  \underbrace{\mu^n(i\stepsizen) -\frac{\stepsizen}{2}{\mu^n}'(i\stepsizen)}_{=1} - \gamma_3\stepsizen^3 {\mu^n}^{(3)}(i\stepsizen)+ \gamma_4\stepsizen^4 {\mu^n}^{(4)}(i\stepsizen)\\
&\quad \quad+  \stepsizen^5\int_0^1 {\mu^n}^{(5)}(i\stepsizen+t\stepsizen)\left(\frac{(1-t)^4}{4!} -\frac{(1-t)^3}{3!\times 2}\right)\d t,
\end{align*}
with $\gamma_k=\frac{1}{k!}- \frac{1}{(k-1)!\times 2}$. Summing both equalities, dividing by $\stepsizen^4$ and taking the limit $n\to+\infty$ yields $(\mu^\infty)^{(4)}(x_\nu)=0$.
   \end{enumerate}
\end{proof}

This other lemma focusses on the limit of the sets $D^n$ defined in~\eqref{eq-defn-Dn}.

\begin{lem}
As $n\to+\infty$, the sets $D^n$ converge towards $D^\infty$ defined in~\eqref{eq-defn-Dinf} (in the sense of set convergence).
\label{lem-cbp-cvset}
\end{lem}
\begin{proof}
  We observe that $E^n \subset D^n\subset F^n$, where
  \begin{align*}
    E^n&\eqdef \enscond{q\in L^2(\TT)}{\max_{t\in \TT}(\Phi^*q)(t)+\frac{\stepsizen}{2}|(\Phi^*q)'(t)| \leq 1},\\
    F^n&\eqdef \enscond{q\in L^2(\TT)}{\max_{k\in\seg{0}{\taillegridn-1}}\Phi^*q(k\stepsizen)\leq 1}
  \end{align*}
so that it suffices to prove that $E^n$ and $F^n$ converge towards $D^\infty$. 
On the one hand, it is clear that $D^\infty=\bigcap_{n\in\NN} F^n$, and the sequence $F^n$ is non-increasing.
On the other hand, it is possible to check that $D^\infty=\overline{\bigcup_{n\in \NN}E^{n}}$, and the sequence $E^n$ is non-decreasing.
As a consequence, the claimed set convergences hold.
\end{proof}

\subsection{Asymptotics of the support for fixed $\la>0$}

Let us recall that the dual problem to~\eqref{eq-thin-cbpasso} is the projection onto the closed convex set
\begin{align*}
  D^n &\eqdef\enscond{q\in L^2(\TT)}{(\Phi^* q)(i\stepsizen) + \frac{\stepsizen}{2}|(\Phi^*q)'(i\stepsizen)|\leq 1}.
\end{align*}

Since the set convergence of $D^n$ (see Lemma~\ref{lem-cbp-cvset}) implies the convergence of the projections onto~$D^n$ (see~\cite{rockafellarwets}, or~\cite{2013-duval-sparsespikes} for a direct proof in a similar context), we obtain:

\begin{prop}\label{prop-cbp-cvdual}
  Let $q_\la^n$ (resp. $q_\la^\infty$) be a solution of~\eqref{eq-thin-dual-cbpasso} (resp.~\eqref{eq-beurl-dual-cbpasso}), and $\mu_\la^n=\Phi^*q_\la^n$ (resp. $\mu_\la^\infty=\Phi^*q_\la^\infty$). Then 
  \begin{align*}
    \lim_{n\to +\infty} q_\la^{n} &=q_\la^\infty \mbox{ for the $L^2(\TT)$ strong topology,}\\ 
    \lim_{n\to +\infty} \mu_\la^{n (k)} &=\mu_\la^{\infty (k)} \mbox{ in the sense of the uniform convergence, for all $k\in \NN$.} 
\end{align*}
\end{prop}

The following proposition states that in the generic case, one may observe up to two pairs of spikes for each spike of the solution of the positive Beurling-lasso.
As before, $r$ is chosen such that $0<r<\frac{1}{2} \min_{\nu\neq\nu'}|x_{\nu}-x_{\nu'}|$.

\begin{prop}
  Let $\la >0$, and assume that there exists a solution to~\eqref{eq-beurl-cbpasso} which is a sum of a finite number of (positive) Dirac masses: $m^\infty_\la= \sum_{\nu=1}^N \alpha_\nu \delta_{x_\nu}$ where $\alpha_\nu> 0$. Assume that $\mu_\la^{\infty}$ satisfies $|\mu_\la^{\infty}(t)|<1$ for all $t\in \TT\setminus \{x_1,\ldots ,x_N \}$. 
  
  Then any sequence of solution $m_{\la}^n=\sum_{i=0}^{\taillegridn-1} \veccont_{\la,i} \delta_{i\stepsizen+\shiftcont_{\la,i}/\veccont_{\la,i} }$ to~\eqref{eq-thin-cbpasso} satisfies 
  \begin{align*}
    \limsup_{n\to+\infty}\left(\supp m_{\la}^n\right)\subset \{x_1, \ldots x_N\}.
  \end{align*}
  If, moreover, $m^\infty_\la$ is the unique solution to~\eqref{eq-beurl-cbpasso}, 
  \begin{align}
    \lim_{n\to+\infty}\left(\supp (m_{\la}^n)\right)= \{x_1, \ldots x_N\}.
  \end{align}

If, additionally, $(\mu_\la^{\infty})''(x_\nu)\neq 0$ for some $\nu \in\{1,\ldots,N\}$, then for all $n$ large enough, the restriction of $m_{\la}^n$ to $(x_\nu-r,x_\nu+r)$ is a sum of Dirac masses whose configuration is given in Table~\ref{tab-nbdirac}, and if $(\mu_\la^{\infty})^{(3)}(x_\nu)\neq 0$, then only the cases indicated with $(\ast)$ may appear. 
\label{prop-cbp-thin-supplambda}
\end{prop}
\begin{table}
\begin{center}
    \renewcommand{\arraystretch}{1.5}\small
    \begin{tabular}{|m{1.6cm}|>{\centering\scriptsize}m{4.3cm}|>{\scriptsize}m{8.1cm}|}
    \hline
    \textbf{Number of Dirac masses } & \textbf{\small Saturations of the certificates} \\ $(\Sright /  \Sleft)$ & \textbf{\small Possible Dirac Locations}\\
    \hline
    \multirow{2}{*}{\textbf{One}} & $\{i\stepsizen\} /\emptyset $ or $\emptyset/\{i\stepsizen\}$  &$i\stepsizen+\varepsilon_n\frac{\stepsizen}{2}$, with $\varepsilon_n\in\{-1,1\}$\quad\hfill $(\ast)$ \\
& $\{i\stepsizen\}/\{i\stepsizen\}$ &  $i\stepsizen+t_i$, with $-\frac{\stepsizen}{2}\leq t_i\leq \frac{\stepsizen}{2}$\\\hline
\multirow{4}{*}{\textbf{Two}} & $\{(i-1)\stepsizen,i\stepsizen\}/\emptyset$\\ or \\ $\emptyset/\{i\stepsizen,(i+1)\stepsizen\}$ & $\left((i-\varepsilon_n)\stepsizen+\varepsilon\frac{\stepsizen}{2},i\stepsizen+\varepsilon_n\frac{\stepsizen}{2}\right)$, with $\varepsilon_n\in\{-1,1\}$\quad\hfill $(\ast)$\\
& $\{i\stepsizen\}/\{j\stepsizen\}$ & $\left(i\stepsizen+\frac{\stepsizen}{2}, j\stepsizen-\frac{\stepsizen}{2}\right)$, $i<j$\\
& $\{i\stepsizen\}/\{j\stepsizen,(i+1)\stepsizen\}$\\ or \\ $\{(i-1)\stepsizen,i\stepsizen\}/\{i\stepsizen\}$  & $\left((i-\varepsilon_n)\stepsizen+\varepsilon_n\frac{\stepsizen}{2},i\stepsizen+t_i\right)$,  $\varepsilon_n\in\{-1,1\}$, $-\frac{\stepsizen}{2}\leq t_i\leq\frac{\stepsizen}{2}$\\\hline
\multirow{3}{*}{\textbf{Three}} & $\{i\stepsizen\}/\{j\stepsizen,(j+1)\stepsizen\}$ & $\left(i\stepsizen+\frac{\stepsizen}{2},j\stepsizen-\frac{\stepsizen}{2}, (j+1)\stepsizen-\frac{\stepsizen}{2}\right)$, with $i<j$\\
& $\{(i-1)\stepsizen,i\stepsizen\}/\{j\stepsizen\}$ & $\left((i-1)\stepsizen+\frac{\stepsizen}{2},i\stepsizen+\frac{\stepsizen}{2},j\stepsizen-\frac{\stepsizen}{2}\right)$, with $i<j$\\
& $\{(i-1)\stepsizen,i\stepsizen\}/\{i\stepsizen,(i+1)\stepsizen\}$ & $\left((i-1)\stepsizen+\frac{\stepsizen}{2},i\stepsizen+t_i,(i+1)\stepsizen-\frac{\stepsizen}{2}\right)$,  $-\frac{\stepsizen}{2}\leq t_i\leq \frac{\stepsizen}{2}$\\\hline
\textbf{Four} & $\{(i-1)\stepsizen,i\stepsizen\}/\{j\stepsizen,(j+1)\stepsizen\}$ & $\left((i-1)\stepsizen+\frac{\stepsizen}{2},i\stepsizen+\frac{\stepsizen}{2}, j\stepsizen-\frac{\stepsizen}{2},(j+1)\stepsizen-\frac{\stepsizen}{2}\right)$, $i<j$\\\hline
 \end{tabular}
 \caption{Number of Dirac masses that may appear if $(\mu_\la^{\infty})''(x_\nu)\neq 0$. For the sake of the simplicity of the table, and since we focus on the saturations of dual certificates, we regard sums like $\delta_{i\stepsizen+\stepsizen/2}+\delta_{(i+1)\stepsizen-\stepsizen/2}$ as ``two'' Dirac masses.}\label{tab-nbdirac}
\end{center}
\end{table}

\begin{proof}
  By Proposition~\ref{prop-cbp-cvdual}, we know that the dual certificates $\mu_\la^{n}$ converge towards $\mu_\la^{\infty}$. By Lemma~\ref{lem-cbp-dualspikes} and the optimality conditions, we have thus $\limsup_{n\to+\infty} (\supp(m^n_\la))\subset \{x_1,\ldots,x_N\}$.
  If $m^\infty_\la$ is the unique solution, assume by contradiction that $\liminf (\supp (m^n_\la))\subsetneq \{x_1,\ldots,x_N\}$. Then there is some $\nu$, some $\varepsilon>0$ such that (up to a subsequence) $(\supp(m^n_\la))\cap (x_\nu-\varepsilon,x_\nu+\varepsilon)=\emptyset$. This contradicts the $\Gamma$-convergence result (Prop.~\ref{prop-cbpthin-convergence}) which ensures that $m^n_\la$ converges towards $m^\infty_\la$ for the weak* topology. As a result $\lim_{n\to+\infty} (\supp (m^n_\la))=\{x_1,\ldots,x_N\}$.

  If $(\mu_\la^{\infty})''(x_\nu)\neq 0$, Lemma~\ref{lem-cbp-dualspikes} ensures that the sets $\Sright(r)$ and $\Sleft(r)$ are of the form $\emptyset$, $\{i\stepsizen\}$, or $\{i\stepsizen,(i+1)\stepsizen\}$. Moreover, since $\lim_{n\to+\infty} (\supp m^n_\la)=\{x_1,\ldots,x_N\}$ we must have $\Sright(r)\neq \emptyset$ or $\Sleft(r)\neq \emptyset$. Using the fact that $\max \Sright(r)  \leq \min \Sleft(r)$, one may check that the only possible saturation points of $\mu_\la^{n}+\frac{\stepsizen}{2}{\mu_\la^{n}}'$ and  $\mu_\la^{n}-\frac{\stepsizen}{2}{\mu_\la^{n}}'$ are given in Table~\ref{tab-nbdirac}. The optimality conditions of Proposition~\ref{prop-optim-cbp} imply that $m^n_\la$ is at most a sum of Dirac masses at those locations.

If $(\mu_\la^{\infty})^{(3)}(x_\nu)\neq 0$ the third point of Lemma~\ref{lem-cbp-dualspikes} implies that for $n$ large enough, $\Sright(r)=\emptyset$ or $\Sright(r)=\emptyset$ (but not both). Hence there are at most two (successive) saturations, produced either by $\mu_\la^{n}+\frac{\stepsizen}{2} {\mu_\la^{n}}'$ or by $\mu_\la^{n}-\frac{\stepsizen}{2} {\mu_\la^{n}}'$.
\end{proof}

\begin{rem}
  Proposition~\ref{prop-cbp-thin-supplambda} states that the support of the C-BP on thin grids actually depends on the properties of the dual certificate $\mu_\la^\infty$ of the (positive) Beurling \lasso. The condition $(\mu_\la^\infty)''(x_\nu)\neq 0$ seems to be overwhelming, if not generic, and it is ensured for instance if $\la$ is small and the Non-Degenerate Source Condition holds (see~\cite{2013-duval-sparsespikes}).
  As for the condition $(\mu_\la^{\infty})^{(3)}(x_\nu)\neq 0$, it also seems to be generic, as there is nothing to impose $(\mu_\la^{\infty})^{(3)}(x_\nu)= 0$ in the positive Beurling \lasso. As a result, in practice, \textit{one does not observe all the configurations given in Table}~\ref{tab-nbdirac}, and \textit{only the cases indicated with $(\ast)$ appear}, the case of two spikes being again overwhelming. 
  
  This means that when approximating the positive Beurling \lasso with the Continuous Basis-Pursuit, one generally sees two spikes instead of one, and those spikes are at successive half-grid points: $(i\stepsize+\frac{\stepsize}{2}, (i+1)\stepsize+\frac{\stepsize}{2})$ or $(i\stepsize-\frac{\stepsize}{2},(i+1)\stepsize-\frac{\stepsize}{2})$.
\label{rem-cbp-nbspikes}
\end{rem}

\subsection{Asymptotic of the low noise support}
Now, we focus on the behavior of the Continuous Basis Pursuit at low noise. As for the \lasso, this analysis is more difficult in whole generality, since it involves the minimal norm solutions of nonlinear problems, in which it is difficult to pass to the limit. Therefore, we are led to assume that $\{x_{0,1}, \ldots , x_{0,N} \}\subset \Gg_n$, and the measure now reads $m_0=\sum_{i=0}^{\taillegridn-1} \veccont_{0,i} \delta_{i\stepsizen}$.

The following property ensures that $m_0$ is a solution to~\eqref{eq-thin-cbp} for each $n$ large enough.

\begin{lem}
  Assume that there exists a function $\mu\in \Im \Phi^*$, such that for all $t\in\TT\setminus \{x_{0,1},\ldots ,x_{0,N} \}$, $\mu(t)<1$ and 
\begin{align}
  \forall \nu\in \{1,\ldots, N\}, \ \mu(x_{0,\nu})=1, \ \mu''(x_{0,\nu})\neq 0, \ \mu^{(3)}(x_{0,\nu})=0,\ \mu^{(4)}(x_{0,\nu})>0.
\end{align}
Then, for all $n$ large enough, $\mu$ is a dual certificate for $m_0= \sum_{\nu=1}^N \alpha_{0,\nu}\delta_{x_{0,\nu}}$ for~\eqref{eq-thin-cbp}, and $m_0$ is a solution to~\eqref{eq-thin-cbp}. Moreover, if $\Gamma_{x_0}$ has full rank, this solution is unique.
%of REF satisify $(i-\frac{1}{2})\De< i\De + \tau_i<(i+\frac{1}{2})$
\label{lem-source-cbp}
\end{lem}

\begin{rem} The condition $\mu^{(3)}(x_{0,\nu})=0$ is natural since our aim is to build a certificate which is valid for all $n$, hence Lemma~\ref{lem-cbp-dualspikes} applies with $\Sright(r)\neq \emptyset$ and $\Sleft(r)\neq \emptyset$.  
\end{rem}

\begin{proof}
  Let $\nu\in \{1,\ldots, N\}$ and $r_\nu\in (0,r)$ such that $\mu''(t)<0$, and $\mu^{(4)}(t)>0$ in $(x_{0,\nu}-r_\nu,x_{0,\nu}+r_\nu)$.
  We shall prove that $\mu(k\stepsizen)+\frac{\stepsizen}{2} |\mu'(k\stepsizen)|<1$ for all $k$ such that $k\stepsizen\in (x_{0,\nu}-r_\nu,x_{0,\nu}+r_\nu)\setminus \{x_{0,\nu}\}$.
  To simplify the notation, we assume without loss of generality that $x_{0,\nu}=0$ and we write $\tilde{r}=r_\nu$. The variations of $\mu$ and its derivatives are given by the table below:

\begin{center}
\begin{tikzpicture}
  \tkzTabInit[deltacl=1]{$t$ / 1,$\mu^{(4)}$ /1, $\mu^{(3)}$/1,$\mu''$ /1.5,$\mu'$ /1, $\mu$ /1.5 }{$-\tilde{r}$, $0$ , $\tilde{r}$}%
  \tkzTabLine{,,+,}
  \tkzTabVar{-/,R/$0$ ,+/}\tkzTabVal{1}{3}{0.5}{0}{0}
  \tkzTabVar{+/$\mu''(-\tilde{r})<0$,-/$\mu''(0)$ ,+/$\mu''(\tilde{r})<0$}
  \tkzTabVar{+/,R/$0$ ,-/}\tkzTabVal{1}{3}{0.5}{0}{0}
  \tkzTabVar{-/,+/1,-/}
\end{tikzpicture}
\end{center}
Let us observe that the function $\theta : t \mapsto \mu(t)-\frac{t}{2} \mu'(t)$ is (strictly) decreasing in $[0,\tilde{r})$, since 
  \begin{align}
    \forall t\in (0,\tilde{r}),\    \theta'(t)=\frac{1}{2} \left(\mu'(t)-t\mu''(t) \right)=\frac{1}{2} \int_0^t\underbrace{(\mu''(u)-\mu''(t))}_{<0}du<0.
  \end{align}
  Hence, for all $k$ such that $k\stepsizen\in (0,\tilde{r})$, 
  \begin{align}
    \mu(k\stepsizen)-\frac{\stepsizen}{2} \mu'(k\stepsizen) &=\underbrace{\mu(k\stepsizen)-\frac{k\stepsizen}{2} \mu'(k\stepsizen)}_{=\theta(k\stepsizen)<\theta(0)=1} + \underbrace{\frac{(k-1)\stepsizen}{2} \mu'(k\stepsizen)}_{<0}<1.
  \end{align}

On the other hand, $\theta$ is (strictly) increasing on $(-\tilde{r}, 0]$ since 
  \begin{align}
    \forall t\in (-\tilde{r},0),\    \theta'(t)=\frac{1}{2} \left(\mu'(t)-t\mu''(t) \right)=\frac{1}{2} \int_0^t\underbrace{(\mu''(u)-\mu''(t))}_{<0}du>0.
  \end{align}
  As a consequence, for all $k$ such that $k\stepsizen\in (-\tilde{r},0)$, 
  \begin{align}
    \mu(k\stepsizen)+\frac{\stepsizen}{2} \mu'(k\stepsizen) &=\underbrace{\mu(k\stepsizen)-\frac{k\stepsizen}{2} \mu'(k\stepsizen)}_{=\theta(k\stepsizen)<\theta(0)=1} + \underbrace{\frac{(k+1)\stepsizen}{2} \mu'(k\stepsizen)}_{\leq 0}<1.
  \end{align}
  Thus we see that $\mu(k\stepsizen)+\frac{\stepsizen}{2} |\mu'(k\stepsizen)|<1$ for all $k\stepsizen \in(-\tilde{r},\tilde{r})\setminus\{0\}$, and we proceed similarly on all the intervals of the form $(x_{0,\nu}-r_\nu,x_{0,\nu}+r_\nu)$. By a compactness argument, there exists a constant $\beta<1$ such that $\mu(t)\leq \beta$ for all $t\in \TT\setminus \bigcup_{\nu=1}^N(x_{0,\nu}-r_{\nu},x_{0,\nu}+r_{\nu})$. For $n$ large enough, the inequality $\frac{\stepsizen}{2} \left(\sup_{t\in \TT} |\mu'(t)|\right)<1-\beta$ holds, and we see that $\mu(k\stepsizen) + \frac{\stepsizen}{2} |\mu'(k\stepsizen)|<1$ for all $t\in \TT\setminus \bigcup_{\nu=1}^N(x_{0,\nu}-r_{\nu},x_{0,\nu}+r_\nu)$.

  As a conclusion, we see that $\mu$ is a valid certificate for $(\veccontO,0)$ (see the optimality conditions of Proposition~\ref{prop-optim-cbp}), thus  $(\veccontO,0)$ is a solution of~\eqref{eq-thin-cbp}.  
\end{proof}

Now, we consider the limit of the minimal norm solutions of~\eqref{eq-thin-dual-cbp}. In general, they do not converge towards the minimal norm solution of~\eqref{eq-beurl-dual-cbp}, and we are led to introduce a new variational problem to carry the study further.

\begin{defn}[Third derivative precertificate]\label{defn-third-deriv-precertif}
Given $m_0\in \Mm(\TT)$, we define the \textit{third derivative precertificate} as $\mu_T \eqdef \Phi^*q_T$ where 
  \begin{align}
    q_T \eqdef \uargmin{q\in\Ldeux(\TT)}\bigg{\{}\|q\|_2; \ & \forall i\in \{1,\ldots, N\},\ (\Phi^*q)(x_i)=1, \nonumber\\
    &  \quad \quad (\Phi^*q)'(x_i)=0 \mbox{ and } (\Phi^*q)^{(3)}(x_i)=0 \bigg{\}}, \label{eq-def-thirdderivcertif}
  \end{align}
whenever the above set is not empty.
\end{defn}

It is clear that the set defined in~\eqref{eq-def-thirdderivcertif} is a closed convex set. It is nonempty for instance if the conditions of Lemma~\ref{lem-source-cbp} hold. Note that $q_T$ corresponds to a quadratic minimization under linear constraint, and can hence be computed by solving a linear system,
\begin{align}
  q_T &= \begin{pmatrix}\Gamma_{x_0}^*\\{\Phi_{x_0}^{(3)}}^* \end{pmatrix}^+
  \begin{pmatrix}\begin{pmatrix}\bun_N\\0\end{pmatrix}\\ 0\end{pmatrix}
  = \Gamma_{x_0}^{+,*}\begin{pmatrix}\bun_N\\0\end{pmatrix}-\tilde{\Pi} {\Phi_{x_0}^{(3)}}^* ({\Phi_{x_0}^{(3)}}^*\tilde{\Pi}{\Phi_{x_0}^{(3)}}^*)^{-1}{\Phi_{x_0}^{(3)}}^* \Gamma_{x_0}^{+,*}\begin{pmatrix}\bun_N\\0\end{pmatrix}\label{eq-qt-expression}\\
&= \pvanishing -\tilde{\Pi} {\Phi_{x_0}^{(3)}}^* ({\Phi_{x_0}^{(3)}}^*\tilde{\Pi}{\Phi_{x_0}^{(3)}}^*)^{-1}{\Phi_{x_0}^{(3)}}^*\pvanishing,
\end{align}
where $\tilde{\Pi}$ is the orthogonal projector onto $(\Im \Gamma_{x_0})^\perp$, and $\Gamma_{x_0}=\begin{pmatrix} \Phi_{x_0} & \Phi_{x_0}'\end{pmatrix}$.

\begin{defn}[Twice Non-Degenerate Source Condition]\label{defn-TNDSC}
  We say that $m_0$ satisfies the \textit{Twice Non-Degenerate Source Condition} (TNDSC) if $q_T$ in~\eqref{eq-def-thirdderivcertif} is well defined and if it satisfies, for $\mu_T=\Phi^*q_T$,
  \begin{align*}
    \forall t\in \TT\setminus \{x_{0,1},\ldots ,x_{0,N} \},\ \mu_T(t)<1,\\
    \forall \nu\in \{1,\ldots, N\},\quad    \mu_T''(x_{0,\nu})<0 \qandq \mu_T^{(4)}(x_{0,\nu})>0.
  \end{align*}
\end{defn}
Observe that if the Twice Non-Degenerate Source Condition holds, the hypotheses of Lemma~\ref{lem-source-cbp} are satisfied and $m_0$ is a solution to~\eqref{eq-thin-cbp} for $n$ large enough. In fact the associated minimal norm certificates (which thus exist) converge towards $\mu_T$.

\begin{prop}\label{prop-cv-cbpcertif}
  Let $m_0\in \Mm(\TT)$ satisfy the \textit{Twice Non-Degenerate Source Condition} (and $\mu_T$ the corresponding Third derivative (pre)certificate). Let $q_0^n$ be the minimal norm solution of~\eqref{eq-thin-dual-cbp}, and $\mu_0^n=\Phi^*q_0^n$. Then,
\begin{align}
    \lim_{n\to +\infty} q_0^{n} &=q_T \mbox{ for the $L^2(\TT)$ strong topology,}\\ 
    \lim_{n\to +\infty} \mu_0^{n (k)} &=\mu_T^{(k)} \mbox{ in the sense of the uniform convergence, for all $k\in \NN$.} 
\end{align}
\end{prop}

\begin{proof}
  As mentioned above, the Twice Non-Degenerate Source Condition implies that $\mu_T$ is a function admissible for Lemma~\ref{lem-source-cbp}, hence a certificate for~\eqref{eq-thin-cbp}.
  As a result, $\|q_0^n\|_2\leq \|q_T\|_2$ and the sequence $(q_0^n)_{n\in\NN}$ is bounded in $L^2(\TT)$. We may extract a subsequence $q_0^{n'}$ which weakly converges towards some $\tilde{q}\in L^2(\TT)$, and then $\|\tilde{q}\|_2\leq \liminf_{n'\to +\infty}\|q_0^n\|_2\leq \|q_T\|_2$. Since $\Phi^*$ and $\Phi^{(k),*}$ are compact (see Lemma~\eqref{lem-phi-compact} in Appendix), we obtain that $\mu_0^{n' (k)}=(\Phi^*q_0^{n'})^{(k)}$ converges toward $\tilde{\mu}=\Phi^* \tilde{q}$ for the (strong) topology of the uniform convergence. We immediately obtain that $\tilde{\mu}(t)\leq 1$ for all $t\in \TT$, and $\tilde{\mu}(x_{0,\nu})=1$, $\tilde{\mu}(x_{0,\nu})=0$ for all $\nu\in \{1,\ldots ,N \}$. 
  
  Moreover, applying Lemma~\ref{lem-cbp-dualspikes} to $\Phi^*q_0^n$ (observing that $x_\nu \in \Sright(r)\cap \Sleft(r)$), we get $\tilde{\mu}^{(3)}(x_{0,\nu})=0$. As a result, $\tilde{q}$ is admissible for~\eqref{eq-def-thirdderivcertif}, hence $\|q_T\|_2\leq \|\tilde{q}\|_2$. Thus in fact $\|q_T\|_2= \|\tilde{q}\|_2$ and $q_T=\tilde{q}$. Since the limit of the extracted subsequence does not depend on the choice of the subsequence, in fact the whole sequence converges. Moreover, the convergence is strong in $L^2(\TT)$ since $\lim_{n\to +\infty}\|q_0^n\|_2=\|q_T\|_2$. 
\end{proof}
As a consequence of the above convergence result, the third derivative precertificate controls the extended support on thin grids. 

\begin{prop}
  Let $m_0\in \Mm(\TT)$ (with $\{x_{0,1},\ldots x_{0,N}\}\subset \Gg_n$) such that the Twice Non Degenerate Source Condition holds.
  Then, for $n$ large enough, $m_0$ is a solution to~\eqref{eq-thin-cbp} and its extended support is given by:
  \begin{align}
    \extun (m_0)= \bigcup_{\nu=1}^N \Sright(r), \qandq  \extdn (m_0)= \bigcup_{\nu=1}^N \Sleft(r),
  \end{align}
  where
  \begin{itemize}
    \item $\Sright(r)$ is equal to $\{x_{0,\nu}\}$ or $\{x_{0,\nu}-\stepsizen,x_{0,\nu}\}$,
    \item $\Sleft(r)$ is  equal to $\{x_{0,\nu}\}$ or $\{x_{0,\nu},x_{0,\nu}+\stepsizen\}$.
  \end{itemize}   
  Moreover, one cannot have simultaneously $\Sright(r)=\{x_{0,\nu}-\stepsizen,x_{0,\nu}\}$ and $\Sright(r)=\{x_{0,\nu},x_{0,\nu}+\stepsizen\}$.
  \label{prop-cbp-thin-extendedsup}
\end{prop}

\begin{proof}
By Lemma~\ref{lem-source-cbp}, $m_0$ is a solution to~\eqref{eq-thin-cbp} and $\mu_T$ is a solution to~\eqref{eq-thin-dual-cbp}. 
Applying Lemma~\ref{lem-cbp-dualspikes} to $\mu_0^{n}$, $\mu_T$, we see that $\Sright(r)$ is of the form $\emptyset$, $\{i\stepsizen\}$ or $\{(i-1)\stepsizen,i\stepsizen\}$, and that $\Sleft(r)$ is of the form $\emptyset$, $\{j\stepsizen\}$ or $\{j\stepsizen,(j+1)\stepsizen\}$, with $i\leq j$.
On the other hand, by the extremality relations between $\mu_0^{n}$ (solution of~\eqref{eq-thin-dual-cbp}) and $m_0^{n}$ (solution of~\eqref{eq-thin-cbp}), $x_{0,\nu}\in\Sright(r)$ and $x_{0,\nu}\in \Sleft(r)$. As a consequence $\Sright(r)$ is equal to $\{x_{0,\nu}\}$ or $\{x_{0,\nu}-\stepsizen,x_{0,\nu}\}$, and $\Sleft(r)$ is  equal to $\{x_{0,\nu}\}$ or $\{x_{0,\nu},x_{0,\nu}+\stepsizen\}$.

Now, since $\mu_T^{4}(0)\neq 0$, the fourth point of Lemma~\ref{lem-cbp-dualspikes} ensures that for $n$ large enough, one cannot have simultaneously $\Sright(r)=\{x_{0,\nu}-\stepsizen,x_{0,\nu}\}$ and $\Sleft(r)=\{x_{0,\nu},x_{0,\nu}+\stepsizen\}$.
\end{proof}

\begin{rem}
  As Proposition~\ref{prop-cbp-thin-extendedsup} shows, for each original spike, at most one pair of spikes appears at low noise : the original spike slightly shifted and either the immediate left neighbor shifted by $+\stepsizen/2$ or the immediate right neighbor shifted by $-\stepsizen/2$.
\end{rem}

We are now in position to provide a sufficient condition for the spikes to appear in pair, with a prediction on the location of the neighbor.

\begin{thm}\label{thm-cbpasso-extended}
  Assume that the operator $\begin{pmatrix}\Phi_{x_0}&\Phi_{x_0}'&\Phi_{x_0}^{(3)}\end{pmatrix}$ has full rank and that the Twice Non Degenerate Source condition (Definiton~\ref{defn-TNDSC}) holds.
  Moreover, assume that all the components of the natural shift
  \begin{align}
    \rho \eqdef  (\Phi_{x_0}^{(3)*}\tilde{\Pi}\Phi_{x_0}^{(3)})^{-1}\Phi_{x_0}^{(3)*}\Gamma_{x_0}^{+,*}\begin{pmatrix}
      \bun_{N}\\0
    \end{pmatrix}
  \end{align}
  are nonzero.
  Then, for $n$ large enough, and all $\nu\in\{1,\ldots,N\}$,
\begin{align}
  \mbox{If } \rho_\nu>0,& \mbox{ then }\Sright(r) = \{x_{0,\nu}-\stepsizen,x_{0,\nu}\}, \qandq \Sleft(r)=\{x_{0,\nu}\}, \\
  \mbox{If } \rho_\nu<0,& \mbox{ then }\Sleft(r) = \{x_{0,\nu}\}, \qandq \Sleft(r)=\{x_{0,\nu},x_{0,\nu}+\stepsizen\},
\end{align}  
  so that the extended support of $m_0$ on the grid $\Gg_n$ has the form
  \begin{align*}
    \extun(m_0) &=\{x_{0,1},\ldots, x_{0,N}\}\cup \enscond{x_{0,\nu}-\stepsizen}{\nu\in\seg{1}{N}\qandq \rho_\nu>0} \\
    \extdn(m_0) &=\{x_{0,1},\ldots, x_{0,N}\}\cup \enscond{x_{0,\nu}+\stepsizen}{\nu\in\seg{1}{N}\qandq \rho_\nu<0}.
  \end{align*}
\end{thm}

\begin{cor}\label{cor-cbp-extended}
  Under the hypotheses of Theorem~\ref{thm-cbpasso-extended}, for $n$ large enough, there exists constants $C^{(1)}_n>0$, $C^{(2)}_n>0$ such that 
  for $\la\leq C^{(1)}_n \min_{1\leq \nu\leq N} |\alpha_{0,\nu}|$, and for all $w\in L^2(\TT)$ such that $\|w\|_2\leq C^{(2)}_n \la$, the solution to~\eqref{eq-thin-cbpasso} is unique, and reads $m_\la=\sum_{\nu=1}^N (\alpha_{\la,\nu}\delta_{x_{0,\nu}+t_\nu}  + \beta_{\la,\nu}\delta_{x_{0,\nu}+\varepsilon_\nu\stepsizen})$, where 
\begin{align*}
-\stepsizen/2<t_\nu<\stepsizen/2 \qandq \varepsilon&\eqdef -\sign\left(\rho\right).
% (\Phi_{x_0}^{(3)*}\tilde{\Pi}\Phi_{x_0}^{(3)})^{-1}\Phi_{x_0}^{(3)*}\Gamma_{x_0}^{+,*} \begin{pmatrix}      \bun_{N}\\0    \end{pmatrix}
  \end{align*}
\end{cor}

\begin{proof}[Proof of Theorem~\ref{thm-cbpasso-extended}]
  We proceed as in the proof of Theorem~\ref{thm-lasso-extended} by building a good candidate for $\mu_0^n$ and using Lemma~\ref{lem-mu0}.
To comply with the notations of Lemma~\ref{lem-mu0}, let us write  $\sum_{\nu=1}^N\alpha_{0,i}\delta_{x_{0,\nu}}= \sum_{k=0}^{\taillegridn-1} a_{0,k}\delta_{k\stepsizen}$, and $\Iup\eqdef \Idown\eqdef I\eqdef\enscond{i\in\seg{0}{\taillegridn-1}}{a_{0,i}\neq 0}$. 

For any choice of shift $(\epsilon_i)_{i\in I}\in \{-1,+1\}^{N}$, we set $\Jup \eqdef \Iup\cup\enscond{i+\varepsilon_i}{i\in I \qandq\varepsilon_i=-1}$ and $\Jdown\eqdef \Idown\cup\enscond{i+\varepsilon_i}{i\in I\qandq\varepsilon_i=+1}$.
  Since $|x_{0,\nu}-x_{0,\nu'}|> {2}{\stepsizen}$ for $\nu'\neq \nu$ and $n$ large enough, we have $\Card \Jup+\Card \Jdown=3\times \Card I=3N$. 
We shall find a choice of $\varepsilon$ such that $u_j > 0$ for all $j\in \Jup\setminus I$, and $v_j> 0$ for all $j\in \Jdown\setminus J$, where
\begin{align*}
  \begin{pmatrix} u_{\Jup}\\v_{\Jdown} \end{pmatrix}&\eqdef-(\Copt^*\Copt)^{-1}\begin{pmatrix}
    \bun_{\Jup}\\\bun_{\Jdown}\end{pmatrix}, \quad
\Copt \eqdef\begin{pmatrix}(\OpU+\frac{\stepsize}{2}\OpD)_{\Jup} &(\OpU-\frac{\stepsize}{2}\OpD)_{\Jdown}\end{pmatrix}\\
\OpU &\eqdef \Phi_{\Gg_n} \qandq \OpD\eqdef\Phi'_{\Gg_n}.
\end{align*}

In this particular case where $\Iup= \Idown= I$, all $j$ in $(\Jup\setminus I)\cup (\Jdown\setminus I)$ may be uniquely written as $j=i+\varepsilon_i$ for some $i\in I$, where $\varepsilon_i\in\{-1,+1\}$.
We may swap the columns of $\Copt$ so as to reformulate the condition  $\begin{pmatrix} u_{\Jup}\\v_{\Jdown} \end{pmatrix}=-(\Copt^*\Copt)^{-1}\bun_{3N}$ into
\begin{align*}  
  \begin{pmatrix} \tilde{u}_I\\ \tilde{v}_I\\ \tilde{t}_I  \end{pmatrix}&=-(\Capt^*\Capt)^{-1}\begin{pmatrix}
    \bun_N\\\bun_N\ \\\bun_N 
  \end{pmatrix},
\end{align*}
where $\Capt \eqdef \begin{pmatrix} \OpU_I+ \frac{\stepsizen}{2}\OpD_I\diag(\varepsilon)  & \OpU_I-\frac{\stepsizen}{2}\OpD_I\diag(\varepsilon) & \OpU_{I+\varepsilon}-\frac{\stepsizen}{2}\OpD_{I+\varepsilon}\diag(\varepsilon)\end{pmatrix}$
and $\tilde{t}_i>0$ for all $i\in I$.
But a Taylor expansion yields
\begin{align*}
  \OpU_{I+\varepsilon}-\frac{\stepsizen}{2}\OpD_{I+\varepsilon}\diag(\varepsilon)&= \underbrace{\Phi_{x_0}}_{=\OpU_I}+ \frac{\stepsizen}{2}\underbrace{\Phi_{x_0}'\diag(\varepsilon)}_{=\OpD_I\diag(\varepsilon)}+ (\stepsizen)^3 \gamma_3\Phi^{(3)}_{x_0}\diag(\varepsilon)+o(\stepsizen^3),
\end{align*}
where we defined
\eql{\label{eq?defn-alphak}
	\gamma_k \eqdef \frac{1}{k!}- \frac{1}{(k-1)!\times 2}. 
}
Hence, we may apply Lemma~\ref{lem-apx-cbpdl} to $\Phi_{x_0}$, $\Phi_{x_0}'\diag(\varepsilon)$ and $\gamma_3\Phi^{(3)}_{x_0}\diag(\varepsilon)$ so as to obtain
\begin{align*}
  \tilde{t}_I &=  -\frac{1}{\gamma_3\stepsizen^3}\diag(\varepsilon) \rho +o\left(\frac{1}{\stepsizen^3}\right).
\end{align*}
Therefore it is sufficient to choose $\varepsilon = -\sign(\rho)$ to make all the components of $\tilde{t}_I$ nonnegative.

%  \left((\Phi_{x_0}^{(3)*}\tilde{\Pi}\Phi_{x_0}^{(3)})^{-1}\Phi_{x_0}^{(3)*}\Gamma_{x_0}(\Gamma_{x_0}^ \Gamma_{x_0})^{-1} 		\begin{pmatrix}      		\bun_{N}\\0    	\end{pmatrix}\right)

With that choice of $\varepsilon$, it remains to prove that
\eq{
	\max \left[\begin{pmatrix}(\OpU+\frac{\stepsize}{2}\OpD)_{\Jup^c}^* \\(\OpU-\frac{\stepsize}{2}\OpD)_{\Jdown^c}^* \end{pmatrix} \Copt \begin{pmatrix} u_{\Jup}\\v_{\Jdown} \end{pmatrix}\right]<1. 
}
Let us write $\tilde{q}_n\eqdef \Copt \begin{pmatrix} u_{\Jup}\\v_{\Jdown} \end{pmatrix}$. Since $\begin{pmatrix} u_{\Jup}\\v_{\Jdown} \end{pmatrix}=-(\Copt^*\Copt)^{-1}\bun_{3N}$, we get $\tilde{q}_n=\Copt^{+,*}\bun_{3N}$, and applying Lemma~\ref{lem-apx-cbpdl} to $\Phi_{x_0}$, $\Phi_{x_0}'\diag(\varepsilon)$ and $\gamma_3\Phi^{(3)}_{x_0}\diag(\varepsilon)$, we see that 
$\tilde{q}_n$ converges towards $q_T$ (using~\eqref{eq-qt-expression}).

By construction of $\tilde{q}_n$, 
\begin{align}
  &\forall j\in\Jup\setminus I, \ (\Phi^*\tilde{q}_n+\frac{\stepsizen}{2}(\Phi^*\tilde{q}_n)')(j\stepsizen)=1,\nonumber\\
  \qandq &\forall j\in\Jdown\setminus I, \  (\Phi^*\tilde{q}_n-\frac{\stepsizen}{2}(\Phi^*\tilde{q}_n)')(j\stepsizen)=1,\label{eq-cbp-pbaux}
  \intertext{which may be summarized as }
  &\forall i\in I, (\Phi^*\tilde{q}_n-\varepsilon_i\frac{\stepsizen}{2}(\Phi^*\tilde{q}_n)')((i+\varepsilon_i)\stepsizen)=1.\nonumber
   \end{align}

   Arguing as in the proof of point~\eqref{item-deriv4} in Lemma~\ref{lem-cbp-dualspikes} (replacing ``$1=\ldots$'' with ``$1\geq\ldots$'' and using that $\mu_T^{(4)}(x_{0,\nu})>0$), we may prove that for $n$ large enough, $(\Phi^*\tilde{q}_n+\varepsilon_i\frac{\stepsizen}{2}(\Phi^*\tilde{q}_n)')((i-\varepsilon_i)\stepsizen)<1$.

   Then, by the same argument of compactness and local concavity as in point~\eqref{item-deriv2} of Lemma~\ref{lem-cbp-dualspikes}, we observe that 
   \begin{align*}
  \enscond{k\in \seg{0}{\taillegridn-1}}{(\Phi^*\tilde{q}_n+\frac{\stepsizen}{2}(\Phi^*\tilde{q}_n)')(k\stepsizen)\geq 1}\subset \Jup,\\
 \enscond{k\in \seg{0}{\taillegridn-1}}{(\Phi^*\tilde{q}_n-\frac{\stepsizen}{2}(\Phi^*\tilde{q}_n)')(k\stepsizen)\geq 1}\subset \Jdown,
\end{align*}
and those inclusions are in fact equalities.
That precisely means that $\max \left[\begin{pmatrix}(\OpU+\frac{\stepsize}{2}\OpD)_{\Jup^c}^* \\(\OpU-\frac{\stepsize}{2}\OpD)_{\Jdown^c}^* \end{pmatrix}\tilde{q}_n\right]<1$.

Hence, by Lemma~\ref{lem-mu0}, $\Phi^*\tilde{q}_n$ is the minimal norm certificate $\mu_0^n$ and $(\Jup\stepsizen,\Jdown\stepsizen)$ is the extended support. This concludes the proof.
\end{proof}

\subsection{Asymptotics of the constants}

Again, we may examine the asymptotic behavior of the constants given in Corollary~\ref{cor-cbp-extended}. Those constants stem from Theorem~\ref{thm-abstract-cbp} which is itself a variant of Theorem~\ref{thm-stability-dbp} for the \lasso.

Replacing the constants $c_1,\ldots, c_3$ of the proof of Theorem~\ref{thm-stability-dbp} with the corresponding expressions for the C-BP, and using Lemma~\ref{lem-apx-cbpdl} we get
\begin{align}\label{eq-cbp-cst-asympt-1}
	c_{1,n}&= \norm{R_{\Iup\cup\Idown}\Copt^+}  
	\sim 
	\frac{1}{(\stepsizen)^3}\normb{\begin{pmatrix}
    (\Phi_{x_0}^{(3),*}\tilde{\Pi} \Phi_{x_0}^{(3)})^{-1}\Phi_{x_0}^{(3),*}\tilde{\Pi} \\
    0
  \end{pmatrix}}_{\infty,2}\\
  	\label{eq-cbp-cst-asympt-2}
  	c_{2,n}&= 
	\normb{\begin{pmatrix} \tilde{u}_I\\\tilde{v}_I \end{pmatrix}}
	\sim 
	\frac{1}{(\stepsizen)^3} 
	\normb{
		\begin{pmatrix}\rho \\ 0 \end{pmatrix}
	}_{\infty}\\
	\label{eq-cbp-cst-asympt-3}
  	c_{3,n} &= 
  	\left(\norm{R_{(\Jup\setminus\Iup)\cup(\Jdown\setminus\Idown)}\Copt^+}_{\infty,2}\right)^{-1}
	\left(\min_{i\in I}\tilde{t}_i\right)
		\sim 
		\frac{
			\min_i \left| \frac{1}{\gamma_3} \rho_i \right|  
		}{
			\norm{(\Phi_{x_0}^{(3),*}\tilde{\Pi} \Phi_{x_0}^{(3)})^{-1}\Phi_{x_0}^{(3),*}\tilde{\Pi}}_{\infty,2}
		}
\end{align}
where $\gamma_k$ is defined in~\eqref{eq?defn-alphak}.
As for $c_{4,n}$ and $c_{5_n}$, like in the case of the \lasso, their expression lead to a pessimistic bound for the low noise regime, and we are led to make finer majorizations.

\begin{prop}\label{prop-asympto-constant-cbp}
  The constants $C^{(1)}_n, C^{(2)}_n$ in Corollary~\ref{cor-cbp-extended} can be chosen as $C^{(1)}_n=O(\stepsizen^3)$ and $C^{(2)}_n=O(1)$, and one has
  	\eql{\label{eq-lipsch-cbp}
		\normb{ 	
  		\begin{pmatrix} \alpha_\la\\\beta_\la \end{pmatrix}
		- 
		\begin{pmatrix} \alpha_0\\0 \end{pmatrix}
 		}_{\infty} = O\pa{ \frac{w}{\stepsizen^3}, \frac{\la}{\stepsizen^3} }.
	}
\end{prop}

\begin{proof}
	The proof of~\eqref{eq-lipsch-cbp} follows from~\eqref{eq-cbp-cst-asympt-1} and~\eqref{eq-cbp-cst-asympt-2}. 
  Using the reformulation~\eqref{eq-reparam-cbpasso} of the C-BP as a (positive) \lasso, we have to ensure that~\eqref{eq-abstract-strict} holds, or more precisely,
\begin{align*}
    \max \left[  \begin{pmatrix}
(\OpU^*+\frac{\stepsizen}{2} \OpD^*)_{(\Jup)^c}\\ (\OpU^*-\frac{\stepsizen}{2} \OpD^*)_{(\Jdown)^c}
  \end{pmatrix}
  \begin{pmatrix}
    y- \OpU\veccont -\OpD\shiftcont
\end{pmatrix}\right]<\la 
\end{align*}
where $\OpU \eqdef \Phi_{\Gg_n}$, $\OpD\eqdef\Phi'_{\Gg_n}$. 
  Let $\Copt \eqdef\begin{pmatrix}(\OpU+\frac{\stepsize}{2}\OpD)_{\Jup} &(\OpU-\frac{\stepsize}{2}\OpD)_{\Jdown}\end{pmatrix}$, $\tilde{\Pi}$ be the orthogonal projector onto $\ker \Copt^*=(\Im \Copt)^\perp$, and $\omega= \Phi^*\tilde{\Pi} w$. Since 
  \eq{
  	y - \OpU\veccont -\OpD\shiftcont= w- \Copt(\Copt^*\Copt)^{-1}\Copt^*w + \la \Copt(\Copt^*\Copt)^{-1}\bun_{3N} = \tilde{\Pi} w  +\la \Copt^{+,*}\bun_{3N}, 
	}
	we are led to check that
  \begin{align}
    (\omega +\la \mu_0^n)(j\stepsizen) + \frac{\stepsizen}{2} (\omega +\la \mu_0^n)'(j\stepsizen) &<\la \quad \mbox{ for all } j\in (\Jup)^C,\label{eq-tight-snr-cbpU}\\
    (\omega +\la \mu_0^n)(j\stepsizen) - \frac{\stepsizen}{2} (\omega -\la \mu_0^n)'(j\stepsizen) &<\la \quad \mbox{ for all } j\in (\Jdown)^C\label{eq-tight-snr-cbpD},
  \end{align}
  where $\mu_0^n\eqdef \Phi^*(\Copt^*\Copt)^{-1})\bun_{3N}$ yields the minimal norm certificate 
  \eq{
  	\begin{pmatrix}
    (\mu_0^n  +\frac{\stepsizen}{2}\mu_0^n)(\Gg_n)\\  (\mu_0^n  -\frac{\stepsizen}{2}\mu_0^n)(\Gg_n) 
  \end{pmatrix}=\begin{pmatrix}(\OpU+\frac{\stepsize}{2}\OpD)^* \\(\OpU-\frac{\stepsize}{2}\OpD)^*\end{pmatrix}(\Copt^*\Copt)^{-1}\bun_{3N}.
  }

Given $0<r<\frac{1}{2}\min_{\nu\neq \nu'}|x_{0,\nu}-x_{0,\nu'}|$, let $N(r)\eqdef\bigcup_{\nu}(x_{0,\nu}-r,x_{0,\nu}+r)$ be a neighborhood of the $x_{0,\nu}$'s. By the Twice Non-Degenerate Source condition, we may choose $r>0$, such that
\begin{align*}
  -\tilde{k}_1\eqdef \sup_{t\in N(r)}\mu_T''(t)< 0,\qandq  \tilde{k}_{2}\eqdef \inf_{t\in N(r)}\mu_T^{(4)}(t)>0.
\end{align*}
By compactness, $\tilde{k}_3\eqdef \sup_{t\in\TT\setminus N(r)}\mu_T(t)<1$.

Let us recall that $\mu_0^n\to \mu_T$ in the sense of the uniform convergence (and similarly for the derivatives). As a result, for $n\in\NN$ large enough, 
\eql{
  	\sup_{t\in N(r)}(\mu_0^n)''(t)<-\frac{\tilde{k}_1}{2}<0,
	\:  
	\inf_{t\in N(r)}(\mu_0^n)^{(4)}(t)>\frac{\tilde{k}_{2}}{2}>0, 
	\: 
  	\sup_{t\in\TT\setminus N(r)} \mu_0^n(t) <\frac{1+\tilde{k}_3}{2}<1, \label{eq-cbp-const-muO}
}
\eq{
  	\frac{\stepsizen}{2} \norm{(\mu_0^n)^{(3)}}_{\infty} \leq \frac{\tilde{k}_1}{8},
	\qandq 
	\frac{\stepsizen}{2} \norm{(\mu_0^n)'}_{\infty} \leq \frac{1-\tilde{k}_3}{6}.
}
Now, we assume that $\frac{\norm{w}_2}{\la}$ is small enough, so that 
\begin{align*}
  \norm{(\Phi^{(k)})^*}_{\infty,2}\frac{\norm{w}_2}{\la} <\frac{\tilde{k}_1}{8}, 
  \mbox{ for } k\in\{2,3\}, \quad  
  \norm{(\Phi^{(4)})^*}_{\infty,2}\frac{\norm{w}_2}{\la} <\frac{\tilde{k}_2}{4},
\end{align*}
\begin{align}
	\qandq  \norm{(\Phi^{(k)})^*}_{\infty,2}\frac{\norm{w}_2}{\la} <\frac{1-\tilde{k}_3}{6}, 
	\mbox{ for } k\in\{0,1\},\label{eq-cbp-const-noise}
\end{align}

Then, using the fact that and $|\omega^{(k)}|(t)\leq \norm{(\Phi^{(k)})^*}_{\infty,2}\norm{w}_2$ and  $\stepsizen\leq 1$, we obtain
\begin{align*}
  \sup_{t\in \TT\setminus N(r)} \left(\frac{\omega}{\la} + \mu_0^n + \frac{\stepsizen}{2}\left|(\frac{\omega}{\la} + \mu_0^n)'\right|\right)(t)<1.
\end{align*}
Thus it remains to prove that for each $\nu\in\{1,\ldots, N\}$, 
\begin{align}
  \left(\frac{\omega}{\la} + \mu_0^n + \frac{\stepsizen}{2}(\frac{\omega}{\la} + \mu_0^n)'\right)(t)<1 \mbox{ for } t \in (x_{0,\nu}-r,x_{0,\nu}+r)\setminus \Sright(r),\label{eq-cbp-const-sat1}\\
\qandq \left(\frac{\omega}{\la} + \mu_0^n - \frac{\stepsizen}{2}(\frac{\omega}{\la} + \mu_0^n)'\right)(t)<1 \mbox{ for } t \in (x_{0,\nu}-r,x_{0,\nu}+r)\setminus \Sleft(r)\label{eq-cbp-const-sat2}.
\end{align} 
We only deal with the case  $\Sright(r)=\{x_{0,\nu} \}$, $\Sleft(r)=\{x_{0,\nu},x_{0,\nu}+\stepsizen\}$, the symmetric case being similar. 
Let $f\eqdef \frac{1}{\la}\omega(\cdot-x_{0,\nu})+\mu_0^n(\cdot-x_{0,\nu})$. By definition of $\tilde{\Pi}$, $\omega(x_{0,\nu})=\omega'(x_{0,\nu})=\omega(x_{0,\nu}+\stepsizen)-\frac{\stepsizen}{2}\omega(x_{0,\nu}+\stepsizen)=0$, so that
\eql{\label{eq-cbp-const-fval}
f(0)=1, \quad f'(0)=1, \qandq f(\stepsizen)-\frac{\stepsizen}{2}f'(\stepsizen)=1.
}
Moreover, from Eq.~\eqref{eq-cbp-const-muO} to~\eqref{eq-cbp-const-noise}, and letting $k_1=\frac{\tilde{k}_1}{8}$, $k_2=\frac{\tilde{k}_2}{4}$, we deduce that 
\eql{\label{eq-cbp-const-fsec}
 \forall t\in(-r,r),\quad f''(t)+\frac{\stepsize}{2}|f^{(3)}(t)| <-k_1<0, \qandq f^{(4)}(t)>k_2>0,
}
so that the strict concavity of $f-\frac{\stepsizen}{2} f'$ implies that $(f-\frac{\stepsizen}{2} f')(t)<1$ for $t\in (-r,-\stepsizen)\cup (0,r)$.

It remains to prove that  $(f+\frac{\stepsizen}{2} f')(t)<1$ for $t\in (-r,r)\setminus (-\stepsizen,0]$. A Taylor expansion of $f$ and $f'$ yields (writing as usual $\gamma_k=\frac{1}{k!}- \frac{1}{(k-1)!\times 2}$)
\begin{align*}
  \underbrace{1-(f(\stepsize)- \frac{\stepsizen}{2}f(\stepsizen))}_{=0} &= \underbrace{1- f(0)- \frac{\stepsizen}{2}f'(0)}_{=0} -\stepsizen^3\gamma_3 f^{(3)}(0) - \stepsizen^4\gamma_4f^{(4)}(0) +R_1(\stepsizen),\\
  1-(f(-\stepsizen)+\frac{\stepsizen}{2} f'(-\stepsizen))&=  \underbrace{1- f(0)+ \frac{\stepsizen}{2}f'(0)}_{=0} +\stepsizen^3\gamma_3 f^{(3)}(0) - \stepsizen^4\gamma_4f^{(4)}(0) + R_2(\stepsizen).
\end{align*}
Adding both equations we get 
\begin{align}
  1-(f(-\stepsizen)+\frac{\stepsizen}{2} f'(-\stepsizen))&= -2\stepsizen^4\gamma_4f^{(4)}(0)+ (R_1+R_2)(\stepsizen)\\ % -\stepsize^5\int_0^1 (f^{(5)}(s\stepsize)-f^{(5)}(-s\stepsize))\left(\frac{(1-s)^4}{4!} -\frac{(1-s)^3}{2\times 3!}\right) \d s.
  \label{eq-minoration-dl} &\geq 2(-\gamma_4) k_2\stepsizen^4 + (R_1+R_2)(\stepsizen), 
\end{align}
where
\begin{align}R_1+R_2(h) = \stepsizen^5\int_0^1 (-f^{(5)}(s\stepsizen)+f^{(5)}(-s\stepsizen))\left(\frac{(1-s)^4}{4!} -\frac{(1-s)^3}{2\times 3!}\right) \d s,\label{eq-rest-integral}
\end{align}
\begin{align}  
	\mbox{with} \quad \norm{f^{(5)}}_{\infty} \leq \norm{\frac{1}{\la}\omega^{(5)}}_\infty+\norm{(\mu_0^n)^{(5)}}_\infty=O(1).
\end{align}
Hence, 
\begin{align}
  1-(f(-\stepsizen)+\frac{\stepsizen}{2} f'(-\stepsizen))&\geq \underbrace{2(-\gamma_4) k_2}_{>0}\stepsizen^4 + O(\stepsizen^5)>0.
\end{align}
Moreover,by the strict concavity of $f+\frac{\stepsizen}{2} f'$, we also deduce that $(f+\frac{\stepsizen}{2} f')(t)<1$ for $t\in (-r,-\stepsizen]\cup (0,r)$, thus we get the local inequalities~\eqref{eq-cbp-const-sat1} and~\eqref{eq-cbp-const-sat2}, hence the global inequalities~\eqref{eq-tight-snr-cbpU} and~\eqref{eq-tight-snr-cbpD}.

To conclude, the constants in the condition on $\frac{\norm{w}_2}{\la}$ are $O(1)$, and gathering the asympotics for $c_{1,n},c_{2,n},c_{3,n}$ we obtain $C^{(1)}_{n}=O(\stepsizen^3)$, $C^{(2)}_n= O(1)$.
\end{proof}

% !TEX root = ../continuous-bp.tex
%%%%%%%%%%%%%%%%%%%%%%%%%%%%%%%%%%%%%%%%%%%%%%%%%%
\section{Numerical illustrations}
\label{sec-numerics}

In this section, we illustrate the usefulness of our analysis to gain a precise understanding of the recovery performance of $\ell^1$-type methods (\lasso and C-BP) for both deconvolution and compressed sensing problems. The code to reproduce these numerical experiments is available online\footnote{\url{https://github.com/gpeyre/2015-IP-lasso-cbp/}}.

\subsection{Convergence of pre-certificates}

In this section and in Section~\ref{sec-numeric-deconv}, we consider the deconvolution problems in the case where $\phi$ is an ideal filter, i.e. whose Fourier coefficients  
\eq{
	\foralls k \in \ZZ, \quad \hat \phi(k) \eqdef \int_\TT \phi(t) e^{-2\imath\pi k t} \d t
} 
satisfy $\hat\phi(k)=1$ if $k\in \{-f_c,\ldots,f_c\}$ and $\hat\phi(k)=0$ otherwise. This allows us to implement exactly the $\Phi$ operator appearing in the \lasso and C-BP problem since $\Im(\Phi)$ is a finite dimensional space of dimension $Q=2f_c+1$, i.e. it can be represented using a matrix of size $(Q,P)$ when evaluated on a grid of $P$ points. In Figures~\ref{fig-certificates} and ~\ref{fig-homotopy} we used $f_c=10$.

\newcommand{\myfigCertif}[2]{\includegraphics[width=0.46\linewidth]{certificates/certificates-k#1-#2}}

\begin{figure}[ht]
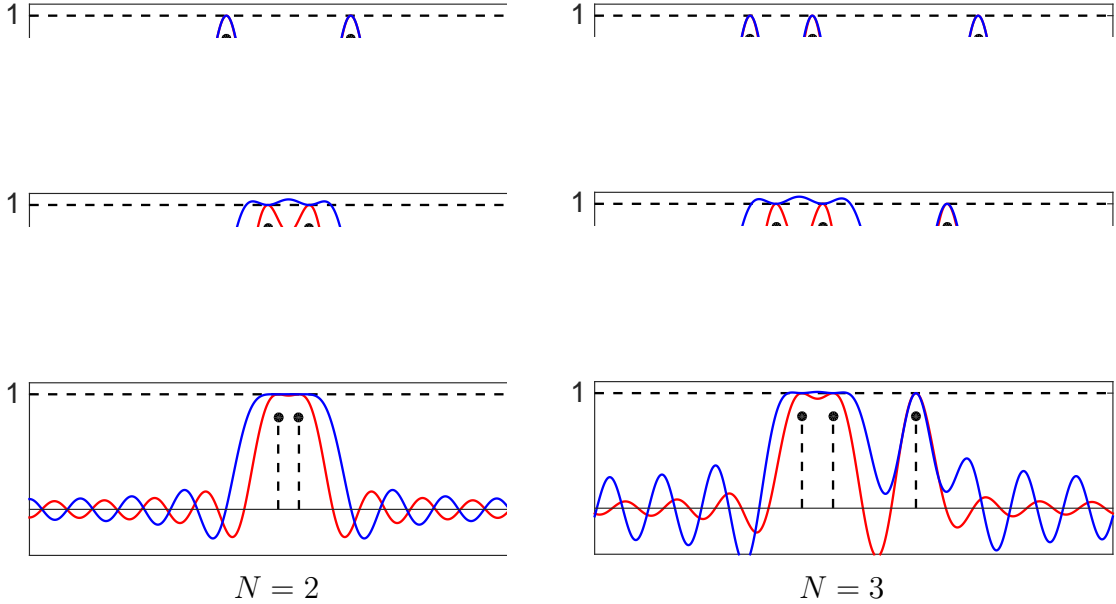

\centering
	\begin{tabular}{@{}c@{\hspace{3mm}}c@{}}
		\myfigCertif{2}{3} & \myfigCertif{3}{2} \\
		\myfigCertif{2}{5} & \myfigCertif{3}{3} \\
		\myfigCertif{2}{6} & \myfigCertif{3}{4} \\
		$N=2$ & $N=3$
	\end{tabular}
\caption{\label{fig-certificates} % 
	Display of $\eta_V^\infty$ (red) and $\mu_T$ (blue) pre-certificate for different input positive measures $m_0$ (showed as black dots to symbolize the position of the Diracs).
	}
\end{figure}

Figure~\ref{fig-certificates} illustrates for the case of two ($N=2$) and three ($N=3$) spikes the behavior of the vanishing pre-certificate $\eta_V^\infty$ (see Definition~\ref{defn-vanishing-bp}) useful to analyze \lasso/\blasso problems and of the pre-certificate $\mu_T$ (see Definition~\ref{defn-third-deriv-precertif}) useful to analyze C-BP problems. 

We first notice that for all the (positive) input measures (i.e. whatever the spacing between the Diracs), $\eta_V^\infty$ is always a non-degenerate certificate (in the sense of Proposition~\ref{prop-etav-nonvanish}), meaning that one actually has $\eta_V^\infty=\eta_0^\infty$ (where the minimal norm certificate $\eta_0^\infty$ is defined in~\eqref{eq-certbeurl0}). This empirical finding is the subject of another recent work on the asymptotic of sparse recovery of positive measures when the spacing between the Diracs tends to zero~\cite{denoyelle2015asymptic}. Since $\eta_0^\infty$ is non-degenerate, one can thus apply Theorem~\ref{thm-lasso-extended} to analyze the extended support of the \lasso (see below Section~\ref{sec-numeric-deconv} for a numerical illustration).

For the C-BP problem, the situation is however more contrasted. We observe that when the Dirac masses are separated enough (first row) then the pre-certificate $\mu_T$ is a valid certificate, meaning the the Twice Non-Degenerate Source Condition (see Definition~\ref{defn-TNDSC}) holds. This means that Theorem~\ref{thm-cbpasso-extended} can be applied to analyze the extended support of C-BP (see Section~\ref{sec-numeric-deconv} below for a numerical illustration). But when the Dirac masses are too close (second and third rows), one has $\norm{\mu_T}_\infty>1$, so that one cannot ensure the support stability of the C-BP solution with our result.

%%%%%
\subsection{Extended support for deconvolution}
\label{sec-numeric-deconv}

\newcommand{\myfigHom}[1]{\includegraphics[width=0.46\linewidth]{homotopy/homotopy-#1}}

We still consider the case of an ideal low pass filter. Figure~\ref{fig-homotopy} displays the evolution, as a function of $\la$ (in abscissa) of the solutions $a_\la$ of~\eqref{eq-thin-lasso} and of $(a_\la,b_\la)$ of~\eqref{eq-thin-cbpasso}. We consider here the case of an input measure with two nearby Diracs (displayed as red/blue dots in the upper-left part of the Figure) and when there is no noise, i.e. $w=0$. Each 1-D curve (either plain or dashed) represents the evolution of a single coefficient, e.g. $(a_\la)_i$, for some index $i$ (only non-zero coefficients are displayed).

\begin{figure}[ht]
\centering
	\begin{tabular}{@{}c@{\hspace{3mm}}c@{}}
		\myfigHom{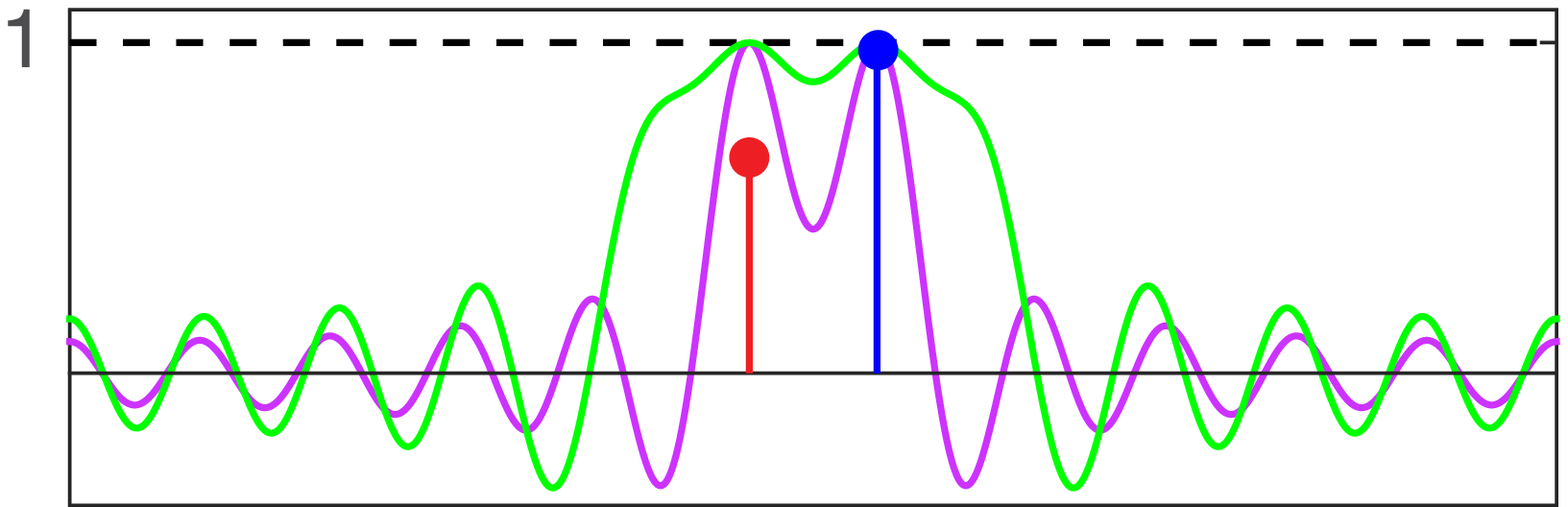} & \myfigHom{k2-d22-n256-bp-A-z0}\\
		Pre-certificates ${\color{magenta} \eta_V}$ and ${\color{green} \mu_T}$ & \lasso, $a_\la$ \\
		\myfigHom{k2-d22-n256-cbp-A-z0} & \myfigHom{k2-d22-n256-cbp-BA-z0} \\[-2mm]
		C-BP, $a_\la$ & C-BP, $\frac{2b_\la}{h a_\la}$ \\
		\myfigHom{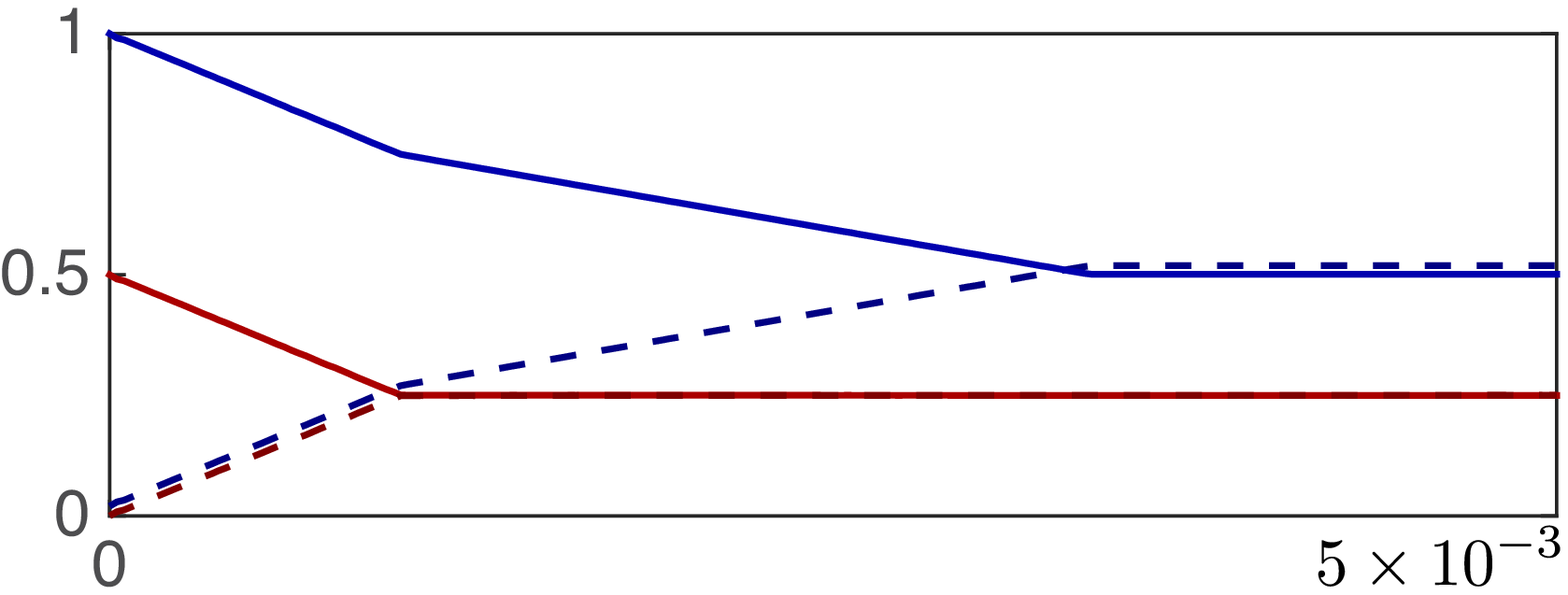} & \myfigHom{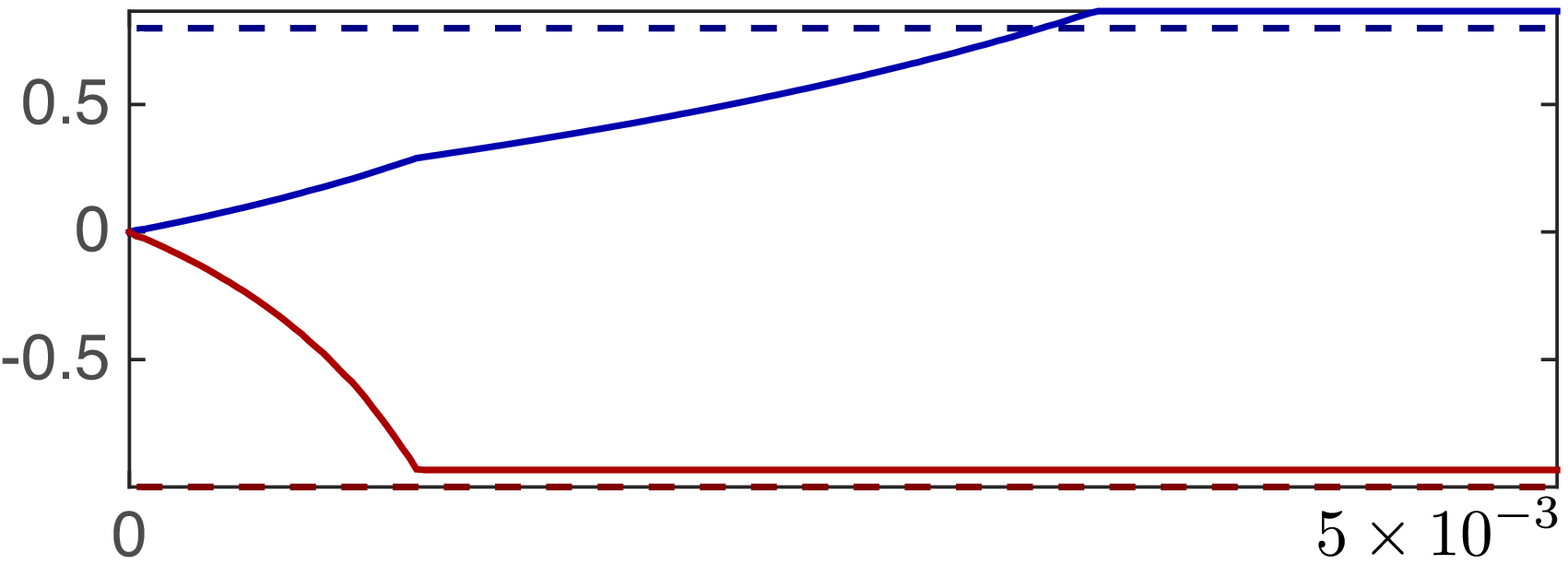} \\[-2mm]
		C-BP, $a_\la$ (zoom) & C-BP, $\frac{2b_\la}{h a_\la}$ (zoom)
	\end{tabular}
\caption{\label{fig-homotopy} % 
	Display of the evolution as a function of $\la$ of the solutions of the \lasso and C-BP problems. 
	Note that dashed curved have been (artificially) slightly shifted to avoid that they overlap with the plain curve. 
	}
\end{figure}

The solutions path $\la \mapsto a_\la$ (for \lasso{}) and $\la \mapsto (a_\la,b_\la)$ (for C-BP) are continuous and piecewise affine, which is to be expected since the regularizations ($\ell^1$ and $\ell^1$ under conic constraints) are polyhedral. The upper-left plot in the figure displays the pre-certificate $\eta_V^\infty$ (in magenta, see Definition~\ref{defn-vanishing-bp}) and $\mu_T$ (in green, see Definition~\ref{defn-third-deriv-precertif}). This shows graphically that these two precertificates are non-degenerate (according to Definitions~\ref{defn-ndsc-bp} and~\ref{defn-TNDSC}) so that the results of Theorems~\ref{thm-lasso-extended} and~\ref{thm-cbpasso-extended} hold, hence precisely describing the evolution of the solution on the extended support when $\la$ is small. On these graphs, this corresponds to the first segment of the corresponding piecewise affine paths. 

The behavior for BP agrees with our analysis. As predicted by Theorem~\ref{thm-lasso-extended}, there exists a range of values $0 < \la < \la_0$ on which the solution is exactly supported on the extended support $J$, which is composed of four spikes (the plain curve corresponds to the support $I$ and the dashed curve corresponds to $J\backslash I$). Also, as predicted by Proposition~\ref{prop-asympto-constant-lasso} in the case $w=0$, we verify that $\la_0 = O(\stepsizen)$ and that the Lipschitz constant of $\la \mapsto a_\la$ is of order $O(1/\stepsizen)$. 

In sharp contrast, the behavior for C-BP is less regular, since the range $0 < \la < \la_0$ on which the solution is supported on the extended support is shorter, as it can be clearly seen on the zoom for very small values of $\la$. This is in agreement with Proposition~\ref{prop-asympto-constant-cbp} which shows that $\la_0$ is of the order of $O(\stepsizen^3)$ and that the Lipschitz constant of $\la \mapsto (a_\la,b_\la)$ is of order $O(1/\stepsizen^3)$.
On this range of small $\la$, as predicted by Theorem~\ref{thm-cbpasso-extended}, the support of the solutions (which correspond to the extended support $J$ described in Theorem~\ref{thm-cbpasso-extended}) is composed of one pair of neighboring spikes for each original spike. For indices on the support $i \in I$, one has $|(b_\la)_i|/(a_\la)_i<h/2$ (the constraint is non-saturating, and the spike moves ``freely'' inside $(i\stepsize-\frac{\stepsize}{2},i\stepsize+\frac{\stepsize}{2})$) while for indices on the extended part $i \in J \backslash I$, one has $|(b_\la)_i|/(a_\la)_i = h/2$ (the constraint is saturating, the spikes are fixed at half-grid points).
Another part of the path is interesting, for $\la$ not so small (say $\la>\la_1$), which is in fact the prominent regime in the non-zoomed figure. For this range of $\la$, there is still a pair of spikes for each original spike, but this time both spikes saturate, on same side. This observation should be related to Proposition~\ref{prop-cbp-thin-supplambda} and Remark~\ref{rem-cbp-nbspikes} which predict that, in the case where ${\mu_\la^\infty}^{(3)}(x_{\la,\nu})\neq 0$, the C-BP yields either one spike or a pair of spikes with the same shift (the latter case is in fact overwhelming).

%%%%%
\subsection{Extended support for compressed sensing}

To show the usefulness of our ``abstract'' support analysis of the \lasso problem (Section~\ref{sec-continuous-abstract}), we illustrate its use to analyze the performance of $\ell^1$ recovery in a compressed sensing setup. Compressed sensing corresponds to the recovery of a high dimensional (but hopefully sparse) vector $a_0 \in \RR^P$ from low resolution, possibly noisy, randomized observations $y=B a_0 +w \in \RR^Q$, see for instance~\cite{CandesWakin} for an overview of the literature on this topic. For simplicity, we assume that there is no noise ($w=0$) and we consider here the case where $B \in \RR^{Q \times P}$ is a realization from the Gaussian matrix ensemble, where the entries are independent and uniformly distributed according to a Gaussian $\Nn(0,1)$ distribution. This setting is particularly well documented, and it has been shown, assuming that $a_0$ is $s$-sparse (meaning that $|\supp(a_0)|=s$), that there are roughly three regimes:
\begin{rs}
	\item If $s < s_0 \eqdef \frac{Q}{2\log(P)}$, then $a_0$ is with ``high probability'' the unique solution of~\eqref{eq-abstract-bp} (it is identifiable), and the support is stable to small noise, because $\eta_F$ (as defined in~\eqref{eq-fuchs-precertif}) is a valid certificate, $\norm{\eta_F}_\infty \leq 1$. This is shown for instance in~\cite{wainwright-sharp-thresh,dossal2011noisy}. 
	\item If $s < s_1 \eqdef \frac{Q}{2\log(P/Q)}$, then $a_0$ is with ``high probability'' the unique solution of~\eqref{eq-abstract-bp}, but the support is not stable, meaning that $\eta_F$ is not a valid certificate. This phenomena is precisely analyzed in~\cite{chandrasekaran2012convex,amelunxen2013living} using tools from random matrix theory and so-called Gaussian width computations. 
	\item If $s > s_1$, then $a_0$ with ``high probability'' is not the solution of~\eqref{eq-abstract-bp}.
\end{rs}
We do not want to give details here on the precise meaning of with ``high probability'', but this can be precisely quantified in term of probability of success (with respect to the random draw of $B$) and one can show that a phase transition occurs, meaning that for large $(P,Q)$ the transition between these regimes is sharp. 

While the regime $s<s_0$ is easy to understand, a precise analysis of the intermediate regime $s_0<s<s_1$ in term of support stability is still lacking. Figure~\ref{fig-cs} shows how Theorem~\ref{thm-stability-dbp} allows us to compute numerically the size of the recovered support, hence providing a quantification of the degree of ``instability'' of the support when a small noise $w$ contaminates the observations. The simulation is done with $(P,Q)=(400,100)$. 

The left cuve shows, as a function of $s$ (in abscissa), the probability (with respect to a random draw of $\Phi$ and $a_0$ a $s$-sparse vector) of the event that $a_0$ is identifiable (plain curve) and of the event that $\eta_F$ is a valid certificate (dashed curve). This clearly highlights the phase transition phenomena between the three different regimes, and one roughly gets that $s_0 \approx 6$ and $s_1 \approx 20$, which is consistent with the theoretical asymptotic bounds found in the literature.

The right part of the figure, shows, for three different sparsity levels $s \in \{14,16,18\}$, the histogram of the repartition of $|J|$ where $J$ is the extended support, as defined in Theorem~\ref{thm-stability-dbp}. According to Theorem~\ref{thm-stability-dbp}, this histogram thus shows the repartition of the sizes of the supports of the solutions to~\eqref{eq-abstract-lasso} when the noise $w$ contaminating the observations $y=B a_0+w$ is small and $\la$ is chosen in accordance to the noise level. As one could expect, this histogram is more and more concentrated around the minimum possible value $s$ (since we are in the regime $s<s_1$ so that the support $I$ of size $s$ is included in the extended support $J$) as $s$ approaches $s_0$ (for smaller values, the histogram being only concentrated at $s$ since $J=I$ and the support is stable). Analyzing theoretically this numerical observation is an interesting avenue for future work that would help to better understand the performance of compressed sensing. 

\newcommand{\myfigCS}[1]{\includegraphics[width=0.46\linewidth]{compressed-sensing/distrib-extsup-s#1}}

\begin{figure}[ht]
\centering
	\begin{tabular}{@{}c@{\hspace{3mm}}c@{}}
		\multirow{6}{*}{\includegraphics[width=0.5\linewidth]{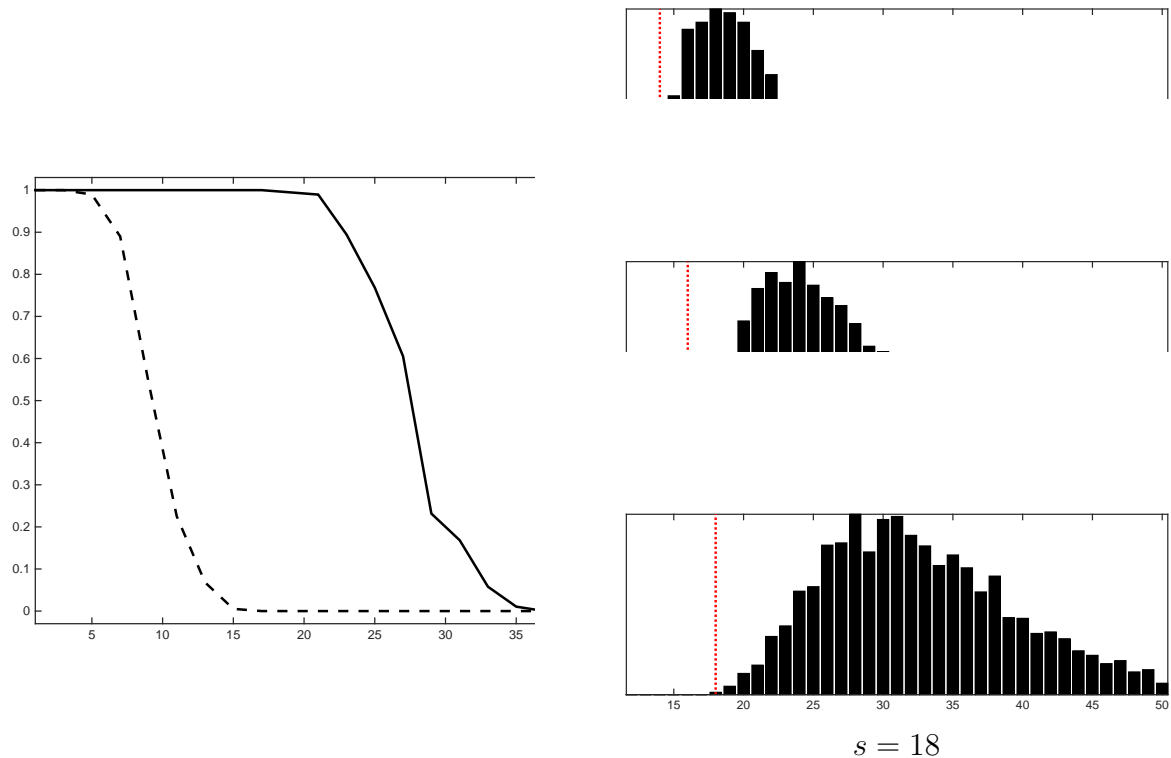}}
		 & \myfigCS{14} \\
		 & $s=14$ \\
		 & \myfigCS{16} \\
		 & $s=16$ \\
		 & \myfigCS{18} \\
		 & $s=18$ 
	\end{tabular}
\caption{\label{fig-cs} % 
	\textit{Left:} probability as a function of $s$ of the event that $a_0$ is identifiable (plain curve) and of the even that its support is stable (dashed curve).
	\textit{Right:} for several value of $s$, display of histogram of repartition of the sizes $|J|$ of the extended support $J$. 
	}
\end{figure}

% !TEX root = ../continuous-bp.tex
%%%%%%%%%%%%%%%%%%%%%%%%%%%%%%%%%%%%%%%%%%%%%%%%%%%
\section*{Conclusion}

In this work, we have provided a precise analysis of the properties of the solution path of $\ell^1$-type variational problems in the low-noise regime. This includes in particular the \lasso and the C-BP problems. A particular attention has been paid to the support set of this path, which in general cannot be expected to match the one of the sought after solution. Two striking examples support the relevance of this approach. For the deconvolution problem, we showed theoretically that in general this support is not stable, and we were able to derive in closed form the solution of the ``extended support'' that is twice larger, but is stable. In the compressed sensing scenario (i.e. when the operator of the inverse problem is random), we showed numerically how to leverage our theoretical findings and analyze the growth of the extended support size as the number of measurements diminishes. This analysis opens the doors for many new developments to better understand this extended support, both for deterministic operators (e.g. Radon transform in medical imaging) and random ones.

%%%%%%%%%%%%%%%%%%%%%%%%%%%%%%%%%%%%%%%%%%%%%%%%%%%
\section*{Acknowledgements} 

We would like to thank Charles Dossal, Jalal Fadili and Samuel Vaiter for stimulating discussions on the notion of extended support. This work has been supported by the European Research Council (ERC project SIGMA-Vision).

\appendix
% !TEX root = ../continuous-bp.tex

\section{Useful properties of the integral transform}

\begin{lem}\label{lem-phi-compact}
  Let $K\in \NN^*$ and assume that $\phi\in \Cont^{K}(\TT\times \TT)$. Then for all $p\in L^2(\TT)$, $\Phi^*p\in \Cont^{K}(\TT)$ and for all $k\in\{0,1\ldots K\}$
  \begin{align}
    \forall y\in \TT,\  (\Phi^*p)^{(k)}(y)&= \int_\TT (\partial_2)^k\varphi(x,y)p(x)dx.
  %  \mbox{or equivalently} \ (\Phi^*p)^{(k)}&= \Phi^{(k),*}p \qwhereq \Phi^{(k)}: m\mapsto \int_\TT \varphi(\cdot,y)dm(y).
  \end{align}
  Moreover, the adjoint operator 
  \begin{align}
    \Phi^{(k),*} : \begin{array}{ccc} L^2(\TT) &\longrightarrow & C(\TT)\\
      p&\longmapsto & \int_\TT (\partial_2)^k\varphi(x,y)p(x)dx \end{array}     ,
  \end{align}
  is compact for all $k\in\{0,1,\ldots K\}$.
\end{lem}

\begin{proof}
  The first part of the lemma is a standard application of the Lebesgue dominated convergence theorem. 

  As for the second part, it is a consequence of the Ascoli-Arzela theorem. Indeed, let $B_{L^2}(0,1)=\{p\in L^2(\TT); \ \|p\|_2\leq 1 \}$, and $\Ff=\{\Phi^*p; \ p\in B_{L^2}(0,1) \}$. Then $\Ff\subset C(\TT)$ is bounded since 
  \eq{
  	|\Phi^{(k),*} p(x)| \leq \sqrt{\int_\TT\left( (\partial_2)^k\varphi(x,y)\right)^2dx}\sqrt{\int_\TT (p(x))^2dx}
	\leq \sqrt{\| (\partial_2)^k\varphi\|_{\infty}}
	}
Moreover, it is equicontinuous since
\begin{align}
  |\Phi^{(k),*}p(x)-\Phi^{(k),*}p(x')| &=\left|\int_\TT ((\partial_2)^k\varphi(x,y)-(\partial_2)^k\varphi(x',y))p(x) \right|\\
  &\leq \sqrt{\omega_{(\partial_2)^k\varphi}(|x-x'|,0)},
\end{align}
where $\omega_{(\partial_2)^k\varphi}$ is the modulus of continuity of $(\partial_2)^k\varphi$. Thus Ascoli-Arzela's theorem ensures that $\Phi^{(k),*}B_{L^2}(0,1)$ is relatively compact, hence the result.
\end{proof}

An interesting consequence of the above lemma is the following. Given any bounded sequence $\{p_n\}_{n\in\NN}$ in $L^2(\TT)$, we may extract a subsequence $\{p_{n'}\}_{n'\in\NN}$ which converges weakly towards some $\tilde{p}\in L^2(\TT)$. Then, the (sub)sequence $\Phi^* p_{n'}$ converges towards $\Phi^*\tilde{p}$ for the (strong) uniform topology, and its derivatives $\Phi^{(k),*}p_{n'}$ also converge towards $\Phi^{(k),*}\tilde{p}$ for that topology.

%%%%%%%%%%%%%%%%%%%%%%%%%%%%%%%%%%%%%%%%%%%%%%%%%%%%%%%%%%%%%%%%%%%%%%%
\section{Asymptotic expansion of the inverse of a Gram matrix}

In this Appendix, we gather some useful lemmas on the asymptotic behavior of inverse Gram matrices.

\begin{lem}
  Let $A\colon \RR^N\rightarrow L^2(\TT)$, $B:\RR^N\rightarrow\RR^N$ be linear operators such that $A$ has full rank and $B$ is invertible. Then $(AB)^+=B^{-1}A^+$.
  \label{lem-apx-inverse}
\end{lem}
\begin{proof}
  It is sufficient to write 
  \eq{ 
  	\left((AB)^*(AB)\right)^{-1}(AB)^*)= B^{-1}(A^*A)^{-1}B^{-1,*}B^*A^*=B^{-1}A^+.
}
\end{proof}

\begin{lem}\label{lem-apx-asymptolasso}
  Let $A,B,B_h\colon \RR^N\rightarrow L^2(\TT)$ be linear operators such that $B_h=B +O(h)$ for $h>0$, and that $\begin{pmatrix}
    A & B
  \end{pmatrix}$ has full rank. Let $\Pi$ be the orthogonal projector onto $(\Im A)^\perp$, and let 
  \eq{
  	G_h\eqdef \begin{pmatrix}
    A^*\\
    A^*+hB_h^*
  \end{pmatrix}\begin{pmatrix}
    A& A+hB_h
  \end{pmatrix}
  } 
  and $s\in\RR^N$. Then for $h>0$ small enough, $G_h$ and $B^*\Pi B$ are invertible, and
  \begin{align}
G_h^{-1}\begin{pmatrix}
    s\\s
  \end{pmatrix}&=\frac{1}{h} \begin{pmatrix} (B^*\Pi B)^{-1}B^*A^{+,*}s\\-(B^*\Pi B)^{-1}B^*A^{+,*}s \end{pmatrix}+O(1),\\
  \begin{pmatrix}
    A& A+hB_h
  \end{pmatrix}^+&= \frac{1}{h}\begin{pmatrix}
  (B^*\Pi B)^{-1}B^*\Pi \\ -(B^*\Pi B)^{-1}B^*\Pi
\end{pmatrix}+O(1),\\
\mbox{but }\begin{pmatrix}
    A& A+hB_h
  \end{pmatrix}^{+,*}\begin{pmatrix}
    s\\s
  \end{pmatrix}
&= A^{+,*}s -\Pi B(B^*\Pi B)^{-1}B^*A^{+,*}s + O(h).
\end{align}
\end{lem}

\begin{proof}
  Observe that $\begin{pmatrix}
    A& A+hB_h
  \end{pmatrix}=\begin{pmatrix}
    A& B_h
  \end{pmatrix}\begin{pmatrix}
    I_N & I_N\\ 0 & hI_N
  \end{pmatrix} $ 
  so that  
  \eq{
  	G_h=\begin{pmatrix}I_N & 0\\ I_N & hI_N\end{pmatrix}
  \begin{pmatrix} A^*A & A^*B_h\\B_h^*A & B_h^*B_h\end{pmatrix}
  \begin{pmatrix}I_N & I_N\\ 0 & hI_N\end{pmatrix}.
  } 
  Since $\begin{pmatrix}A & B\end{pmatrix}$
  has full rank, the middle matrix is invertible for $h$ small enough, and 
  \eq{
  	G_h^{-1}=\begin{pmatrix}I_N & -\frac{1}{h}I_N\\0 & \frac{1}{h}I_N\end{pmatrix}
  \begin{pmatrix} A^*A & A^*B_h\\B_h^*A & B_h^*B_h\end{pmatrix}^{-1}
  \begin{pmatrix}I_N & 0\\ -\frac{1}{h}I_N & \frac{1}{h} I_N\end{pmatrix}.
 }
  
  %\\
%  &= \begin{pmatrix}I_N & -\frac{1}{h}I_N\\0 & \frac{1}{h}I_N\end{pmatrix}
%  \left(\begin{pmatrix} A^*A & A^*B\\B^*A & B^*B\end{pmatrix}^{-1} + O(h)\right)
%  \begin{pmatrix}I_N & 0\\ -\frac{1}{h}I_N & \frac{1}{h} I_N\end{pmatrix}.
Writing  $\begin{pmatrix} a & b\\ c & d\end{pmatrix}\eqdef\begin{pmatrix} A^*A & A^*B_h\\B_h^*A & B_h^*B_h\end{pmatrix}$, the block inversion formula yields
%\begin{align*}
\eq{
	\begin{pmatrix} a & b\\ c & d\end{pmatrix}^{-1}=  \begin{pmatrix}
    a^{-1} + a^{-1}bS^{-1}ca^{-1} & -a^{-1}bS^{-1}\\
    -S^{-1}ca^{-1} & S^{-1}\end{pmatrix},
}
\eq{
	\qwhereq 
	S\eqdef d-ca^{-1}b = B_h^*B_h-B_h^*A(A^*A)^{-1}A^*B_h = B_h^*\Pi B_h
}
is indeed invertible for small $h$ since $\begin{pmatrix} A & B \end{pmatrix}$ has full rank. Moreover, $a^{-1}bS^{-1}= A^+B_h(B_h^*\Pi B_h)^{-1}$, and $S^{-1}ca^{-1}= (B_h^*\Pi B_h)^{-1}B_h^*A^{+,*}$.

Now, we evaluate $G_h^{-1}\begin{pmatrix} s\\s\end{pmatrix}=\begin{pmatrix}I_N & -\frac{1}{h}I_N\\0 & \frac{1}{h}I_N\end{pmatrix}
 \begin{pmatrix}
   a^{-1}s +  a^{-1}bS^{-1}ca^{-1}s\\
  -S^{-1}ca^{-1}s 
 \end{pmatrix}$. We obtain 
 \begin{align*}G_h^{-1}\begin{pmatrix} s\\s\end{pmatrix}= \frac{1}{h}  \begin{pmatrix}
   S^{-1}ca^{-1}s \\-S^{-1}ca^{-1}s 
 \end{pmatrix}+O(1) = \frac{1}{h} \begin{pmatrix} (B^*\Pi B)^{-1}B^*A^{+,*}s\\-(B^*\Pi B)^{-1}B^*A^{+,*}s \end{pmatrix}+O(1).
 \end{align*}

 Eventually, by Lemma~\ref{lem-apx-inverse}, $\begin{pmatrix}
    A& A+hB_h
  \end{pmatrix}^+= \begin{pmatrix}I_N & -\frac{1}{h}I_N\\0 & \frac{1}{h}I_N\end{pmatrix}
  \begin{pmatrix} A^*A & A^*B_h\\B_h^*A & B_h^*B_h\end{pmatrix}^{-1}
\begin{pmatrix}A^*\\B_h^*\end{pmatrix}$.
We obtain 
\begin{align*}\begin{pmatrix}
    A& A+hB_h
  \end{pmatrix}^+&= \begin{pmatrix}I_N & -\frac{1}{h}I_N\\0 & \frac{1}{h}I_N\end{pmatrix}
\begin{pmatrix}
  A^+ -A^+B_h(B_h^*\Pi B_h)^{-1}B_h^*\Pi\\
  -(B_h^*\Pi B_h)^{-1}B_h^*\Pi
\end{pmatrix}\\
\intertext{and we deduce}
\begin{pmatrix}
    A& A+hB_h
  \end{pmatrix}^+&= \frac{1}{h}\begin{pmatrix}
  (B^*\Pi B)^{-1}B^*\Pi \\ -(B^*\Pi B)^{-1}B^*\Pi
\end{pmatrix}+O(1),\\
\mbox{and }
\begin{pmatrix}
    A& A+hB_h
  \end{pmatrix}^{+,*}\begin{pmatrix}
    s\\s
  \end{pmatrix}&= \begin{pmatrix}
    A^{+,*}-\Pi B_h(B_h\Pi B_h)^{-1}B_h^*A^{+,*} & \Pi B_h(B_h^*\Pi B_h)^{-1}
  \end{pmatrix}\\
&\qquad \qquad \qquad\begin{pmatrix}I_N &0 \\ -\frac{1}{h}I_N & \frac{1}{h}I_N\end{pmatrix}
\begin{pmatrix}s\\s\end{pmatrix}\\
&= A^{+,*}s -\Pi B(B^*\Pi B)^{-1}B^*A^{+,*}s + O(h).
\end{align*}
\end{proof}

\begin{lem}\label{lem-apx-cbpdl}
Let $A,B,C,C_h\colon \RR^N\rightarrow L^2(\TT)$ be linear operators such that $C_h=C +o(1)$ for $h>0$, and that $\begin{pmatrix}
    A & B & C
  \end{pmatrix}$ has full rank. Let $\tilde{\Pi}$ be the orthogonal projector onto $(\Im \begin{pmatrix}A&B\end{pmatrix})^\perp$, and let 
  \eq{
  	G_h\eqdef \begin{pmatrix}
    (A+\frac{h}{2}B)^*\\
    (A-\frac{h}{2}B)^*\\
    (A+\frac{h}{2}B+h^3 C_h)^*
  \end{pmatrix}\begin{pmatrix}
   A+\frac{h}{2}B & A-\frac{h}{2}B & A+\frac{h}{2}B+ h^3 C_h
  \end{pmatrix}.
  } 
  Then for $h>0$ small enough, $G_h$ and $C^*\tilde{\Pi} C$ are invertible, and
\begin{align*}
 	G_h^{-1}\begin{pmatrix}\bun_N\\\bun_N\\\bun_N\end{pmatrix} 
	&= -\frac{1}{h^3}\begin{pmatrix} -I_N\\0\\I_N\end{pmatrix} (C^*\tilde{\Pi} C)^{-1}C^*\begin{pmatrix} A& B\end{pmatrix}^{+,*} \begin{pmatrix} \bun_N\\0\end{pmatrix} +o\left(\frac{1}{h^3}\right)
\end{align*}
\begin{align*}	
\begin{pmatrix}
 A+\frac{h}{2}B & A-\frac{h}{2}B & A+\frac{h}{2}B+h^3 C_h
\end{pmatrix}^+&= \frac{1}{h^3}  \begin{pmatrix}
    -(C^*\tilde{\Pi} C)^{-1}C^*\tilde{\Pi} \\ 0 \\(C^*\tilde{\Pi} C)^{-1}C^*\tilde{\Pi} 
    \end{pmatrix}+ o\left(\frac{1}{h^3}\right),
  \end{align*}
  but
\begin{align*}
    \begin{pmatrix}A^*+\frac{h}{2}B^* \\ A^*-\frac{h}{2}B^* \\ A^*+\frac{h}{2}B^*+h^3 C_h^*\end{pmatrix}^+
    \begin{pmatrix}\bun_N\\\bun_N\\\bun_N\end{pmatrix}
    = \begin{pmatrix}A^*\\ B^*\end{pmatrix}^{+}\begin{pmatrix}\bun_N\\0\end{pmatrix}-\tilde{\Pi}C(C^*\tilde{\Pi}C)^{-1}C^*\begin{pmatrix}A^*\\ B^*\end{pmatrix}^{+}\begin{pmatrix}\bun_N\\0\end{pmatrix}\\
    \qquad\qquad\qquad+o(1).
\end{align*}
\end{lem}

\begin{proof}
    Observe that 
\begin{align*}
\begin{pmatrix}
 A+\frac{h}{2}B & \!A-\frac{h}{2}B & \!A+\frac{h}{2}B+ h^3 C_h
\end{pmatrix}
	\!=\!
	\begin{pmatrix} A\! & B\! & C_h  \end{pmatrix}
\diag\left(1, \frac{h}{2}, h^3 \right)
  \begin{pmatrix}I_N & I_N & I_N\\ I_N &-I_N & I_N\\ 0 & 0 &I_N  \end{pmatrix}
\end{align*}

As a result, for $h>0$ small enough $G_h$ is invertible and 
\begin{align*}
  G_h^{-1}&=\begin{pmatrix}\frac{1}{2} I_N & \frac{1}{2} I_N & -I_N\\ \frac{1}{2}I_N &-\frac{1}{2} I_N & 0\\ 0 & 0 &I_N  \end{pmatrix}
  \diag\left(1,\frac{2}{h},\frac{1}{h^3}\right)
  \begin{pmatrix} A^*A & A^*B & A^*C_h\\ B^*A & B^*B& B^*C_h\\ C_h^*A & C_h^*B & C_h^*C_h\end{pmatrix}^{-1}\\
 &\qquad \qquad\times\diag\left(1,\frac{2}{h},\frac{1}{h^3}\right)
  \begin{pmatrix}\frac{1}{2} I_N & \frac{1}{2} I_N & 0\\ \frac{1}{2}I_N &-\frac{1}{2} I_N & 0\\ -I_N & 0 &I_N  \end{pmatrix}
\end{align*}
the middle matrix being invertible from the full rank assumption on $\begin{pmatrix}A & B & C\end{pmatrix}$. Moreover, 
 writing  $\tilde{\Gamma}\eqdef\begin{pmatrix}A & B\end{pmatrix}$ and $\begin{pmatrix} a & b\\ c & d\end{pmatrix}\eqdef \begin{pmatrix} \tilde{\Gamma}^*\tilde{\Gamma} & \tilde{\Gamma}^*C_h \\ C_h^*\tilde{\Gamma} & C_h^*C_h\end{pmatrix}$,  we obtain 
\begin{align*}
  \begin{pmatrix} A^*A & A^*B & A^*C_h\\ B^*A & B^*B& B^*C_h\\ C_h^*A & C_h^*B & C_h^*C_h\end{pmatrix}^{-1}
 &= \begin{pmatrix}
    u & -a^{-1}bS^{-1}\\
    -S^{-1}ca^{-1} & S^{-1}\end{pmatrix}
\end{align*}
where $u \eqdef a^{-1} + a^{-1}bS^{-1}ca^{-1}$, $S \eqdef d-ca^{-1}b=C_h^*\tilde{\Pi} C_h$, $a^{-1}bS^{-1} = (\tilde{\Gamma}^*\tilde{\Gamma})^{-1} \tilde{\Gamma}^*C_h(C_h^*\tilde{\Pi} C_h)^{-1}$, $S^{-1}ca^{-1}=(C_h^*\tilde{\Pi} C_h)^{-1} C_h^*\tilde{\Gamma}(\tilde{\Gamma}^*\tilde{\Gamma})^{-1}$, and $\tilde{\Pi}$ is the orthogonal projector onto $(\Im \tilde{\Gamma})^\perp$.
Thus 
\begin{align*}
  G_h^{-1}\begin{pmatrix}\bun_N\\\bun_N\\\bun_N\end{pmatrix} &= \begin{pmatrix}\frac{1}{2} I_N & \frac{1}{2} I_N & -I_N\\ \frac{1}{2}I_N &-\frac{1}{2} I_N & 0\\ 0 & 0 &I_N  \end{pmatrix}
\diag\left(1,\frac{2}{h},\frac{1}{h^3}\right)
\begin{pmatrix}
    u & -a^{-1}bS^{-1}\\
    -S^{-1}ca^{-1} & S^{-1}\end{pmatrix}
\begin{pmatrix} \bun_N\\0\\0\end{pmatrix}\\
&=\frac{1}{h^3} \begin{pmatrix}0 & 0 & -I_N\\ 0&0& 0\\ 0 & 0 &I_N  \end{pmatrix}
\begin{pmatrix}
   u \begin{pmatrix} \bun_N\\0\end{pmatrix} \\ 
   -S^{-1}ca^{-1}\begin{pmatrix} \bun_N\\0\end{pmatrix}
\end{pmatrix}+o\left(\frac{1}{h^3}\right) \\
&= -\frac{1}{h^3}\begin{pmatrix} -I_N\\0\\I_N\end{pmatrix} (C^*\tilde{\Pi} C)^{-1}C^*\tilde{\Gamma}^{+,*} \begin{pmatrix} \bun_N\\0\end{pmatrix} +o\left(\frac{1}{h^3}\right).
\end{align*}

Eventually, one has 
\begin{align*}
  &\begin{pmatrix}
 A+\frac{h}{2}B & A-\frac{h}{2}B & A+\frac{h}{2}B+ h^3 C_h
\end{pmatrix}^+\\
&=\begin{pmatrix}\frac{1}{2} I_N & \frac{1}{2} I_N & -I_N\\ \frac{1}{2}I_N &-\frac{1}{2} I_N & 0\\ 0 & 0 &I_N  \end{pmatrix}
  \diag(1,\frac{2}{h},\frac{1}{h^3})
\begin{pmatrix}
  \tilde{\Gamma} & C_h
\end{pmatrix}^+ \\
&=\frac{1}{h^3} \begin{pmatrix}0 & 0 & -I_N\\ 0&0& 0\\ 0 & 0 &I_N  \end{pmatrix}\begin{pmatrix}
    u & -a^{-1}bS^{-1}\\
    -S^{-1}ca^{-1} & S^{-1}\end{pmatrix} \begin{pmatrix} \tilde{\Gamma}^*& C_h^*\end{pmatrix}+o\left(\frac{1}{h^3}\right)\\
  &= \frac{1}{h^3}  \begin{pmatrix}
    -(C^*\tilde{\Pi} C)^{-1}C^*\tilde{\Pi} \\ 0 \\ (C^*\tilde{\Pi} C)^{-1}C^*\tilde{\Pi} 
    \end{pmatrix}+ o\left(\frac{1}{h^3}\right),\\
    \intertext{and }
&\begin{pmatrix}
 A+\frac{h}{2}B & A-\frac{h}{2}B & A+\frac{h}{2}B+ h^3 C_h
\end{pmatrix}^{+,*} \begin{pmatrix}\bun_N\\\bun_N\\\bun_N\end{pmatrix}\\
&=\begin{pmatrix}\tilde{\Gamma}^* \\ C_h^*\end{pmatrix}^+
  \diag(1,\frac{2}{h},\frac{1}{h^3})
\begin{pmatrix}\frac{1}{2} I_N & \frac{1}{2} I_N & 0\\ \frac{1}{2}I_N &-\frac{1}{2} I_N & -I_N\\ 0 & 0 &I_N\end{pmatrix}
\begin{pmatrix}\bun_N\\\bun_N\\\bun_N\end{pmatrix}\\
&= \left[\tilde{\Gamma}^{+,*}-\tilde{\Pi}C_h(C_h^*\tilde{\Pi}C_h)^{-1}C_h^*\tilde{\Gamma}^{+,*}\right]\begin{pmatrix}\bun_N\\0\end{pmatrix}\\
&= \tilde{\Gamma}^{+,*}\begin{pmatrix}\bun_N\\0\end{pmatrix}-\tilde{\Pi}C(C^*\tilde{\Pi}C)^{-1}C^*\tilde{\Gamma}^{+,*}\begin{pmatrix}\bun_N\\0\end{pmatrix}+o(1)
\end{align*}

%  &= \begin{pmatrix} A^*A & A^*B & A^*C\\ B^*A & B^*B& B^*C\\ C^*A & C^*B & C^*C\end{pmatrix}^{-1}+O(1)
% where

\end{proof}

%%%
%\bibliographystyle{plain}
%\bibliography{bibliography}
\printbibliography

\end{document}